\documentclass[10pt]{article}

\usepackage{bm}
\usepackage{fullpage}  

\usepackage{url}
\usepackage{algorithm, algorithmicx,algpseudocode}
\usepackage{multirow} 
\usepackage{hhline}  
\usepackage{verbatim}

\usepackage{graphicx}
\usepackage{amsmath,amssymb,amsfonts,graphicx,amsthm,mathtools,nicefrac}
\usepackage{lscape}
\usepackage{color}
\usepackage{authblk}
\usepackage{subfigure}
\usepackage{rotating}

\allowdisplaybreaks
\newcommand\numberthis{\addtocounter{equation}{1}\tag{\theequation}}

\newtheorem{theorem}{Theorem}[section]
\newtheorem{assumption}{Assumption}[section]
\newtheorem{remark}{Remark}[section]

\numberwithin{equation}{section}

\newtheorem{lemma}[theorem]{Lemma}
\newtheorem{corollary}[theorem]{Corollary}

\DeclareMathOperator{\rank}{\mathrm{rank}}
\DeclareMathOperator{\diag}{\mathrm{diag}}

\DeclareMathOperator{\trace}{trace}

\DeclareMathOperator{\Range}{Range}

\DeclareMathOperator{\vect}{vec}
\DeclareMathOperator{\conv}{conv}
\def\la{\left\langle}
\def\ra{\right\rangle}
\def\lb{\left(}
\def\rb{\right)}
\def\lcb{\left\{}
\def\rcb{\right\}}

\def\lsb{\left[}
\def\rsb{\right]}
\def\lab{\left|}
\def\rab{\right|}

\newcommand{\vecnorm}[2]{\left\| #1\right\|_{#2}}

\newcommand{\matsnorm}[2]{\left\| #1\right\|_{{#2}}}

\newcommand{\nucnorm}[1]{\ensuremath{\matsnorm{#1}{\footnotesize{\mbox{$\ast$}}}}}
\newcommand{\fronorm}[1]{\ensuremath{\matsnorm{#1}{\footnotesize{\mathsf{F}}}}}
\newcommand{\opnorm}[1]{\ensuremath{\matsnorm{#1}{}}}

\newcommand{\ginfnorm}[1]{\ensuremath{\matsnorm{#1}{\footnotesize{\mbox{$\calG$,$\infty$}}}}}
\newcommand{\gfronorm}[1]{\ensuremath{\matsnorm{#1}{\footnotesize{\mbox{$\calG$,$\mathsf{F}$}}}}}
\newcommand{\bfm}[1]{\bm{#1}}

\newcommand{\E}[2][]{\mathbb{E}_{#1} \left[ #2 \rule{0mm}{3mm}\right]}

\def\va{\bfm a}   \def\mA{\bfm A}  
\def\vb{\bfm b}   \def\mB{\bfm B}  
     \def\C{\mathbb{C}}
   \def\mD{\bfm D}  
\def\ve{\bfm e}   \def\mE{\bfm E}  
  \def\mF{\bfm F}  
\def\vg{\bfm g}   \def\mG{\bfm G}  
\def\vh{\bfm h}   \def\mH{\bfm H}  
   \def\mI{\bfm I}

   \def\mL{\bfm L}  
   \def\mM{\bfm M}

\def\vp{\bfm p}   \def\mP{\bfm P}  \def\P{\mathbb{P}}
   \def\mQ{\bfm Q}  
   \def\mR{\bfm R}  \def\R{\mathbb{R}}
   \def\mS{\bfm S}  
     
   \def\mU{\bfm U}  
   \def\mV{\bfm V}  
   \def\mW{\bfm W}  
\def\vx{\bfm x}   \def\mX{\bfm X}  
\def\vy{\bfm y}   \def\mY{\bfm Y}  
\def\vz{\bfm z}   \def\mZ{\bfm Z}

\def\calA{{\cal  A}} 
\def\calB{{\cal  B}} 
 
\def\calD{{\cal  D}}

\def\calG{{\cal  G}} 
\def\calH{{\cal  H}} 
\def\calI{{\cal  I}}

\def\calO{{\cal  O}} 
\def\calP{{\cal  P}}

\def\calX{{\cal  X}}

\def\bzero{\bfm 0}

\newcommand{\bfsym}[1]{\bm{#1}}

\def\bbeta{\bfsym \beta}
\def\bgamma{\bfsym \gamma}             
           \def\bDelta {\bfsym {\Delta}}

             \def\bSigma{\bfsym \Sigma}
        \def\bLambda {\bfsym {\Lambda}}

\def\btau{\bfsym {\tau}}

\def \calGT{\calG^{\ast}}
\def \GGT{\calG\calGT}
\def \calAT{\calA^{\ast}}

\def \mVT{\mV^{\ast}}

\def \tran {\mathsf{T}}
\def \tranH{\ast}

\def \bzero{\bm 0}

\newcommand{\kw}[2]{{}{#2}}

\newcommand{\cjc}[1]{{{#1}}}

\newcommand{\msh}[2]{{}{#2}}

\begin{document}
	
\title{Vectorized Hankel Lift: A Convex Approach for Blind Super-Resolution of Point Sources \footnotetext{Authors  are listed alphabetically.}}
\author[1]{Jinchi Chen}
\author[1,2]{Weiguo Gao}
\author[1]{Sihan Mao}
\author[1]{Ke Wei}
\affil[1]{School of Data Science, Fudan University, Shanghai, China.\vspace{.15cm}}
\affil[2]{School of Mathematical Sciences, Fudan University, Shanghai, China.}



\maketitle
\begin{abstract}
We consider the problem of resolving $ r$ point  sources from $n$ samples at the low end of the spectrum when point spread functions (PSFs) are not known. Assuming that the spectrum samples of the PSFs lie in low dimensional subspace (let $s$ denote the dimension), this problem can be reformulated as a matrix recovery problem, followed by location estimation. By exploiting the low rank structure of the vectorized Hankel matrix associated with the target matrix, a convex approach called Vectorized Hankel Lift is proposed for the matrix recovery. 
It is shown that $n\gtrsim rs\log^4 n$ samples are sufficient for Vectorized Hankel Lift to achieve the exact recovery.  \msh{}{For the location retrieval from the matrix, applying the single snapshot MUSIC method within the vectorized Hankel lift framework corresponds to the  spatial smoothing technique proposed to improve the performance of  the MMV MUSIC for the direction-of-arrival (DOA) estimation.} 
\\

\end{abstract}


\section{Introduction}
\subsection{Problem formulation}
\label{sec: problem formulation}
  
In this paper, we study the super-resolution of point sources when point spread functions (PSFs) are not known. More specifically, consider a point source signal $x(t)$ of the form
\begin{align}
\label{eq: point sources}
x(t) = \sum_{k=1}^{r}d_k \delta(t-\tau_k),
\end{align}
where $\delta(\cdot)$ is the Dirac function, $\{\tau_k\}$ and $\{d_k\}$ are the locations and amplitudes of the point source signal, respectively. Let $y(t)$ be its convolution with unknown point spread functions,
\begin{align}
\label{eq: convolution}
y(t)  =\sum_{k=1}^{r}d_k\delta(t-\tau_k)\ast g_k(t) = \sum_{k=1}^{r}d_k \cdot g_k(t-\tau_k),
\end{align}
where $\{g_k\}_{k=1}^r$ \kw{}{are the point spread functions  depending on} the locations of the point sources.

Taking the Fourier transform on both sides of \eqref{eq: convolution} yields
\begin{align}
\label{eq: countinuous convolution}
\widehat{y}(f)  = \int_{-\infty}^{+\infty} y(t)e^{-2\pi i f t}dt = \sum_{k=1}^{r}d_ke^{-2\pi i f \tau_k} \widehat{g}_k(f). 
\end{align}
\kw{}{The goal in blind super-resolution is to recover $\{d_k,\tau_k\}_{k=1}^r$}  from the low end of the spectrum
\begin{align}
\label{eq: discrete convolution}
\vy[j] =  \sum_{k=1}^{r}d_ke^{-2\pi i \tau_k\cdot j}\vg_k[j]\quad\text{for } j=0,\cdots, n-1
\end{align} 
when $\vg_k=[\widehat{g}_k(0),\cdots,\widehat{g}_k(n-1)]^\tran$, $k=1,\cdots,r$, are not known. Here we assume the index $j\in\{0,1,\cdots, n-1\}$ rather than  $j\in\{-\lfloor n/2\rfloor,\cdots,\lfloor n/2\rfloor \}$ only for convenience of notation. 
\kw{}{In addition to  blind super-resolution,} the  observation model \kw{}{ \eqref{eq: discrete convolution} also arises from many other important} applications, such as 3D single-molecule microscopy \cite{quirin2012optimal}, multi-user communication system \cite{luo2006low} and nuclear magnetic resonance spectroscopy \cite{qu2015accelerated}.

It is evident that the blind super-resolution problem is ill-posed without any further assumptions. To address this issue, \kw{}{we assume that the set of  vectors $\{\vg_k\}_{k=1}^r$ corresponding to the unknown point spread functions belong to a common and known low-dimensional subspace represented by} $\mB\in\C^{n\times s}$, i.e.,
\begin{align}
\label{eq: supspace assumption}
\vg_k = \mB\vh_k,
\end{align} 
where $\vh_k\in\C^s$ is the unknown orientation of $\vg_k$ in this subspace.  \kw{}{As is pointed out in \cite{yang2016super}, the subspace assumption is reasonable in several application scenarios. Moreover,} it has been extensively used in the literature, see for example \cite{ahmed2013blind,chi2016guaranteed,yang2016super,li2019atomic,li2019rapid}.

For any $\tau\in[0,1)$, define the vector $\va_{\tau}\in\C^n$ as
\begin{align*}
\va_{\tau} = \begin{bmatrix}
1 & e^{-2\pi i\tau \cdot 1}  &&\cdots& e^{-2\pi i \tau\cdot (n-1)}\numberthis\label{eq:atom_def}
\end{bmatrix}^\tran.
\end{align*}
Let $\vb_j\in\C^s$ be the $j$th column vector of $\mB^\tranH$. If we define the matrix $\mX^\natural\in\C^{s\times n}$ as
\begin{align}
\label{eq: data matrix}
\mX^\natural =\sum_{k=1}^{r}d_k\vh_k\va_{\tau_k}^\tran,
\end{align}
then under the subspace assumption \eqref{eq: supspace assumption} \msh{}{and using the lifting trick \cite{ahmed2013blind, candes2013phaselift, choudhary2014identifiability, chi2016guaranteed, yang2016super, li2016identifiability, ye2016compressive, ling2017blind, li2019rapid}}, the observation model \eqref{eq: discrete convolution} can be reformulated as a linear measurement of $\mX^\natural$:
\begin{align}
\label{eq: measurements}
\vy[j] = \big\langle \vb_j\ve_j^\tran, \sum_{k=1}^{r}d_k\vh_k\va_{\tau_k}^\tran  \big\rangle\text{ for } j=0,\cdots, n-1,
\end{align}
where the inner product of two matrices is given by $\la\mA,\mB\ra = \trace\lb\mA^*\mB\rb$,  $\ve_j$ is $(j+1)$th column of the $n\times n$ identity matrix $\mI_n$, and throughout this paper vectors and matrices are indexed starting with zero. Moreover, we can further rewrite \eqref{eq: measurements} in the following compact form,
\begin{align}
\label{eq: linear sampling}
\vy = \calA(\mX^\natural),
\end{align}
where $\calA:\C^{s\times n}\rightarrow \C^n$ is a linear operator defined by $[\calA(\mX)]_j=\langle\vb_j\ve_j^\tran,\mX\rangle$. 
\kw{}{The adjoint of the operator $\calA(\cdot)$, denoted $\calA^*(\cdot)$, is defined as $\calA^\ast(\vy) = \sum_{j=0}^{n-1}\vy[j]\vb_j\ve_j^\tran$.}

\kw{}{Based on the above reformulation of  blind super-resolution under the subspace assumption, it can be seen that the key is to recover $\mX^\natural$ from the linear measurement vector $\vy$.}  Once $\mX^\natural$ is reconstructed, the frequency components can be extracted from $\mX^\natural$ by the subspace methods which will be detailed in Section \ref{freq retrieval}. After the frequency components are obtained,  $\{d_k,\vh_k\}$ can be recovered by 
solving a least squares system. Moreover, due to the multiplicative form of $d_k$ and $\vh_k$ in \eqref{eq: data matrix}, we only expect to recover them separately up to a scaling ambiguity. \kw{}{Thus, we will assume that $\vecnorm{\vh_k}{2}=1$ without loss of generality.}

\cjc{
Note that the formulations in \eqref{eq: discrete convolution} and \eqref{eq: linear sampling}  are by no means new and they have been utilized in \cite{yang2016super}. Moreover, when the point spread function $\vg$ is shared among all point sources (i.e., the stationary case), \eqref{eq: discrete convolution} reduces to the blind sparse spikes deconvolution model considered in  \cite{chi2016guaranteed}. To recover the target matrix $\mX^\natural$ from the linear measurements $\vy$, following the approach developed in \cite{tang2013compressed} for spectrally sparse signal recovery, a similar atomic norm minimization method is proposed in  \cite{yang2016super},
\begin{align}
	\label{eq atomic norm minimization}
	\min_{\mX} ~\|\mX\|_{\calB} \text{ subject to }y=\calA(\mX),
\end{align}
where the atomic norm $\|\mX\|_{\calB}$ is defined as 
	\begin{align*}
		\|\mX\|_{\calB} :=\inf \{ t>0: \mX\in t\cdot\conv(\calB) \} = \inf_{d_k, \tau_k, \|\vh_k\|_2=1}\bigg\{ \sum_{k=1}^r d_k: \mX = \sum_{k=1}^r d_k \vh_k \va_{\tau_k}^\tranH, d_k>0 \bigg\},
	\end{align*}
The successful recovery guarantee of \eqref{eq atomic norm minimization} is studied in \cite{yang2016super}, while the robust analysis is provided separately in  \cite{li2019atomic}. Note  that for  spectrally sparse signal recovery, in addition to  atomic norm minimization, there are also methods which exploit the low rank property of the structured matrix formed from the signal \cite{chen2014robust, cai2018spectral, cai2019fast}. This motivates  us to develop a low rank approach for blind super-resolution.

}

\subsection{Exploiting the low rank structure: Vectorized Hankel Lift}
\label{sec: vectoried Hankel}

We start with a brief view of spectrally sparse signal recovery based on the hidden low rank structure. Let $x(t)$ be a spectrally sparse signal consisting of $r$ complex sinusoids,
\begin{align*}
x(t) = \sum_{k = 1}^{r} d_k e^{-2\pi i t \tau_k}.
\end{align*}
Let $\vx=[x(0),\cdots,x(n-1)]^\tran$ be a vector of length $n$ which is obtained by sampling $x(t)$ at $n$ contiguous, equally-spaced points. In a nutshell, spectrally sparse signal recovery is about reconstructing the signal $\vx$ from its partial samples.
\kw{}{Recalling the definition of $\va_\tau$ in \eqref{eq:atom_def},} we can represent $\vx$ as
\begin{align}
\label{eq: spectrally sparse}
\vx 
= \sum_{k=1}^rd_k \va_{\tau_k}^\tran.
\end{align}
Let $\calH$ be a linear operator which maps a vector $\vx$ into an $n_1\times n_2$ Hankel matrix,
\begin{align}
\label{eq: standard hankel}
\calH(\vx) = \begin{bmatrix}
x_0 & x_1 &\cdots & x_{n_2-1}\\
x_1 & x_2 &\cdots & x_{n_2}  \\
\vdots& \vdots&\ddots & \vdots\\
x_{n_1-1}&x_{n_1}&\cdots & x_{n-1}
\end{bmatrix}\in\C^{n_1\times n_2},
\end{align}
where $x_i$ is the $i$th entry of $\vx$ and $n_1+n_2 = n+1$. Without loss of generality, we assume $n_1=n_2=(n+1)/2$ in this paper. Due to the particular expression of $\vx$ in \eqref{eq: spectrally sparse}, it is not hard to see that the rank of $\calH(\vx)$ is at most $r$ according to the Vandermonde decomposition of $\calH(\vx)$ \kw{}{\cite{chen2014robust}}. \kw{Also add another reference where Hankellization is first proposed.}{}

Note  that the expression for the data matrix $\mX^\natural$ in \eqref{eq: data matrix} is overall similar to that for the spectrally sparse vector $\vx$ in \eqref{eq: spectrally sparse}, except that the weights $d_k\vh_k$ in front of $\va_{\tau_k}^\tran$ in \eqref{eq: data matrix} are vectors and consequently $\mX^\natural$ is a matrix rather than a vector.  Intuitively, if we treat each column of $\mX^\natural$ as a single element and form a matrix in the same fashion as in \eqref{eq: standard hankel}, it can be expected that the resulting matrix is also low rank. This is indeed true. Specifically, let $\calH$ be the vectorized Hankel lifting operator which maps a matrix $\mX\in\C^{s\times n}$ with columns $\{ \vx_j \}$ into an $sn_1\times n_2$ matrix,
\begin{align}
\label{eq: vec hankel}
\calH(\mX) = \begin{bmatrix}
\vx_0 & \vx_1 &\cdots & \vx_{n_2-1}\\
\vx_1 & \vx_2 &\cdots & \vx_{n_2}  \\
\vdots& \vdots&\ddots & \vdots\\
\vx_{n_1-1}&\vx_{n_1}&\cdots & \vx_{n-1}
\end{bmatrix}\in\C^{sn_1\times n_2},
\end{align}
where $n_1+n_2 = n+1$. 
\kw{}{To distinguish the matrix $\calH(\mX)$ in \eqref{eq: vec hankel} from the one in \eqref{eq: standard hankel}, we refer to $\calH(\mX)$ as the {\em vectorized Hankel matrix associated with}  $\mX$.}
Then a simple algebra yields that the vectorized Hankel matrix $\calH(\mX^\natural)$ associated with $\mX^\natural$ \kw{}{appearing in the blind super-resolution problem} admits the following decomposition:
\begin{align*}
\calH(\mX^\natural) = \mE_{\vh, L}\diag(d_1,\cdots, d_r) \mE_R^\tran,\numberthis\label{eq:vandermonde}
\end{align*}
where the matrices $\mE_{\vh, L}$ and $\mE_R$ are given by
\begin{align}
\label{eq: left}
\mE_{\vh, L} &= \begin{bmatrix}
\vh_1 & \vh_2 &\cdots& \vh_r\\
\vh_1 e^{-2\pi i \tau_1\cdot 1} & \vh_2 e^{-2\pi i \tau_2\cdot 1} &\cdots & \vh_r e^{-2\pi i \tau_r\cdot 1}\\
\vdots  & \vdots & \ddots & \vdots\\
\vh_1 e^{-2\pi i \tau_1\cdot(n_1-1)} &\vh_2e^{-2\pi i \tau_2\cdot (n_1-1)} &\cdots &\vh_re^{-2\pi i \tau_r\cdot (n_1-1)}\\
\end{bmatrix}\in\C^{s n_1\times r}
\end{align}
and 
\begin{align}
\label{eq: right}
\mE_R &= \begin{bmatrix}
1 & 1 & \cdots & 1\\
e^{-2\pi i \tau_1} &e^{-2\pi i \tau_2}   &\cdots & e^{-2\pi i \tau_r}\\
\vdots &\vdots &\ddots &\vdots \\
e^{-2\pi i \tau_1\cdot (n_2-1)} &e^{-2\pi i \tau_2\cdot (n_2-1)}   &\cdots & e^{-2\pi i \tau_r\cdot (n_2-1)}\\ 
\end{bmatrix}\in\C^{n_2\times r}.
\end{align}
It follows immediately that the rank of $\calH(\mX^\natural)$ is at most $r$ and thus it is a low rank matrix when 
$r$ is smaller than $\min(s n_1,n_2)$. 

In this paper we adopt the popular nuclear norm minimization to exploit the low rank structure of $\calH(\mX^\natural)$, yielding a convex approach for the reconstruction of $\mX^\natural$ which is also referred to {\em Vectorized Hankel Lift}.  
Exact
recovery guarantee will be established based on certain assumptions on the subspace matrix $\mB$ in \eqref{eq: supspace assumption}.
\subsection{{\color{black}Other} Related Work}
{\color{black} In this section, we give  a  brief introduction of other related work in addition to \cite{chi2016guaranteed, yang2016super, li2019atomic}.}
When the point spread functions are known and do not depend on the locations of the point sources, the measurement model \eqref{eq: discrete convolution} reduces to
\begin{align}
\label{eq: super resolution}
\vy[j] =  \sum_{k=1}^{r}d_ke^{-2\pi i \tau_k\cdot j}\text{ for } j=0,\cdots, n-1.
\end{align}
\kw{}{In this case,  estimating the locations  $\tau_k$ and amplitudes $d_k$ from $\vy$ is typically known as super-resolution or line spectrum estimation.  This problem arises
	in many areas of science and engineering, such as array imaging \cite{krim1996two,stoica2005spectral}, Direction-of-Arrival (DOA) estimation \cite{schmidt1982signal}, and inverse scattering \cite{fannjiang2010compressive}}.  The solution to this problem can date back to Prony \cite{prony1795essai}. 
In the Prony's method, \kw{}{the locations are retrieved from the roots of a polynomial whose coefficients form an 
	annihilating filter for the observation vector. Nevertheless, the Prony's method is numerical unstable despite that in the noiseless setting successful retrieval is guaranteed in exact arithmetic. 
	As alternatives, several subspace methods have been developed, including
	MUSIC \cite{schmidt1986multiple},  ESPRIT \cite{roy1989esprit}, and the matrix pencil method \cite{hua1990matrix}. In the absence of noise, the subspace methods are also able to identify the locations of the point sources. When there is noise, the stability of these methods has  been studied in \cite{liao2016music,liao2015music,li2020super,moitra2015super} in the regime when $\Delta>C/n$, where $\Delta$ is the minimum (wraparound) separation between any two locations, and $C>1$ is a proper numerical constant.  The analysis essentially relies on the lower bound on the smallest singular value of the Vandermonde matrix.
	The super-resolution limits of MUSIC and ESPRIT have been discussed in \cite{li2019conditioning, li2020super}, which is about the noise level that can be tolerated in order for the algorithms to achieve super-resolution when $\Delta<1/n$. In this regime, it is difficult to obtain a general and nontrivial lower bound on the smallest singular value of the Vandermonde matrix. Thus, the super-resolution limits in \cite{li2019conditioning, li2020super} are established for point sources whose locations obey certain configurations. }

Inspired by compressed sensing and low rank matrix reconstruction, various optimization based methods have also been developed for super-resolution and related problems. In \cite{candes2014towards}, the total variation (TV) minimization method is used to resolve the locations of the point sources. It is shown that when $\Delta>C/n$, exact recovery of the locations can be guaranteed. Moreover, the solution to the TV minimization problem can be computed by solving a semidefinite programming (SDP). Note, in the discrete setting, super-resolution can be interpreted within the framework of compressed sensing. However, since the measurement model in super-resolution considers the low end spectrum, and hence is deterministic, the typical successful recovery guarantee for compressed sensing \cite{candes2006robust} cannot sufficiently explain the success of the TV norm minimization method for super-resolution. The robustness of TV norm minimization is studied in \cite{candes2013super}, and the super-resolution problem of non-negative point sources is considered in \cite{duval2015exact,schiebinger2018superresolution,denoyelle2017support,duval2020characterization,eftekhari2019sparse}. \kw{}{Moreover, super-resolution from time domain samples has been investigated in \cite{bendory2016robust,bernstein2019deconvolution,eftekhari2019sparse}.}
\kw{}{When only partial entries of $\vy$ are observed in \eqref{eq: super resolution},  filling in the missing entries  is indeed the spectrally sparse signal recovery problem.} Motivated by the work in \cite{chandrasekaran2012convex}, an atomic norm minimization  method (ANM) is proposed for this problem. 
It is shown that $\vy$ can be reconstructed from $\calO(r\log r\log n)$ random samples provided the frequencies are well separated. ANM has been extended in \cite{li2015off,yang2016exact}  to handle the case when multiple measurement vector (MMV) are available. In the setting of MMV,  multiple snapshots of observations are collected and they share the same frequencies information. \kw{}{As  already mentioned previously, the Hankel matrix corresponding to $\vy$ is a low rank matrix. Consequently,  spectrally sparse signal recovery  can be reformulated as a low rank Hankel matrix completion problem, and replacing the rank objective with the nuclear norm yields  a recovery method known as EMaC. It  has been shown that EMaC is able to
	reconstruct a spectrally sparse signal with high probability provided the number of observed entries
	is $O(r\log^4n)$. 
	{ In \cite{yang2018sparse}, a formulation of EMaC for the multi-snapshots scenario is  presented.  
	}
	Additionally,  based on the low rank property of the Hankel matrix,  provable non-convex algorithms have been developed in \cite{cai2018spectral,cai2019fast}  to reconstruct spectrally sparse signals.}  Later, Zhang et.al. \cite{zhang2018multichannel} extend one of the non-convex algorithms to complete an MMV matrix, and in this work the same vectorized Hankel lift technique is used to  exploit the hidden low rank structure. 
Recently, \kw{}{a matrix completion problem based on the low dimensional structure in the transform domain is studied in \cite{chen2019exact}.  More precisely, it is assumed that after applying the Fourier transform to each column of the target matrix,  each row of the resulting matrix will be a spectrally sparse signal. Since it does not require the spectrally signals share the same frequency information, a 
	block-diagonal low rank structure is adopted to exploit the low dimensional structure.  Exact recovery guarantee is also established provided the sampling complexity is nearly optimal.}

\kw{}{Apart from super-resolution and spectrally sparse signal recovery, our work is also related to blind deconvolution. After the reparametrization of the signal and blurring kernel under the subspace assumption \cite{ahmed2013blind}, the goal in blind deconvolution is to recover the vectors $\vx^\natural$ and $\vh^\natural$ simultaneously  from the measurement vector in the form of  
	$$\vy = \diag(\mB\vh^\natural)\mA\vx^\natural.$$
	Noting that the above measurement model can be reformulated as a linear operation on a rank-$1$ matrix, a nuclear norm minimization method is proposed for blind deconvolution. The performance guarantee of the method has been  established in the case when $\mB$ is a partial Fourier matrix and $\mA$ is a Gaussian matrix. A non-convex gradient descent approach for blind deconvolution is developed and analyzed in \cite{li2019rapid}, and the identifiability problem is studied in \cite{li2016identifiability, choudhary2014identifiability}.
}

\subsection{Notation and Organization}
Throughout this work, vectors, matrices and operators are denoted by bold lowercase letters, bold uppercase letters and calligraphic letters, respectively. Note that  vectors and matrices are indexed starting with zero. \kw{}{The letter} $\calI$ denotes the identity operator. We use $\mG_i$ to denote the matrix defined by
\begin{align}
\label{eq: Hankel basis}
\mG_i = \frac{1}{\sqrt{w_i}}\sum_{\substack{j+k=i\\0\leq j\leq n_1-1\\ 0\leq k \leq n_2-1}} \ve_j\ve_k^\tran,
\end{align}
where $w_i$ is a constant defined as 
\begin{align}
\label{def wi}
w_i = \#\{(j,k) | j+k=i, 0\leq j \leq n_1-1, 0\leq k\leq n_2-1 \}.
\end{align}
In fact, $ \{\mG_i\}_{i=0}^{n-1}$ forms an orthonormal basis of the space of $n_1\times n_2$ Hankel matrices.

We use $\vx[i]$ to denote the $i$th entry of $\vx$ and $\mX_{j,k}$ or $\mX[j,k]$ to denote the $(j,k)$th entry of $\mX$. Additionally, the $i$th row and $j$th column of $\mX$ are denoted by $\mX_{i,\cdot}$ and $\mX_{\cdot,j}$, respectively. Furthermore, we  use the MATLAB notation $\mX(i:j,k)$ to denote a vector of size $j-i+1$, with entries $\mX_{i,k},\cdots,\mX_{j,k}$, i.e.,
\begin{align*}
\mX(i:j,k)=\begin{bmatrix}
\mX_{i,k},\cdots, \mX_{j,k}
\end{bmatrix}^\tran.
\end{align*}
For any matrix $\mX$, $\trace(\mX), \mX^\tranH, \mX^\tran$ and $\vect(\mX)$ are used to denote the trace, conjugate transpose, transpose and column vectorization of $\mX$, respectively. Also, $\opnorm{\mX}$, $\fronorm{\mX}$ and $\nucnorm{\mX}$ denote its spectral norm, Frobenius norm and nuclear norm, respectively.  

We use $\diag(\va)$ to denote the diagonal matrix specified by the vector $\va$. For a natural number $n$, we use $[n]$ to denote the set $\{0,\cdots,n-1\}$. 
For any two matrices $\mA,\mB$ \kw{}{of the same size}, their inner product is defined as  $\la\mA,\mB \ra = \trace(\mA^\tranH\mB)$. 
Moreover, we will \kw{}{refer to} $\mA\circ\mB,\mA\otimes\mB,\mA\odot\mB$ as the Hadamard, Kronecker product and Khatri-Rao product respectively. \kw{}{More precisely}, the Hadamard product is  the element-wise product of two matrices and the Kronecker product between  $\mA$ and $\mB$ is given by
\begin{align*}
\mA\otimes\mB = \begin{bmatrix}
\mA_{11}\mB & \mA_{12}\mB &\cdots &\mA_{1r}\mB\\
\mA_{21}\mB & \mA_{22}\mB &\cdots &\mA_{2r}\mB\\
\vdots &\vdots &\ddots &\vdots\\
\mA_{s1}\mB & \mA_{s2}\mB &\cdots &\mA_{sr}\mB\\
\end{bmatrix}\in\C^{sn_1\times rn_2},
\end{align*}
and the Khatri-Rao product is given by
\begin{align*}
\mA\odot\mB = \begin{bmatrix}
\va_1\otimes\vb_1 &\cdots &\va_r\otimes\vb_r\\
\end{bmatrix}\in\C^{sn_1\times r},
\end{align*}
where $\va_i$, $\vb_i$ denote the $i$th column of $\mA$ and $\mB$, respectively. By the application of \kw{}{the} Khatri-Rao product, we can rewrite $\mE_{\vh, L}$ in \eqref{eq: left} as $\mE_{\vh, L} = \mE_L\odot \mH$, where $\mE_L$ and $\mH$ are matrices given by
\begin{align}
\label{eq: Eleft}
\mE_L = \begin{bmatrix}
1 & 1 & \cdots & 1\\
e^{-2\pi i \tau_1} &e^{-2\pi i \tau_2}   &\cdots & e^{-2\pi i \tau_r}\\
\vdots &\vdots &\ddots &\vdots \\
e^{-2\pi i \tau_1\cdot (n_1-1)} &e^{-2\pi i \tau_2\cdot (n_1-1)}   &\cdots & e^{-2\pi i \tau_r\cdot (n_1-1)}\\ 
\end{bmatrix}\in\C^{n_1\times r}
\end{align}
and $\mH = \begin{bmatrix}\vh_1&\cdots &\vh_r
\end{bmatrix}\in\C^{s\times r}$. 

\kw{}{Throughout this paper, $c, c_1, c_2, \cdots$ denote absolute positive numerical constants whose values may vary from line to line. The notation $n\gtrsim f(m)$ means that there exists an absolute constant $c > 0$ such that
	$n \geq c\cdot f(m)$. Similarly, the notation $n\lesssim f(m)$ means that there exists an absolute constant $c>0$
	such that $n \leq  c\cdot f(m)$.}

The rest of this paper is organized as follows. \kw{}{Section \ref{sec: recovery guarantee} begins with the presentation of Vectorized Hankel Lift  and its recovery guarantee,} followed by the retrieval of  the point source locations. \kw{}{Numerical results to demonstrate the performance of Vectorized Hankel Lift is presented at the end of Section \ref{sec: recovery guarantee}}. The proofs of the main result are provided from Section \ref{sec:proofs} to Section \ref{aux}. Finally, we conclude this paper with a few future directions in Section \ref{conclusion}.

\section{Vectorized Hankel Lift and Frequency Retrieval}
\label{sec: recovery guarantee}
\subsection{Vectorized Hankel Lift and recovery guarantee}

Under the assumption that $\calH(\mX^\natural)$ is a low rank matrix, it is natural to reconstruct $\mX^\natural$ by solving the affine rank minimization problem
\begin{align}
\label{opt: rank mini}
\min \rank(\calH(\mX))\text{ s.t. }\vy = \calA(\mX).
\end{align}
However, the problem \eqref{opt: rank mini} is computational intractable due to the rank objective. Since the nuclear norm of a matrix is the tightest convex envelope of the matrix rank, seeking a solution with a small nuclear norm is also able to enforce the low rank structure. 
Therefore, instead of solving \eqref{opt: rank mini} directly, we consider the following nuclear norm minimization problem  for the recovery of  $\mX^\natural$:
\begin{align}
\label{convex opti}
\min_{\mX\in\C^{s\times n}} \nucnorm{\calH(\mX)}\text{ s.t. } \calA(\mX) = \vy.
\end{align}
In this paper, we refer to \eqref{convex opti} as  Vectorized Hankel Lift. 
There are many existing software packages that can be used to solve this problem. Thus we restrict our attention on the theoretical recovery guarantee of Vectorized Hankel Lift and investigate when the solution of \eqref{convex opti} coincides with $\mX^\natural$. 

We need to reformulate \eqref{convex opti} in order to facilitate the analysis. Let $\mZ$ be an $sn_1\times n_2$ matrix which can be expressed as
\begin{align*}
\mZ = \begin{bmatrix}
\vz_{0,0}  &\cdots &\vz_{0,n_2-1}\\
\vdots &\ddots &\vdots\\
\vz_{n_1-1,0} &\cdots &\vz_{n_1-1,n_2-1}\\
\end{bmatrix}\in\C^{sn_1\times n_2},
\end{align*}
where $\vz_{j,k}=\mZ(js:(j+1)s-1,k)$ for $j=0,\cdots,n_1-1$ and $k=0,\cdots,n_2-1$. Recall that $\calH$ is the vectorized Hankel lift operator defined in \eqref{eq: vec hankel}. The adjoint of $\calH$, denoted $\calH^\ast$, is a linear mapping from $sn_1\times n_2$ matrices to matrices of size ${s\times n}$. In particular, for any matrix $\mZ\in\C^{sn_1\times n_2}$, the $i$th column of $\calH^\ast(\mZ)$ is given by
\begin{align*}
\calH^\tranH(\mZ)\ve_i = \sum_{\substack{j+k = i\\0\leq j\leq n_1-1\\ 0\leq k \leq n_2-1}}\vz_{j,k}, \text{ for } i=0,\cdots, n-1.
\end{align*} Letting $\calD^2 = \calH^\tranH\calH$, we have 
\begin{align*}
\calD^2(\mX) = \begin{bmatrix}
w_0\vx_0 &\cdots &w_{n-1}\vx_{n-1}
\end{bmatrix},\quad \mbox{for any }\mX\in\C^{s\times n},
\end{align*}
where the scalar $w_i$ is defined as
\begin{align*}
w_i = \#\{(j,k) | j+k=i, 0\leq j \leq n_1-1, 0\leq k\leq n_2-1 \} \text{ for }i=0,\cdots, n-1.
\end{align*}
Moreover, we define $\calG = \calH\calD^{-1}$. Then 
\begin{align}
\label{eq: calG}
\calG(\mX)=\sum_{i=0}^{n-1}\calG \lb\vx_i\ve_i^\tran\rb = \sum_{i=0}^{n-1}\mG_i\otimes\vx_i,
\end{align}
where the set of matrices $\{\mG_i\}_{i=0}^{n-1}$ defined in \eqref{eq: Hankel basis} forms an orthonormal basis of the space of $n_1\times n_2$ Hankel matrices. The adjoint of $\calG$, denoted $\calG^\ast$, is given by $\calG^\ast=\calD^{-1}\calH^\ast$. Additionally, $\calG$ and $\calG^\ast$ satisfy
\begin{align*}
\calGT\calG = \calI\quad\quad
\opnorm{\calG}=1,\quad\mbox{and } \opnorm{\calGT}\leq 1.
\end{align*}
Letting $\mZ = \calH(\mX) = \calG\calD(\mX)$, it can be readily verified that
\begin{align*}
\calD(\mX) = \calGT(\mZ)\quad\mbox{and}\quad(\calI - \GGT)(\mZ) = \bzero. \end{align*}  
Furthermore, define $\mD = \diag(\sqrt{w_0},\cdots, \sqrt{w_{n-1}})$. We have $\calA\calD(\mX) = \mD\calA(\mX)$ for any matrix $\mX$.
Therefore, the optimization problem \eqref{convex opti} can be reformulated as 
\begin{align}
\label{opti}
\min_{\mZ\in\C^{sn_1\times n_2}} \nucnorm{\mZ} \text{ s.t. } \mD\vy = \calA\calGT(\mZ) \text{ and } (\calI - \GGT)(\mZ) = \bzero.
\end{align}

Due to the equivalence between \eqref{convex opti} and \eqref{opti}, it suffices to investigate the recovery guarantee  of  \eqref{opti}.
To this end, we make two assumptions.
\cjc{
\begin{assumption}
	\label{assumption 1}
	The column vectors $\{\vb_j\}_{j=0}^{n-1}$ of the subspace matrix $\mB^\tranH$ are independently and identically sampled  from a distribution $F$ which obeys the following properties: 
	\begin{itemize}
		\item Isotropy property. A distribution $F$ obeys the isotropy property if for $\vb\sim F$,
		\begin{align}
		\label{eq: isotropy}
		\E{\vb\vb^\tranH} = \mI_s.
		\end{align}
		\item Incoherence property. A distribution $F$ satisfies the incoherence property with parameter $\mu_0$ if for $\vb\sim F$,
		\begin{align}
		\label{eq: incoherence b}
		\max_{0\leq \ell \leq s-1} \left| \vb[\ell] \right|^2\leq \mu_0
		\end{align}
		holds, where $\vb[\ell]$ denotes the $\ell$th entry of $\vb$.
		\item For $\vb\sim F$, the sampled column vectors $\{\vb_j\}_{j=0}^{n-1}$ satisfy
		\begin{align} 
		\label{eq: lower bound b}
			\min_{0\leq j \leq n-1}\vecnorm{\vb_j}{2}^2 \geq 1.
		\end{align}
		\end{itemize}
\end{assumption}
}
\msh{}{The first two conditions (\ref{eq: isotropy}) and (\ref{eq: incoherence b}) in Assumption \ref{assumption 1} are first introduced in \cite{candes2011probabilistic} in the context of compressed sensing and these two properties are also made in \cite{chi2016guaranteed,yang2016super,li2019atomic} for the blind super-resolution problem.} If $F$ has mean zero, the isotropy condition states that the entries of $\vb$ have unit variance and are uncorrelated, which implies $\mu_0\geq 1$ in the incoherence property. The lower bound  $\mu_0=1$ is achievable by several examples, for instance, when the components  of $\vb$ are 
Rademacher random variables taking the values $\pm 1$  with equal probability or  $\vb$ is  uniformly sampled  from
the rows of a Discrete Fourier Transform  (DFT) matrix. \msh{}{In addition to (\ref{eq: isotropy}) and (\ref{eq: incoherence b}), we also need (\ref{eq: lower bound b}) to establish our main result. However, we would like to point out that (\ref{eq: lower bound b}) is not a stringent condition, but holds (either trivially or with high probability) by many common random ensembles. 
\begin{itemize}
    \item If the components of $\vb$ are Rademacher random variables or $\vb$ is uniformly sampled from the rows of a DFT matrix, it is trivial that for any fixed $j\in[n]$, $\vecnorm{\vb_j}{2}^2 = s\geq 1$. 
	\item Suppose the components of $\vb$ are independently and identically sampled from a distribution with mean zero and unit variance, such as the uniform distribution on the interval $[-\sqrt{3},\sqrt{3}]$. In such case, we can apply the bounded difference inequality to show that (\ref{eq: lower bound b}) holds with high probability, see Lemma~\ref{lem: bounded difference ineqaulity}. 	
\end{itemize}
}

\begin{assumption}
	\label{assumption 2}
	There exists a constant $\mu_1>0$ such that
	\begin{align}
	\label{basic incoherence}
	\sigma_{\min}(\mE_L^\tranH\mE_L) \geq \frac{n_1}{\mu_1}\quad\text{and}\quad\sigma_{\min}(\mE_R^\tranH\mE_R) \geq \frac{n_2}{\mu_1},
	\end{align}
	where $\mE_L$ and $\mE_R$ are given in \eqref{eq: Eleft} and \eqref{eq: right} and $\sigma_{\min}(\cdot)$ denotes the smallest singular value of a matrix.
\end{assumption}

Assumption~\ref{assumption 2} is the same as the one made in 
\cite{chen2014robust,cai2018spectral,cai2019fast} for spectrally sparse signal recovery. Later, we will show that $\sigma_{\min}(\mE_{\vh,L}^\tranH\mE_{\vh,L}) \geq \frac{n_1}{\mu_1}$ also holds when $\sigma_{\min}(\mE_L^\tranH\mE_L) \geq \frac{n_1}{\mu_1}$, see Lemma~\ref{lem: incoherence 02}.
Recalling the definition of $\mE_L$ and $\mE_R$, this assumption is essentially about the conditioning property  of the Vandermonde matrix.
This property is studied in \cite{liao2016music} through the discrete Ingham inequality \cite{ingham1936some} 
and in \cite{moitra2015super} through the discrete large sieve inequality \cite{vaaler1985some}. 
{\color{black} In particular, it  follows from \cite{moitra2015super} that Assumption~\ref{assumption 2} holds when the minimum wrap-around distance between the frequencies, denoted $\Delta$, satisfies \begin{align*}
    \Delta \geq \frac{2\mu_1/(\mu_1-1)}{n}.\numberthis\label{eq:condition delta}
\end{align*}}

We are in position to present the main result of this  paper.
\begin{theorem}[Exact recovery guarantee of Vectorized Hankel Lift]
	\label{main result}
	Under Assumptions~\ref{assumption 1} and \ref{assumption 2}, $\mZ^\natural=\calH(\mX^\natural)$ is the unique optimal solution to \eqref{opti} \msh{}{with probability exceeding $1-c_0(sn)^{-c_1} - ns^{-c_2}$, provided that $n\gtrsim \mu_0\mu_1\cdot sr\log^4(sn)$, where $c_0, c_1, c_2$ are absolute constants.}
\end{theorem}

{\color{black}\begin{remark}
The sampling complexity established in \cite{yang2016super} for the atomic norm minimization method is $n\gtrsim \mu_0\cdot sr\log^3(sn)$. While this is slightly better than the sampling complexity for Vectorized Hankel Lift, our analysis is based on less stringent assumptions. In our analysis,  the coefficients  $\{\vh_k\}_{k=1}^r$ are not required  to be i.i.d. samples from the uniform distribution on the complex unit sphere, but can be any unit norm vectors. In addition,  noting that the right-hand side of \eqref{eq:condition delta} is about { $2/n$} for moderately large $\mu_1$, which is smaller than $4/n$, the separation required in the main result of \cite{yang2016super}. It is worth noting that the  robust  analysis  of  the  atomic  norm  minimization  method  has  been  studied  in \cite{li2019atomic} and we will leave the robust analysis of Vectorized Hankel Lift for future work.
\end{remark}}

The proof of Theorem~\ref{main result} follows a well established route that has been widely used for compressed sensing and low rank matrix recovery. In a nutshell, a dual variable needs to be constructed to verify the optimality of $\mZ^\natural$. 
That being said, the details
of the proof itself are nevertheless quite involved and technical, and cannot be covered by the results from existing works. In particular, we need to show that there exists a partition of the measurements satisfying a list of desirable properties in order to construct the dual certificate.


\subsection{{\color{black}Variants of} MUSIC for frequency retrieval}
\label{freq retrieval}
In this section, we discuss the subspace method, particularly the MUltiple SIgnal Classification (MUSIC) algorithm \cite{schmidt1986multiple}, for computing the frequency parameters $\{\tau_k\}_{k=1}^r$ from the matrix $\mX^{\natural}$. Note that once $\{\tau_k\}_{k=1}^r$ are obtained, the weights $\{d_k, \vh_k\}$ can be computed by solving an overdetermined linear system. As can be seen later, 
applying the idea of the single snapshot MUSIC {\color{black}  to $\calH(\mX^\natural)$ yields a variant which is equivalent to  the existing spatial smoothing technique proposed to improve the performance of   the Multiple Measurement Vector (MMV) MUSIC. }

The careful reader may notice that every single row of $\mX^{\natural}$ is a spectrally sparse signal of the form \eqref{eq: spectrally sparse}, and moreover, all the rows share the same frequency parameters $\{\tau_k\}_{k=1}^r$. Thus we can apply the single snapshot MUSIC algorithm to a row of $\mX^\natural$ for frequency retrieval. 
Let $\vx_{\ell} =\sum_{k=1}^r d_k \vh_k[\ell] \va_{\tau_k}^\tran, 1\leq \ell \leq s$. Recall that $\calH(\vx_{\ell})$ is the Hankel matrix of rank $r$ and it admits the Vandermonde decomposition 
\begin{align*}
\calH(\vx_{\ell}) = \mE_L\diag(d_1 \vh_1[\ell],\cdots, d_r \vh_r[\ell])\mE_R^\tran.\numberthis\label{eq:vanderHx}
\end{align*}
Moreover, letting 
\begin{align*}\calH(\vx_\ell)^\tran=
\begin{bmatrix}
\mU & \mU_\perp
\end{bmatrix}\begin{bmatrix}
\bSigma &\\
&\bzero\\
\end{bmatrix}\begin{bmatrix}
\mV^\tranH\\
\mV_\perp^\tranH
\end{bmatrix}\numberthis\label{eq:svdHx}
\end{align*}
be the SVD of $\calH(\vx_\ell)^\tran$, where $\mU\in\C^{n_2\times r}, \mU_\perp\in\C^{n_2\times (n_2-r)}, \bSigma\in\R^{r\times r}, \mV\in\C^{n_1\times r}$ and $\mV_\perp\in\C^{n_1\times (n_1-r)}$, it is evident that 
{$\mU$ and $\mE_R$ span the same column space.
}Note that $\mE_R = \begin{bmatrix} \va_{\tau_1},\cdots, \va_{\tau_{r}} \end{bmatrix}$, where $\va_{\tau_k} = \begin{bmatrix} 1,\cdots, e^{-2\pi i\tau_k\cdot (n_2-1)}\end{bmatrix}^\tran$. It follows from the property of the Vandermonde matrix that 
\begin{quote}\em \centering $\va_\tau\in\Range(\mE_R)$ if and only if $\tau\in \{\tau_1,\cdots, \tau_r\}$.
\end{quote}Therefore we conclude that $\tau\in \{\tau_1,\cdots, \tau_r\}$ if and only if $1/\vecnorm{\mU_\perp^\tranH\va_\tau}{2}^2=\infty$. The single snapshot MUSIC algorithm utilizes this idea to identify the frequencies, and it consists of the following two steps:
\begin{enumerate}
	\item Compute the  SVD of $\calH(\vx_\ell)^\tran$ as in \eqref{eq:svdHx};
	\item Identify $\{\tau_k\}_{k=1}^r$ as the $r$ largest local maxima of 
	the pseudospectrum: $f(\tau)=1/\vecnorm{\mU_\perp^\tranH\va_\tau}{2}^2$.
\end{enumerate}
Here we present the single snapshot MUSIC algorithm directly based on the Hankel matrix $\calH(\vx_\ell)$. Equivalently, it  can be interpreted from the autocorrelation matrix model for signals, see for example \cite{li2018recovery} and references therein.  In the  noiseless setting, it  is easy to see that the  single snapshot MUSIC algorithm is able to compute $\{\tau_k\}_{k=1}^r$  exactly. When noise exists in $\vx_\ell$, the procedure of the algorithm remains unchanged, but with the SVD of $\calH(\vx_\ell)^\tran$ being  replaced by the SVD of the noisy Hankel matrix and with $\mU_\perp$ being the left singular vectors corresponding to the $n_2-r$ smallest singular values. The stability analysis of the single snapshot algorithm is discussed in \cite{liao2016music}.

To motivate the new variant of the MUSIC algorithm for  estimating the frequencies from $\mX^\natural$, we note that $\mE_R$ appears as a separate   component both in the Vandermonde decomposition of $\calH(\vx_\ell)$  and that of $\calH(\mX^\natural)$, see \eqref{eq:vandermonde}  and \eqref{eq:vanderHx}.
Therefore, we can replace the SVD of $\calH(\vx_{\ell})^\tran$ with the SVD of $\calH(\mX^\natural)^\tran$ in the first step of the single snapshot MUSIC algorithm. This gives the following variant:
\begin{enumerate}
	\item Compute the  SVD of $\calH(\mX^\natural)^\tran$:  $\calH(\mX^\natural)^\tran = \begin{bmatrix}\mU& \mU_\perp \end{bmatrix}\bSigma\mV^\tranH$, where $\mU\in\C^{n_2\times r}$ and $\mU_\perp\in\C^{n_2\times (n_2-r)}$;
	\item Identify $\{\tau_k\}_{k=1}^r$ as the $r$ largest local maxima of 
	the pseudospectrum: $f(\tau)=1/\vecnorm{\mU_\perp^\tranH\va_\tau}{2}^2$.
\end{enumerate}
 The following lemma establishes a connection between this variant and the single snapshot MUSIC, showing that the former one actually utilizes the SVD of the matrix formed by stacking all $\calH(\vx_{\ell})$  ($\ell=1,\cdots,s$) together.

\begin{lemma}\label{lem:MUSIC via VHM}
	Let $\widetilde{\calH}(\mX^\natural)$ be a matrix constructed by stacking all $\calH(\vx_\ell)$ on top of one another:
	\begin{align*}
	\widetilde{\calH}(\mX^\natural) = \begin{bmatrix}
	\calH(\vx_1)\\
	\vdots\\
	\calH(\vx_s)\\
	\end{bmatrix}\in\C^{sn_1\times n_2}.
	\end{align*}
	There exists a permutation matrix $\mP\in\R^{sn_1\times sn_1}$ such that $\widetilde{\calH}(\mX^\natural) = \mP\calH(\mX^\natural)$.
\end{lemma}
\begin{proof}
	Following the Vandermonde decomposition, the $\ell$th block of $\widetilde{\calH}(\mX^\natural)$ can be rewritten as
	\begin{align*}
	\calH(\ve_\ell^\tran \mX^\natural) &= \mE_L\begin{bmatrix}
	d_1\cdot \vh_1[\ell]&&\\
	&\ddots &\\
	&&d_r\cdot \vh_r[\ell]\\
	\end{bmatrix}\mE_R^\tran\\
	&=(\mE_L\odot\ve_\ell^\tran\mH)\begin{bmatrix}
	d_1&&\\
	&\ddots &\\
	&&d_r\\
	\end{bmatrix}\mE_R^\tran
	\end{align*}
	where $\vh_i$ is the $i$th column of $\mH$ and $\vh_i[\ell]$ is the $\ell$th entry of $\vh_i$. Thus $\widetilde{\calH}(\mX^\natural)$ has the following decomposition
	\begin{align*}
	\widetilde{\calH}(\mX^\natural) = (\mH\odot \mE_L)\mD\mE_R^\tran.
	\end{align*}
	According to the commutative law in \cite[Section 1.10.3]{zhang2017matrix}, there exists a permutation matrix $\mP$ such that  $\mH\odot \mE_L = \mP(\mE_L\odot \mH)$. 
\end{proof}

\msh{}{Based on Lemma~\ref{lem:MUSIC via VHM}, we will see that the variant  obtained by applying the single snapshot MUSIC idea to $\calH(\mX^\natural)$ corresponds to the spatial smoothing technique (more precisely the forward only spatial smoothing  technique). First, treating  the  rows of $\mX^\natural$ as i.i.d samples of a random signal whose covariance matrix can be used to compute the signal space  $\mU$  as in \eqref{eq:svdHx}, MMV MUSIC \cite{stoica2005spectral} uses the principal  eigenspace of the empirical covariance matrix (up to a scaling factor $1/s$)
\begin{align*}
    \mR = \sum_{i=1}^s \vx_i\vx_i^*
\end{align*}
to compute $\mU$.  However, when the signal comes from coherence sources,   the performance of MMV MUSIC will degrade. To deal with this difficulty, the forward only spatial smooth technique proposes to increase the number of samples by partitioning each $\vx_i$ into $n_2$ overlapped short samples (with each short sample being of length $n_1$, where $n_1+n_2=n+1$), and then construct the empirical covariance matrix from  all  the $s\cdot n_2$ short samples. A simple algebra yields that the  new empirical covariance matrix is indeed given  by  (up to a scaling factor ${1}/{(sn_2)}$)
\begin{align*}
    \widehat{\mR} = \sum_{i=1}^s \calH(\vx_i)\calH(\vx_i)^\ast.
\end{align*}
It is not  hard to  see   that the principal eigenspace of $\widehat{\mR}$ is the same as the principal singular vector space of $\widetilde{\calH}({\mX^\natural})$. Thus, by Lemma~\ref{lem:MUSIC via VHM}, we know that the variant obtained by applying the single snapshot MUSIC idea to $\calH(\mX^\natural)$ is equivalent to the spatial smoothing MUSIC. For more details  about spatial smoothing, see  \cite{evans1981high, evans1982application, yang2019source}.
    }

\subsection{Extension to higher dimension}
Vectorized Hankel Lift and the analysis are easily extended to higher dimensional array recovery problem. For ease of exposition, we give a brief discussion of the two-dimensional (2D) case but emphasize that the situation in higher dimensions is similar. For the 2D blind super-resolution problem, the data matrix can be expressed as
\begin{align*}
\mY_{j,\ell} = \sum_{k=1}^r d_k e^{-2\pi i (j\cdot \tau_{1k} + \ell\cdot \tau_{2k})} \mG_k[j,\ell],
\end{align*}
where $d_k$ is the amplitude, $\btau_k:=(\tau_{1k},\tau_{2k})$ is the 2D frequency and $\mG_k$ corresponds to the Fourier samples of the unknown 2D point spread function. Letting $\va_{\tau_{s k}}=\begin{bmatrix}
1& e^{-2\pi i \tau_{s k}\cdot 1}&\cdots &e^{-2\pi i \tau_{s k}\cdot (n-1)}\end{bmatrix}^\tran\in\C^n$ for $s = 1,2$, the 2D data array can be rewritten in a more compact form:
\begin{align*}
\mY = \sum_{k=1}^r d_k \lb \va_{\tau_{1k}}\va_{\tau_{2k}}^\tran \rb\circ \mG_k,
\end{align*}
Likewise, we assume that there exists a subspace matrix $\mB\in\C^{n^2\times s}$ such that $\vect(\mG_k) = \mB\vh_k$ for any $k=1,\cdots,r$. Then
\begin{align*}
\vy:=\vect(\mY) &= \sum_{k=1}^rd_k\vect(\va_{\tau_{1k}}\va_{\tau_{2k}}^\tran)\circ\vect(\mG_k)=\sum_{k=1}^rd_k\lb\va_{\tau_{2k}}\otimes \va_{\tau_{1k}} \rb\circ\lb\mB\vh_k\rb.
\end{align*}
For any $0\leq j, \ell\leq n-1$, the $(jn+\ell)$th entry of $\vy$ is given by
\begin{align*}
\vy_{jn+\ell} &= \sum_{k=1}^r d_k \lb\va_{\tau_{2k}}\otimes \va_{\tau_{1k}} \rb^\tran\ve_{jn+\ell}
\lb\vb^\tranH _{jn+\ell}\vh_k\rb\\
&=\sum_{k=1}^r \trace\lb d_k \lb\va_{\tau_{2k}}\otimes \va_{\tau_{1k}} \rb^\tran\ve_{jn+\ell}
\lb\vb^\tranH _{jn+\ell}\vh_k\rb \rb\\
&=\sum_{k=1}^r \trace\lb \ve_{jn+\ell}
\vb^\tranH_{jn+\ell} d_k\vh_k  \lb\va_{\tau_{2k}}\otimes \va_{\tau_{1k}} \rb^\tran\rb\\
&=\Big\langle
\vb_{jn+\ell}\ve_{jn+\ell}^\tran, \sum_{k=1}^r  d_k\vh_k  \big(\va_{\tau_{2k}}\otimes \va_{\tau_{1k}} \big)^\tran\Big\rangle,
\end{align*}
where $\vb_{jn+\ell}$ is the $(jn+\ell)$th column of $\mB^\tranH$. Therefore, we have  $\vy = \calA(\mX^\natural)$, where
$\mX^\natural =\sum_{k=1}^r  d_k \va_{\tau_{2k}}^\tran \otimes  \lb\vh_k\va_{\tau_{1k}}^\tran\rb$, and 
$\calA:\C^{s\times n^2} \rightarrow\C^{n^2}$ is a linear operator given by
\begin{align*}
\lsb \calA(\mX)\rsb_{jn+\ell}=\la
\vb_{jn+\ell}\ve_{jn+\ell}^\tran, \mX\ra.
\end{align*}
As in the 1D case, the blind super-resolution problem is essentially about recovering the  target  matrix $\mX^\natural$ from the observation vector $\vy$.

Note that the target matrix $\mX^\natural$ can be rewritten as the following block form:
\begin{align*}
\mX^\natural 
&=\begin{bmatrix}
\sum_{k=1}^r  d_k\lb\vh_k\va_{\tau_{1k}}^\tran \rb & \sum_{k=1}^r d_k e^{-2\pi i\tau_{2k}}\lb\vh_k\va_{\tau_{1k}}^\tran \rb &\cdots & \sum_{k=1}^r d_k e^{-2\pi i\tau_{2k}\cdot(n-1)} \lb\vh_k\va_{\tau_{1k}}^\tran \rb
\end{bmatrix}.
\end{align*}
Letting $\mX^\natural_\ell:=\sum_{k=1}^r d_k e^{-2\pi i\tau_{2k}\cdot\ell} \lb\vh_k\va_{\tau_{1k}}^\tran \rb$, we define the two-fold vectorized Hankel lift  of $\mX^\natural$ as follows:
\begin{align*}
\calH(\mX^\natural) =  \begin{bmatrix}
\calH(\mX_0^{\natural}) & \calH(\mX_1^{\natural}) & \cdots &\calH(\mX_{n_2-1}^{\natural})\\
\calH(\mX_1^{\natural}) & \calH(\mX_2^{\natural}) & \cdots &\calH(\mX_{n_2}^{\natural})\\
\vdots &\vdots &\ddots &\vdots\\
\calH(\mX_{n_1-1}^{\natural}) & \calH(\mX_{n_1}^{\natural}) & \cdots &\calH(\mX_{n-1}^{\natural})\\
\end{bmatrix},
\end{align*}
where $\calH(\mX_i^{\natural})$ is the vectorized Hankel matrix defined in \eqref{eq: vec hankel}. It can be readily shown that $\calH(\mX^\natural)$ has the following decomposition
\begin{align}
\label{eq: 2D V-factor}
\calH(\mX^\natural) = \begin{bmatrix}
\lb \mE_L\odot\mH\rb\mY^0\\
\lb \mE_L\odot\mH\rb\mY^1\\
\vdots\\
\lb \mE_L\odot\mH\rb\mY^{n_1-1}\\
\end{bmatrix}\mD\begin{bmatrix}
\mY^0\mE_R^\tran & \mY^1\mE_R^\tran&\cdots &\mY^{n_2-1}\mE_R^\tran
\end{bmatrix}:=\mL\mD\mR^\tran,
\end{align}
where $\mE_L,\mE_R$ are two matrices defined in \eqref{eq: Eleft} and \eqref{eq: right} but with the frequencies $\{\tau_{1k}\}_{k=1}^r$, $\mH = \begin{bmatrix}\vh_1&\cdots&\vh_r\end{bmatrix}\in\C^{s\times r},\mD = \diag(d_1,\cdots,d_r)$ and $\mY=\diag(e^{-2\pi i\tau_{21}},\cdots,e^{-2\pi i\tau_{2r}} )$.

If all frequencies $\tau_{1k},\tau_{2k}$ are distinct and all $d_k$ are non-zeros, it is not hard to see that $\calH(\mX^\natural)$ is a low rank matrix. Therefore, we can  recover $\mX^\natural$ by solving the following convex programming
\begin{align} 
\label{opt: 2D}
\min_{\mX\in\C^{s\times n^2} } \nucnorm{\calH(\mX)}\text{ s.t. }\calA(\mX) = \vy.
\end{align}
\msh{}{The recovery guarantee of  \eqref{opt: 2D} can be similarly established in the following theorem.} 
The proof details are overall similar to that for Theorem \ref{main result}, and thus are omitted. 

\msh{}{
\begin{theorem}\label{thm:2D case}
	Under Assumption II.1 and suppose $\sigma_{\min}(\mL^\tranH\mL)\geq \frac{n_1^2}{\mu_1}$ and $\sigma_{\min}(\mR^\tranH\mR)\geq \frac{n_2^2}{\mu_1}$, the data matrix $\mX^\natural\in \C^{s\times n^2}$ is the unique optimal solution to \eqref{opt: 2D} with probability at least $1-c_0(sn)^{-c_1}-n^2s^{-c_2}$ for absolute constants $c_0, c_1, c_2$, provided that $n^2 \gtrsim \mu_0\mu_1\cdot sr\log^5(sn)$.
\end{theorem}
}

After the matrix $\mX^\natural$ is recovered, the frequency $\{\btau_k=(\tau_{1k}, \tau_{2k})\}_{k=1}^r$ can be estimated by a 2D-MUSIC algorithm \cite{berger2010signal,liao2015music,zheng2017super} based on the two-fold vectorized Hankel matrix $\calH(\mX^\natural)$ in \eqref{eq: 2D V-factor}, followed by the recovery of  $\{d_k\vh_k\}_{k=1}^r$ through least-squares.

\subsection{Numerical Experiments}
In this section, 
we  empirically evaluate  the  performance of Vectorized Hankel Lift for the recovery of $\mX^\natural$ in the  blind super-resolution problem.
Vectorized Hankel Lift is solved by SDPT3 \cite{tutuncu2001sdpt3} based on CVX \cite{cvx}.
The recovery ability of Vectorized Hankel Lift will be evaluated via the framework of empirical phase transition and we compare it with the atomic norm minimization method \cite{yang2016super}. The locations $\{\tau_k\}_{k=1}^r$ of the point source signals are generated randomly from  $[0,1)$, while the amplitudes $\{d_k\}_{k=1}^r$ are generated via $d_k = (1+10^{c_k})e^{-i\psi_k}$ with $\psi_k$ being uniformly sampled from $[0,2\pi)$ and $c_k$ being uniformly sampled from $[0,1]$. 
\msh{}{The subspace matrix $\mB$ are sampled from two random ensembles which all satisfy the conditions in Assumption~\ref{assumption 1}. The first one is the random submatrix sampled from the DFT matrix, and the other one is the random matrix whose entries satisfy the uniform distribution over $[-\sqrt{3},\sqrt{3}]$.} The coefficients  $\{\vh_k\}_{k=1}^r$ are i.i.d. standard Gaussian random vectors followed by normalization. 
In our tests, 20 Monte Carlo trails are repeated for each problem instance and we report the probability of successful recovery out of those trials. 
A trail is declared to be successful if the relative reconstruction  error of $\mX^\natural$ in terms of the Frobenius norm is less than $10^{-3}$.

We first fix $n=64$ and vary the values of $r$ and $s$. Figure \ref{fig: phase transition for subspace vs spikes 1}(a) and Figure \ref{fig: phase transition for subspace vs spikes 1}(b) show the phase transitions of Vectorized Hankel Lift and atomic norm minimization method when \msh{}{the subspace matrix $\mB$ is randomly sampled from the DFT matrix and} the locations of point sources are randomly generated without imposing the separation condition, and Figure \ref{fig: phase transition for subspace vs spikes 1}(c) illustrates the phase transition of the atomic minimization method when the separation condition $\Delta:= \min_{k\neq j}\lab\tau_k-\tau_j\rab\geq\frac{1}{n}$ is imposed.
{\color{black}Here we omit the phase transition plot of Vectorized Hankel Lift for the frequency separation case  because the plot is similar to Figure \ref{fig: phase transition for subspace vs spikes 1}(a).}
It can be observed that the atomic norm minimization method has a higher phase transition curve when the separation condition is satisfied. However, 
in contrast to  Vectorized Hankel Lift, its performance degrades severely when there is no frequency separation requirement. That is, Vectorized Hankel Lift is less sensitive to the separation condition.
\msh{}{We also conduct the phase transition tests when the entries of $\mB$ are i.i.d. sampled from the uniform distribution over $[ -\sqrt{3},\sqrt{3}]$. The phase  transition diagrams are presented in  Figure~\ref{fig: phase transition for subspace vs spikes 2}, and similar observations can be made.
Note that the phase transition plot of Vectorized Hankel Lift for the frequency separation case is still omitted due to the high similarity with Figure \ref{fig: phase transition for subspace vs spikes 2}(a). } 
\begin{figure}[ht!]
	\centering
	\subfigure[]{
		\includegraphics[width= .31\textwidth]{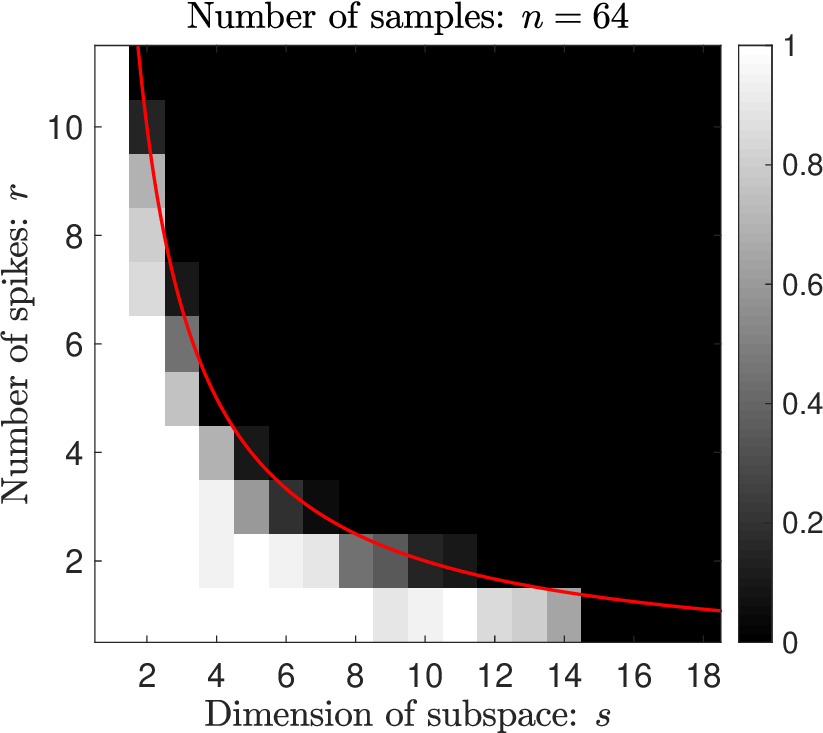}
	}
	\subfigure[]{
		\includegraphics[width= .31\textwidth]{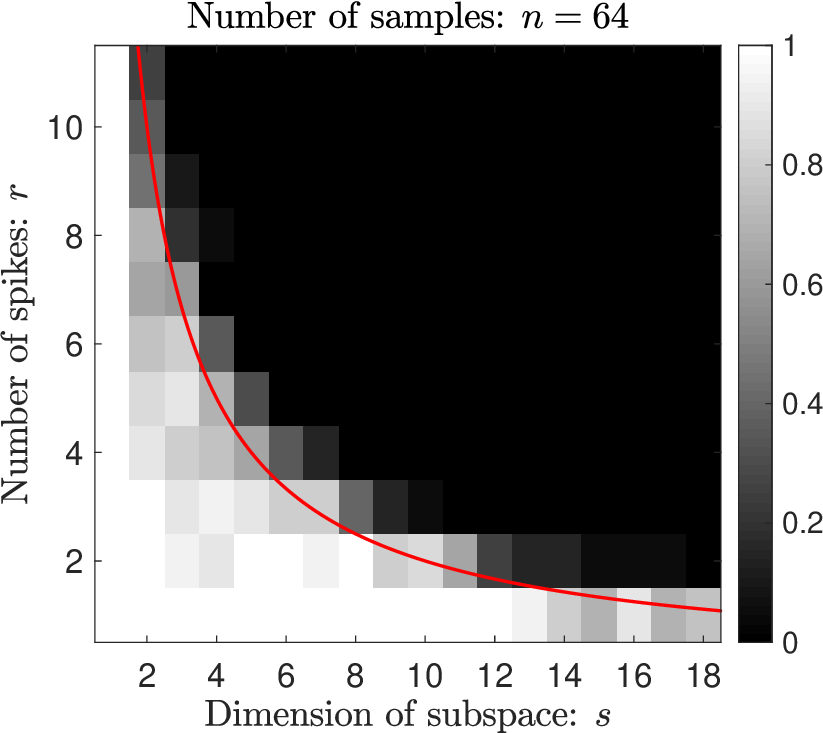}
	}
	\subfigure[]{
		\includegraphics[width= .31\textwidth]{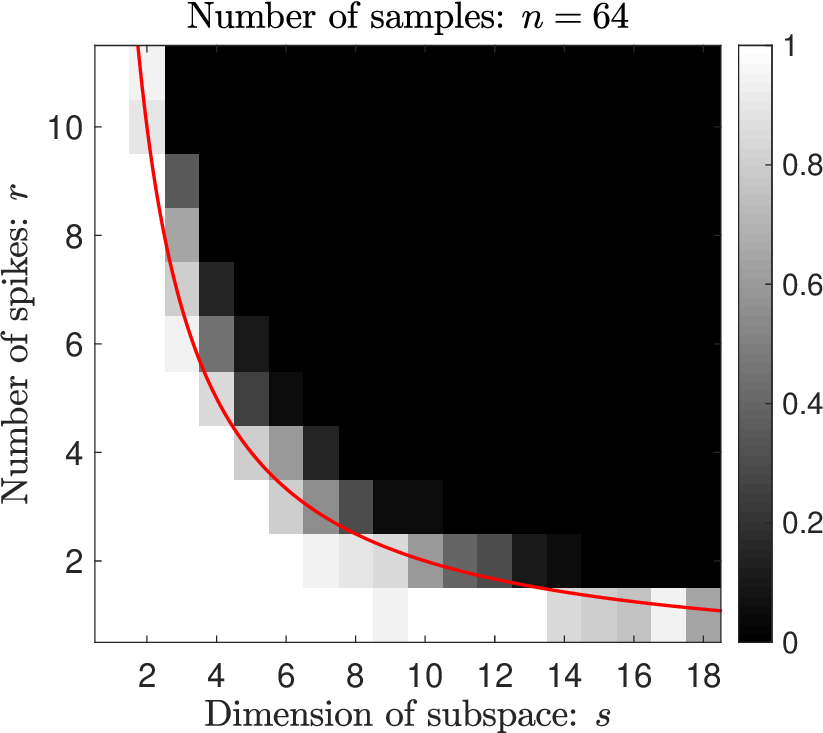}
		
	}
	
	\caption{\msh{}{The phase transitions of Vectorized Hankel Lift and the atomic norm minimization method when the subspace matrix $\mB$ is randomly sampled from the DFT matrix. (a) Vectorized Hankel Lift for randomly generated frequencies,  (b) atomic norm minimization for randomly generated frequencies,  and (c) atomic norm minimization  for frequencies obeying the separation condition $\Delta:= \min_{k\neq j}\lab\tau_k-\tau_j\rab\geq\frac{1}{n}$. The number of measurements is fixed to be $n=64$. The red curve plots the hyperbola curve $rs=20$.}}
	\label{fig: phase transition for subspace vs spikes 1}
\end{figure}
\begin{figure}[ht!]
	\centering
	\subfigure[]{
		\includegraphics[width= .31\textwidth]{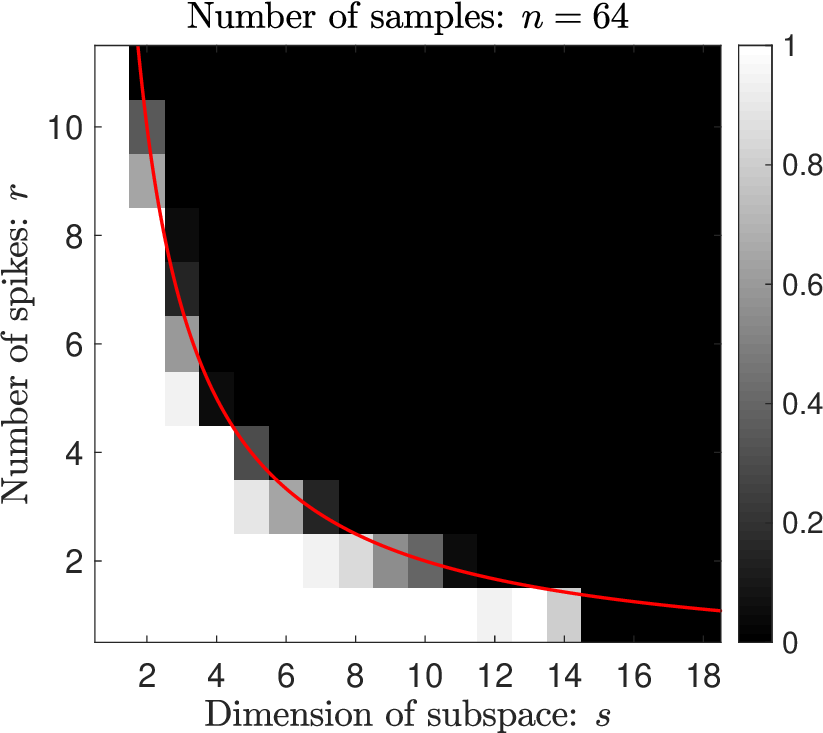}
	}
	\subfigure[]{
		\includegraphics[width= .31\textwidth]{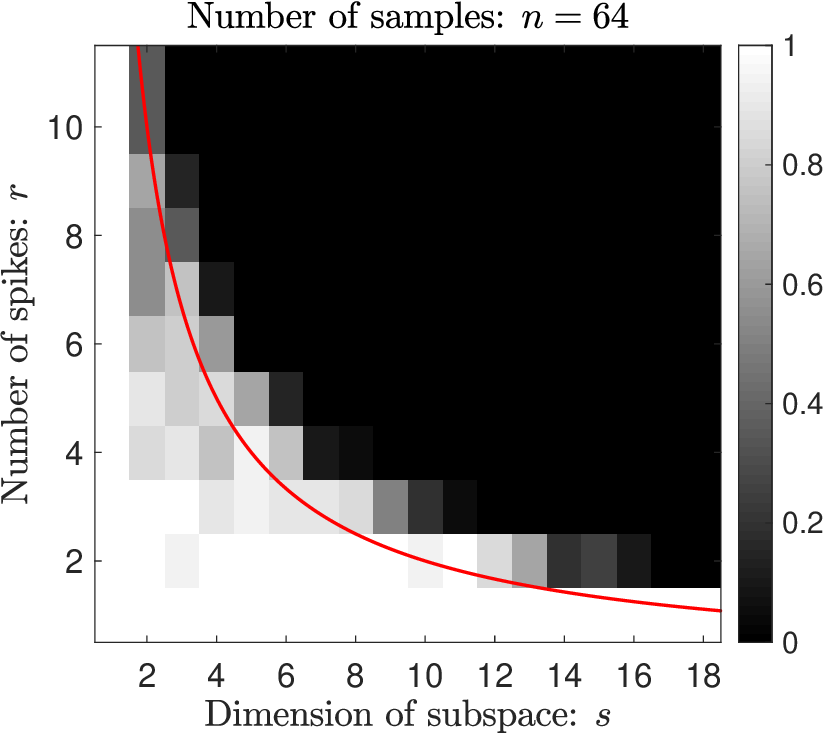}
	}
	\subfigure[]{
		\includegraphics[width= .31\textwidth]{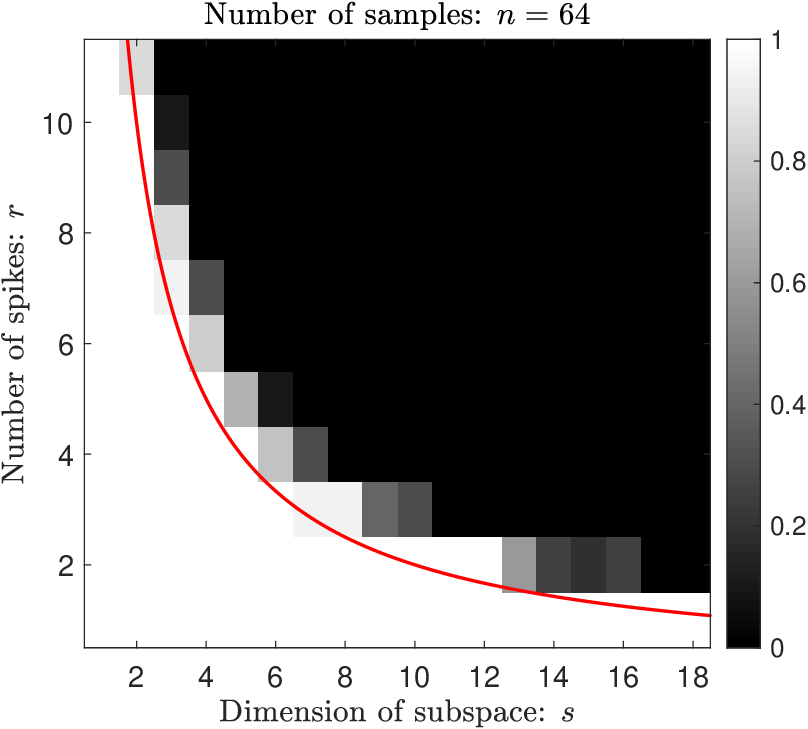}
		
	}
	
	\caption{\msh{}{The phase transitions of Vectorized Hankel Lift and the atomic norm minimization method when the entries of $\mB$ are i.i.d. sampled from the uniform distribution over $[ -\sqrt{3},\sqrt{3}]$. (a) Vectorized Hankel Lift for randomly generated frequencies,  (b) atomic norm minimization for randomly generated frequencies,  and (c) atomic norm minimization  for frequencies obeying the separation condition $\Delta:= \min_{k\neq j}\lab\tau_k-\tau_j\rab\geq\frac{1}{n}$. The number of measurements is fixed to be $n=64$. The red curve plots the hyperbola curve $rs=20$.}}
	\label{fig: phase transition for subspace vs spikes 2}
\end{figure}

\msh{}{In the above phase transition tests, the  coefficients $\{\vh_k\}_{k=1}^r$ are sampled from random Gaussian with normalization. In order to test whether the choice of $\{\vh_k\}_{k=1}^r$ matters, we also test another two cases for the coefficients. One is the Identical Gaussian, where $\{\vh_k\}_{k=1}^r$ are the same across $r$ (sampled from random Gaussian with normalization). The other one is  QR where $\{\vh_k\}_{k=1}^r$ are obtained from the Q matrix in the QR decomposition of an $s\times r$ random Gaussian matrix. Tests are conducted for fixed $s=4$ and $n=64$, and the plots of successful recovery probability against the number of spikes $r$ are presented in Figure~\ref{fig: diff h}. It can be clearly seen that no significant differences over different types of $\{\vh_k\}_{k=1}^r$ are observed from the plots. Therefore, the numerical results validate  that our main result can hold without any conditions of $\{\vh_k\}_{k=1}^r$.     

\begin{figure}[ht!]
	\centering
	\subfigure[]{
		\includegraphics[width= .4\textwidth]{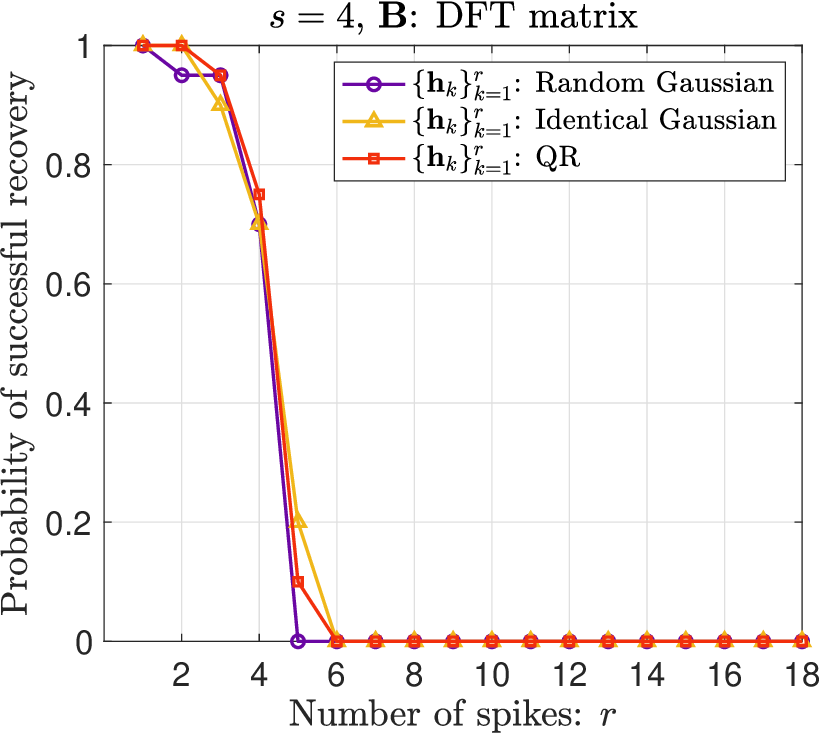}
	}\hspace{0.4cm}
	\subfigure[]{
		\includegraphics[width= .4\textwidth]{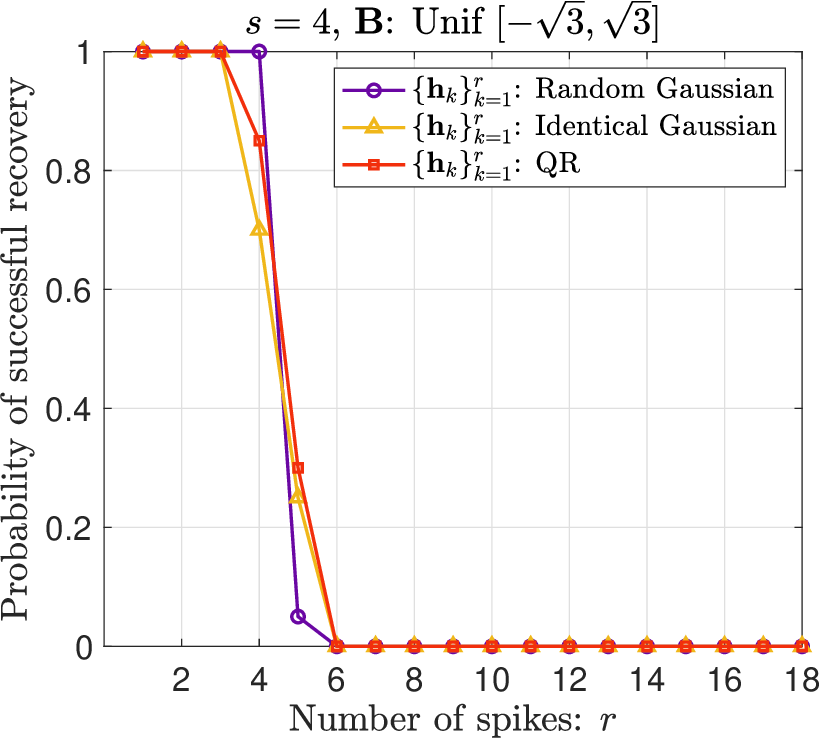}
	}
	
	\caption{\msh{}{The probability of successful recovery of Vectorized Hankel Lift against $r$ with three different subspace coefficients $\{\vh_k\}_{k=1}^r$ ($s=4$, $n = 64$). (a): The subspace matrix $\mB$ are randomly sampled from the DFT matrix. (b): The entries of $\mB$ are i.i.d. sampled from the uniform distribution over $[-\sqrt{3},\sqrt{3}]$.}}
	\label{fig: diff h}
\end{figure}
}

\msh{}{
In order to examine the effect of the separation condition more carefully, we further conduct tests for fixed $s=3$, $r=3$, and vary the number of samples $n$. In the tests, we impose that there are at least two spikes with separation equal to $1.0/n$ and $0.5/n$, respectively. For each problem instance, we repeat 50 Monte Carlo trails and report the probability of successful recovery out of those trials. The numerical results are presented in Figure~\ref{fig: bad sep}. It is evident that Vectorized Hankel Lift presents a better performance when the minimum separation is $0.5/n$. When the spikes are well separated (i.e., the minimum separation is $\bDelta
= 1.0/n$), the atomic norm minimization method performs better. In addition, the results confirm that Vectorized Hankel Lift is overall not affected by the separation condition. 
\begin{figure}[ht!]
	\centering
	\subfigure[]{
		\includegraphics[width= .4\textwidth]{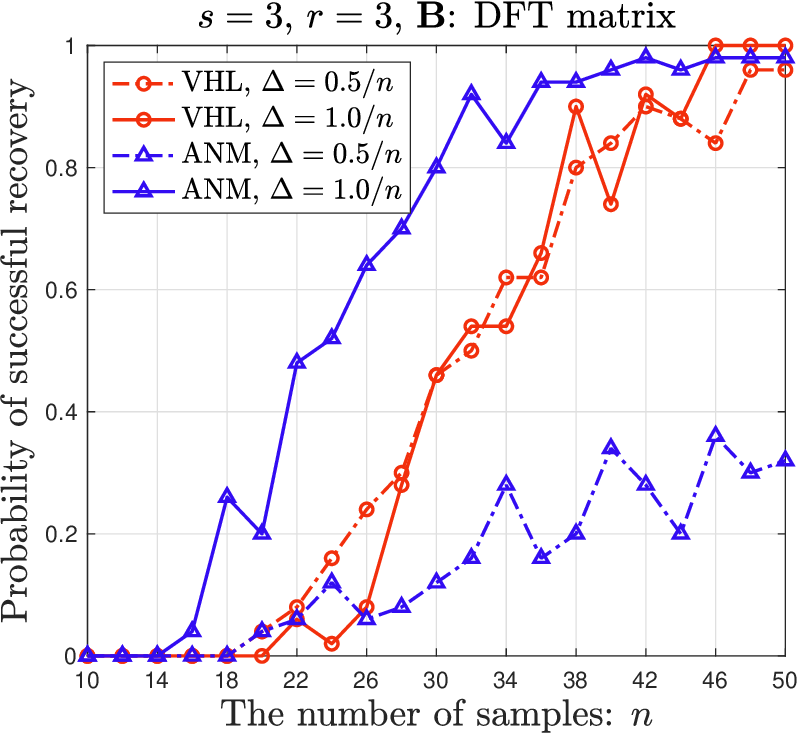}
	}\hspace{0.4cm}
	\subfigure[]{
		\includegraphics[width= .4\textwidth]{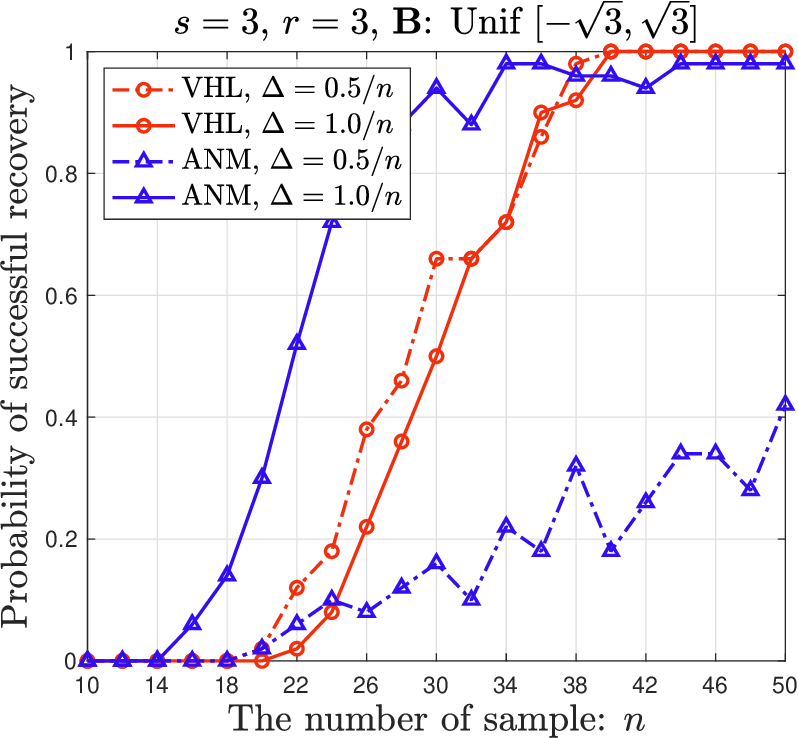}
	}
	
	\caption{\msh{}{The probability of successful recovery of Vectorized Hankel Lift and the atomic norm minimization method under two separation conditions, $\bDelta = \frac{0.5}{n}$ and $\bDelta = \frac{1.0}{n}$. The dimension of subspace and the number of spikes are both fixed to be $s=3$ and $r=3$. The number of samples $n$ is varied. (a): The subspace matrix $\mB$ are randomly sampled from the DFT matrix. (b): The entries of $\mB$ are i.i.d. sampled from the uniform distribution over $[-\sqrt{3},\sqrt{3}]$.}}
	\label{fig: bad sep}
\end{figure}
}


We also plot  the locations of the point sources  $\{\tau_k\}_{k=1}^r$ and the unknown point spread function samples $\{\vg_k\}_{k=1}^r$ computed from $\mX^\natural$ for a random instance corresponding to $n=64,s=3$ and $r=4$. 
We apply the MUSIC variant introduced in Section~\ref{freq retrieval} (i.e., the spatial smoothing MUSIC)  to localize the $\{\tau_k\}_{k=1}^r$. Figure \ref{fig: simple example}(a) shows the pseudospectrum $f(\tau)$ 
on a set of points on $[0,1]$ with equal distance $10^{-4}$. As can be seen from this figure, the function $f(\tau)$ peaks at the locations of true frequencies. After the $\{\tau_k\}_{k=1}^r$ are identified, the coefficients $\{\vh_k\}_{k=1}^r$ are computed by solving a least squares problem and $\{\vg_k\}_{k=1}^r$ are estimated  as $\mB\vh_k$.  Figure \ref{fig: simple example}(b) includes the plots of the estimates of $\{|\vg_k|\}_{k=1}^r$ against the true values which clearly show that $\{\vg_k\}_{k=1}^r$ can be recovered. 
\begin{figure}[ht!]
	\centering
	\subfigure[]{
		\includegraphics[width= .48\textwidth]{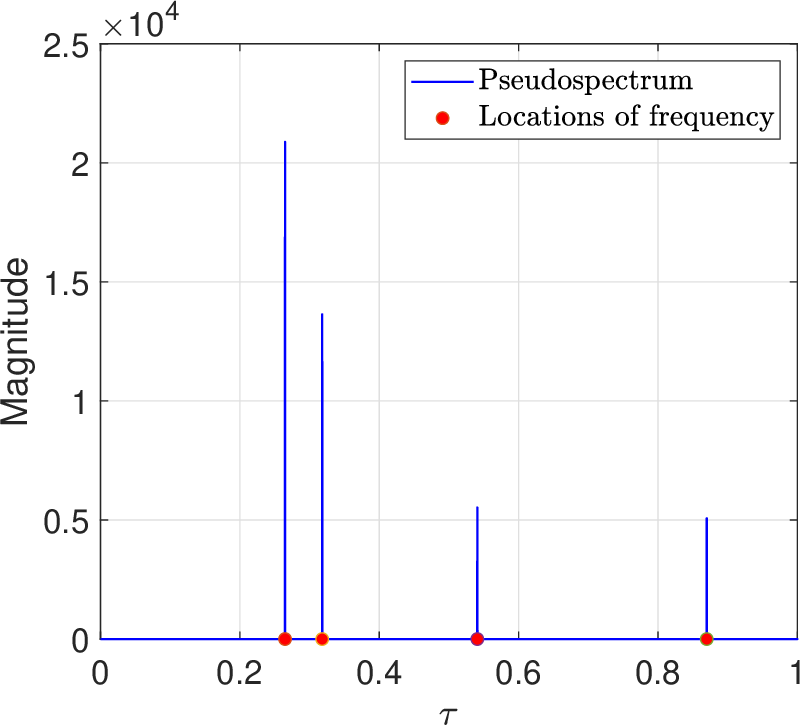}
	}
	\subfigure[]{
		\includegraphics[width= .48\textwidth]{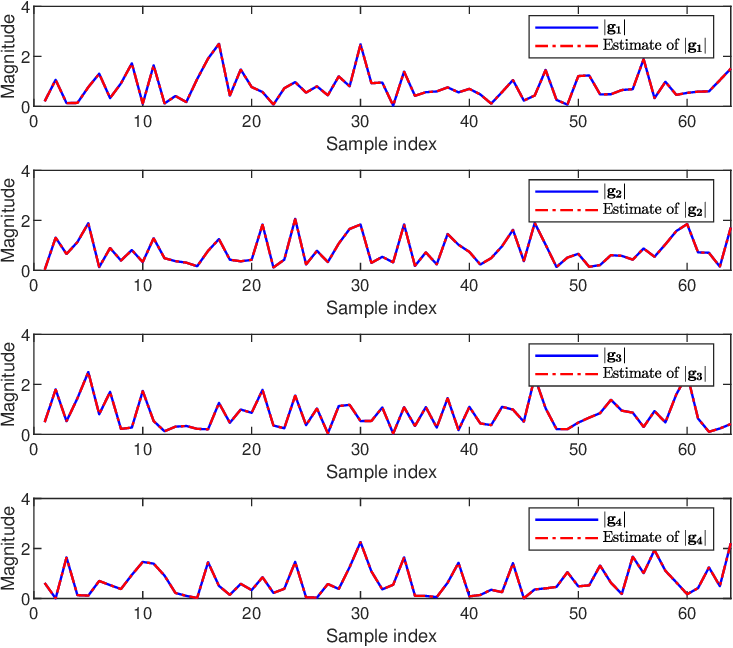}
	}
	
	\caption{(a) Plots of pseudospectrum $f(\tau)$ when $n=64,s=3,r=4$ and locations of the true frequencies when the subspace $\mB$ is generated randomly from the standard Gaussian distribution. (b) The magnitudes of Fourier samples of the point spread functions $\vg_1,\vg_2,\vg_3,\vg_4$ and their estimates from least squares.}
	\label{fig: simple example}
\end{figure}

\section{Proof Architecture of Main Result}\label{sec:proofs}
\subsection{\kw{}{Preliminaries}}

\msh{}{We first apply the bounded difference inequality to show that for the column vectors $\{\vb_j\}_{j=0}^{n-1}$ with independent entries, the condition \eqref{eq: lower bound b} in Assumption~\ref{assumption 1} holds with high probability given \eqref{eq: isotropy} and \eqref{eq: incoherence b}.
\begin{lemma}
    \label{lem: bounded difference ineqaulity}
    The column vectors $\{\vb_j\}_{j=0}^{n-1}$ of the subspace matrix $\mB^\tranH$ are independently and identically sampled from a distribution $F$ which obeys the conditions (\ref{eq: isotropy}) and (\ref{eq: incoherence b}) in Assumption~\ref{assumption 1}. Assume the components of $\vb$ are independent, the event \begin{align}
	\label{ineq: union lower bound of b}
	\min_{0\leq j \leq n-1} \vecnorm{\vb_j}{2}^2\geq 1
	\end{align}
	\kw{}{occurs} with probability at least $1-n\exp\lb-\frac{s}{16\mu_0^2}\rb$.  
\end{lemma}
\begin{proof}
    Since \kw{}{$\vb_j$ satisfies \eqref{eq: isotropy}}, we first have 
	\begin{align*}
	\E{\vecnorm{\vb_j}{2}^2} = \E{ \trace(\vb_j^\tranH\vb_j) }=\E{ \trace(\vb_j\vb_j^\tranH) } = s.
	\end{align*}
	\kw{}{Define $f(x_1,\cdots,x_s)=\sum_{i=1}^s|x_i|^2$. It is evident that 
		\begin{align*}
		|f(x_1,\cdots,x_{i-1},x_i,x_{i+1},\cdots,x_s)-f(x_1,\cdots,x_{i-1},x_i',x_{i+1},\cdots,x_s)|\leq |x_i|^2+|x_{i}'|^2\leq 2\mu_0
		\end{align*}
		when $|x_i|^2\leq\mu_0$ and $|x_{i}'|^2\leq\mu_0$. Because $\vb_j$ also satisfies \eqref{eq: incoherence b},}
	the application of \kw{}{the} bounded difference inequality  yields that
	\begin{align*}
	\P\lsb\lab\vecnorm{\vb_j}{2}^2 - s\rab\geq t\rsb \leq 2\exp\lb-\frac{t^2}{4s\mu_0^2}\rb.
	\end{align*}
	\kw{}{Consequently,} \kw{}{we can take} $t = \frac{s}{2}$ \kw{}{to obtain}  
	\begin{align*}
	\P\lsb\vecnorm{\vb_j}{2}^2 \geq \frac{s}{2}\rsb \geq 1 - \exp\lb-\frac{s}{16\mu_0^2}\rb.
	\end{align*}
	Taking the uniform bound yields that for all $j\in[n]$, with probability at least $1-n\exp\lb-\frac{s}{16\mu_0^2}\rb$, $\vecnorm{\vb_j}{2}^2 \geq \frac{s}{2} \geq 1$ when $s\geq 2$.
\end{proof}

}

{Next, we present a lemma about the basic properties of the linear operator $\calA$.
\begin{lemma}
	\label{lem: prop of A}
	Under Assumption~\ref{assumption 1}, the following properties hold:
	\begin{align}
	\label{ineq: lower bound of ATA}
	\la\vy, \calA\calA^\tranH(\vy)\ra \geq \vecnorm{\vy}{2}^2\quad \mbox{for any fixed vector $\vy\in\C^{n}$},
	\end{align}
	\begin{align}\label{eq:more of ATA}
	\opnorm{\calA\calA^\ast - \calI} \leq s\mu_0\text{ and }\opnorm{\calA} \leq \sqrt{s\mu_0 }.
	\end{align}
\end{lemma} 

\begin{proof}
	Since
	\begin{align*}
	\calA\calA^\ast(\vy) &= \calA\left( \sum_{i=0}^{n-1}\vy[i]\vb_i\ve_i^\tran\right)\\
	&=\begin{bmatrix}
	\la \vb_0\ve_0^\tran, \sum_{i=0}^{n-1}\vy[i]\vb_i\ve_i^\tran\ra\\
	\vdots\\
	\la \vb_{n-1}\ve_{n-1}^\tran, \sum_{i=0}^{n-1}\vy[i]\vb_i\ve_i^\tran\ra\\
	\end{bmatrix}= \begin{bmatrix}
	\vecnorm{\vb_0}{2}^2\cdot \vy[0]\\
	\vdots\\
	\vecnorm{\vb_{n-1}}{2}^2\cdot \vy[n-1]\\
	\end{bmatrix}\in\C^n,
	\end{align*}
	\kw{}{\eqref{ineq: lower bound of ATA} follows immediately from \eqref{ineq: union lower bound of b}}.
	
	The properties in \eqref{eq:more of ATA} follows directly from the definition of $\calA$. For the left inequality, we have 
	\begin{align*}
	\opnorm{\calA\calA^\ast - \calI} &= \sup_{\vy\in\C^n: \vecnorm{\vy}{2}=1}\vecnorm{\calA\calA^\ast(\vy) - \vy}{2}\\
	&=\sup_{\vy\in\C^n: \vecnorm{\vy}{2}=1} \sqrt{ \sum_{i=0}^{n-1}\left( \vecnorm{\vb_i}{2}^2-1 \right)^2\cdot |\vy[i]|^2 }\\
	&\leq \max_{0\leq i \leq n-1}\left|  \vecnorm{\vb_i}{2}^2-1  \right|\\
	&\leq s\mu_0.
	\end{align*}
	The right one can be proved as follows
	\begin{align*}
	\opnorm{\calA} &= \sup_{\mX\in\C^{s\times n}: \fronorm{\mX} = 1} \vecnorm{\calA(\mX)}{2}\\
	&=\sup_{\mX\in\C^{s\times n}: \fronorm{\mX} = 1} \sqrt{\sum_{i=0}^{n-1}|\vb_i^\tranH\mX\ve_i|^2}\\
	&\leq \sup_{\mX\in\C^{s\times n}: \fronorm{\mX} = 1} \sqrt{\sum_{i=0}^{n-1} \vecnorm{\vb_i}{2}^2\cdot \vecnorm{\mX\ve_i}{2}^2}\\
	&\leq \max_{0\leq i\leq n-1}\vecnorm{\vb_i}{2}\cdot \sup_{\mX\in\C^{s\times n}: \fronorm{\mX} = 1} \sqrt{\sum_{i=0}^{n-1} \vecnorm{\mX\ve_i}{2}^2}\\
	&\leq \sqrt{s\mu_0}.
	\end{align*}
	The proof is now complete. 
\end{proof}
}


\kw{ }{The following lemma suggests that the smallest singular value of $\mE_{\vh,L}$ can be lower bounded by the smallest singular value of $\mE_L$.}
\begin{lemma}
	\label{lem: incoherence 02}
	\msh{}{Recall that $\mH = \begin{bmatrix}\vh_1&\cdots &\vh_r
    \end{bmatrix}\in\C^{s\times r}$ and} suppose all columns of $\mH$ are \kw{}{of unit norm}. Under the incoherence condition \eqref{basic incoherence}, we have 
	\begin{align*}
	\sigma_{\min}(\mE_{\vh,L}^\tranH\mE_{\vh,L}) \geq \frac{n_1}{\mu_1},
	\end{align*}
	where $\mE_{\vh,L}$ is the matrix defined in \eqref{eq: left}.
\end{lemma}
\begin{proof}
	Let $\va_{\tau_\ell} = \begin{bmatrix}1 & e^{-2\pi i\tau_{\ell}} &\cdots &e^{-2\pi i\tau_{\ell}\cdot(n_1-1)}
	\end{bmatrix}^\tran\in\C^{n_1}$ be the $\ell$th column of $\mE_L$. Since $\mE_{\vh,L} = \mE_L\odot \mH$, it can be easily seen that
	\begin{align*}
	\mE_{\vh,L}^\tranH\mE_{\vh,L} &= \begin{bmatrix}
	\va_{\tau_1}^\tranH\otimes \vh_1^\tranH\\
	\vdots\\
	\va_{\tau_{r}}^\tranH\otimes \vh_{r}^\tranH\\
	\end{bmatrix}\begin{bmatrix}
	\va_{\tau_1}\otimes \vh_1 &\cdots & \va_{\tau_{r}}\otimes\vh_{r}
	\end{bmatrix}\\
	&=\begin{bmatrix}
	(\va_{\tau_1}^\tranH\otimes \vh_1^\tranH)(\va_{\tau_1}\otimes \vh_1)&\cdots &(\va_{\tau_1}^\tranH\otimes \vh_1^\tranH)(\va_{\tau_{r}}\otimes \vh_{r})\\
	\vdots &\ddots&\vdots\\
	(\va_{\tau_{r}}^\tranH\otimes \vh_{r}^\tranH)(\va_{\tau_1}\otimes \vh_1)&\cdots &(\va_{\tau_{r}}^\tranH\otimes \vh_{r}^\tranH)(\va_{\tau_{r}}\otimes \vh_{r})\\
	\end{bmatrix}\\
	&=\begin{bmatrix}
	(\va_{\tau_1}^\tranH\va_{\tau_1})\cdot (\vh_1^\tranH\vh_1) &\cdots&(\va_{\tau_1}^\tranH\va_{\tau_{r}})\cdot (\vh_1^\tranH\vh_{r})\\
	\vdots&\ddots&\vdots\\
	(\va_{\tau_{r}}^\tranH\va_{\tau_1})\cdot (\vh_{r}^\tranH\vh_1)&\cdots&(\va_{\tau_{r}}^\tranH\va_{\tau_{r}})\cdot (\vh_{r}^\tranH\vh_{r})\\
	\end{bmatrix}\\
	&= 	\begin{bmatrix}
	\va_{\tau_1}^\tranH\va_{\tau_1}  	&\cdots&\va_{\tau_1}^\tranH\va_{\tau_{r}} \\
	\vdots&\ddots&\vdots\\
	\va_{\tau_{r}}^\tranH\va_{\tau_1} 	&\cdots& \va_{\tau_{r}}^\tranH\va_{\tau_{r}} \\
	\end{bmatrix} 
	\circ\begin{bmatrix}
	\vh_1^\tranH\vh_1 	&\cdots 	&\vh_1^\tranH\vh_{r}\\
	\vdots 				&\ddots 	&\vdots\\
	\vh_{r}^\tranH\vh_1	& \cdots 	&\vh_{r}^\tranH\vh_{r}\\
	\end{bmatrix}\\
	&=(\mE_L^\tranH\mE_L)\circ(\mH^\tranH\mH),
	\end{align*}
	\kw{}{Recall that} {a selection matrix} $\mP\in\R^{n^2\times n}$ is the unique matrix such that
	\begin{align*}
	\mP\vz = \vect\lb \diag(\vz)\rb\text{ for all } \vz\in\C^n,
	\end{align*} 
	and it has the remarkable property that $\mP^\tran(\mA\otimes\mB)\mP = \mA\circ\mB$ \cite[Corollary 2]{visick2000quantitative}. Thus we have
	\begin{align*}
	\sigma_{\min}(\mE_{\vh,L}^\tranH\mE_{\vh,L}) &= \inf_{\|\msh{}{\bbeta}\|_2=1}\left| \msh{}{\bbeta}^\tranH\lb (\mE_L^\tranH\mE_L)\circ(\mH^\tranH\mH)\rb\msh{}{\bbeta} \right|\\
	&=\inf_{\|\msh{}{\bbeta}\|_2=1}\left| \msh{}{\bbeta}^\tranH \mP^\tran\lb (\mE_L^\tranH\mE_L)\otimes(\mH^\tranH\mH)\rb\mP\msh{}{\bbeta} \right|\\
	&=\inf_{\|\msh{}{\bbeta}\|_2=1}\left| \msh{}{\bbeta}^\tranH \mP^\tran  (\mE_L^\tranH\otimes \mH^\tranH)(\mE_L\otimes\mH) \mP\msh{}{\bbeta} \right|\\
	&=\inf_{\|\msh{}{\bbeta}\|_2=1}\|(\mE_L\otimes\mH) \mP\msh{}{\bbeta} \|_2^2\\
	&=\inf_{\|\msh{}{\bbeta}\|_2=1}\|(\mE_L\otimes\mH) \vect\lb\diag(\msh{}{\bbeta})\rb \|_2^2\\
	&=\inf_{\|\msh{}{\bbeta}\|_2=1}\|\vect\lb\mH\diag(\msh{}{\bbeta})\mE_L^\tran \rb\|_2^2\\
	&=\inf_{\|\msh{}{\bbeta}\|_2=1}\fronorm{ \mH\diag(\msh{}{\bbeta})\mE_L^\tran }^2\\
	&\geq \sigma_{\min}^2(\mE_L)\cdot \inf_{\|\msh{}{\bbeta}\|_2=1}\fronorm{ \mH\diag(\msh{}{\bbeta})}^2\\
	&=\sigma_{\min}^2(\mE_L)\cdot\inf_{\|\msh{}{\bbeta}\|_2=1}\sum_{k=1}^{r}\vecnorm{\msh{}{\bbeta}[k]\cdot \vh_k}{2}^2\\
	&=\sigma_{\min}^2(\mE_L)\cdot\inf_{\|\msh{}{\bbeta}\|_2=1}\sum_{k=1}^{r}|\msh{}{\bbeta}[k]|^2\\
	&\geq \frac{n_1}{\mu_1},
	\end{align*}
	which completes the proof.
\end{proof}

\kw{}{A straightforward application of Lemma~\ref{lem: incoherence 02} yields the following result,} which can be regarded as a variant of \cite[Lemma 1]{cai2019fast}.
\begin{lemma}
	\label{incoherence 2}
	Suppose $\calH{(\mX^\natural})$ obeys the incoherence condition \eqref{basic incoherence} with parameter $\mu_1$. Let $\calH(\mX^\natural)= \mU\mS\mV^\tranH$ be the singular value decomposition of $\calH(\mX^\natural)$, where $\mU\in\C^{sn_1\times r}, \mS\in\R^{r\times r}$ and $\mV\in\C^{n_2\times r}$.  \kw{}{If we rewrite $\mU$ as }
	\begin{align*}
	\mU = \begin{bmatrix}
	\mU_{0}\\
	\vdots\\
	\mU_{n_1-1}\\
	\end{bmatrix},
	\end{align*}
	where the $\ell$th block is $\mU_\ell = \mU(\ell s:(\ell+1)s-1,:)$ for $\ell=0,\cdots, n_1-1$, then
	\begin{align}
	\label{eq: incoherence}
	\max_{0\leq \ell \leq n_1-1} \fronorm{\mU_\ell}^2\leq \frac{\mu_1 r}{n} \text{ and } \max_{0\leq j \leq n_2-1} \vecnorm{\ve_j^\tran\mV}{2}^2 \leq \frac{\mu_1 r}{n},
	\end{align}
\end{lemma}
\begin{proof}
	We only need to prove the left inequality in \eqref{eq: incoherence} as the right one can be similarly established. Recall that $\calH(\mX^\natural) = \mE_{\vh,L}\diag(d_1,\cdots, d_r)\mE_R^\tran$. Since $\mU\in\C^{s n_1\times r}$ and $\mE_{\vh,L}$ span the same subspace and $\mU$ is orthogonal, there exists an orthonormal matrix $\mQ\in\C^{r\times r}$ such that $\mU = \mE_{\vh, L}(\mE_{\vh, L}^\tranH\mE_{\vh, L})^{-1/2}\mQ$. Therefore,  \begin{align*}
	\fronorm{\mU_\ell}^2 &= \sum_{j=\ell s}^{(\ell+1)s-1}\vecnorm{\ve_j^\tran\mE_{\vh, L}(\mE_{\vh, L}^\tranH\mE_{\vh, L})^{-1/2}}{2}^2\\
	&\leq \sum_{j=\ell s}^{(\ell+1)s-1}\vecnorm{\ve_j^\tran\mE_{\vh, L}}{2}^2\cdot\opnorm{(\mE_{\vh, L}^\tranH\mE_{\vh, L})^{-1/2}}^2\\
	&\leq \frac{\mu_1}{n}\cdot \sum_{j=\ell s}^{(\ell+1)s-1}\vecnorm{\ve_j^\tran\mE_{\vh, L}}{2}^2\\
	&= \frac{\mu_1}{n}\cdot\sum_{k=1}^r\vecnorm{e^{-2\pi i \tau_k\cdot \ell} \vh_k}{2}^2\\
	&= \frac{\mu_1 r}{n},
	\end{align*}
	where the second inequality is due to Lemma \ref{lem: incoherence 02}.
\end{proof}

The following corollary is a direct consequence of Lemma \ref{incoherence 2} and will be frequently used \kw{}{in the sequel}. 

\begin{corollary}
	Suppose $\calH{(\mX^\natural})$ obeys the incoherence condition \eqref{basic incoherence} with parameter $\mu_1$. Then, 
	\begin{align}
	\label{incoherence condition}
	\max_{0\leq i\leq n-1}\frac{1}{w_i}\sum_{\substack{\ell+j=i\\ 0\leq \ell\leq n_1-1\\ 0\leq j\leq n_2-1}}\fronorm{\mU_\ell}^2 \leq \frac{\mu_1 r}{n}\text{ and } \max_{0\leq i\leq n-1}\frac{1}{w_i}\sum_{\substack{\ell+j=i\\ 0\leq \ell\leq n_1-1\\ 0\leq j\leq n_2-1}}\vecnorm{\ve_j^\tran\mV}{2}^2\leq \frac{\mu_1 r}{n}.
	\end{align}
\end{corollary}

The matrix Bernstein inequality, stated below,  will be used  frequently in our analysis.
\begin{lemma}[\cite{tropp2012user,chen2014robust}]
		\label{lem: bernstein}
		Let $\{\mX_\ell\}_{\ell=1}^n$ be a set independent random matrices of dimension $n_1\times n_2$, which satisfy $\E{\mX_\ell} = 0$ and $\opnorm{\mX_\ell}\leq B$. Define 
		\begin{align*}
		\sigma^2:=\max\lcb\opnorm{\E{\sum_{\ell=1}^n\mX_\ell\mX_\ell^*}},\opnorm{\E{\sum_{\ell=1}^n\mX_\ell^*\mX_\ell}}\rcb.
		\end{align*}
		Then the event
		\begin{align}
		\label{ineq: bernstein}
		\opnorm{\sum_{\ell=1}^n \mX_\ell} \leq c\lb\sqrt{\sigma^2\log(n_1+n_2)}+B\log(n_1+n_2)\rb
		\end{align}
		holds with probability at least $1-(n_1+n_2)^{-c_1}$, where $c,c_1>0$ are absolute constants.
\end{lemma}


\subsection{Deterministic optimality condition}
\kw{}{As is typical in the analysis of low rank matrix recovery, in order to show that $\mZ^\natural$ is the unique optimal solution to the convex program \eqref{opti},  we need to construct a dual certificate which satisfies a set of sufficient conditions}. These conditions can be viewed as a variant of the KKT condition \kw{}{for the optimality of $\mZ^\natural$.} \msh{}{Recall that the singular value decomposition (SVD) of $\calH(\mX^\natural)$ is $\calH(\mX^\natural) = \mU\mS\mV^\tranH$. The tangent space $T$ of the nuclear norm at $\calH(\mX^\natural)$ can be defined as
\begin{align*}
T = \left\{ \mU\mA^\tranH + \mB\mV^\tranH~:~ \mA\in\C^{n_2\times r}, \mB\in\C^{sn_1\times r} \right\}.
\end{align*}
The projections $\calP_T(\mZ)$ onto the tangent space can be defined as
\begin{align}
\label{eq: PT}
	\calP_T(\mZ) := \mU\mU^*\mZ + \mZ\mV\mV^* - \mU\mU^*\mZ\mV\mV^\tranH.
\end{align}
and the corresponding projector onto the orthogonal complement of $T$ is given by $\calP_{T^\perp}(\mZ)= \mZ - \calP_T(\mZ)$.
}
\begin{theorem}
	\label{thr:optimality}
	Suppose 
	$\opnorm{\calA\calA^\ast}\geq 1$
	and 
	\begin{align}
	\label{ineq: RIP}
	\opnorm{\calP_{T}\calG\calAT \calA\calGT\calP_{T} -\calP_{T}\calG\calGT\calP_{T} } \leq\frac{1}{2}.
	\end{align}
	If there exists a dual certificate $\bLambda \in \C^{sn_1\times n_2}$ such that 
	\begin{align}
	\label{ineq: F norm}
	&\fronorm{\calP_{T}(\mU\mVT - \bLambda)}\leq \frac{1}{16s\mu_0}, \\
	\label{ineq: perp op norm}
	&\opnorm{\calP_{T^\perp}(\bLambda)} \leq \frac{1}{2},\\
	\label{belong}
	&\calGT(\bLambda) \in \rm{Range}(\calAT),
	\end{align}
	then $\mZ^\natural$ is the unique solution to \eqref{opti}.
\end{theorem}
\begin{proof} \kw{}{The structure of the proof is overall similar to those in} \cite{chen2014robust,chen2015incoherence,chen2019exact}.
	Consider any feasible solution $\mZ^\natural + \msh{}{\mM}$, where the perturbation $\msh{}{\mM} \in \C^{sn_1\times n_2}$ satisfies
	\begin{align}
	\label{eq: w cond 1}
	&\calA\calGT(\msh{}{\mM})=0,\\
	\label{eq: w cond 2}
	&(\calI-\calG\calGT)(\msh{}{\mM})=0.
	\end{align}
	The first condition \eqref{eq: w cond 1} implies that $\calGT(\msh{}{\mM})$ is in the null space of $\calA$, while the second condition \eqref{eq: w cond 2} guarantees that $\msh{}{\mM}$ has the vectorized Hankel structure. Note that for any matrix $\msh{}{\mM}$, there exists \msh{}{an $sn_1\times n_2$} matrix $\mS \in T^\perp$ such that
	\begin{align*}
	\la \msh{}{\mM},\mS \ra = \nucnorm{\calP_{T^\perp}(\msh{}{\mM})}\quad\mbox{and}\quad\opnorm{\mS}\leq 1.
	\end{align*}
	\kw{}{In the meantime}, we have $\mU\mV^{\ast}+\mS\in\partial\nucnorm{\mZ^\natural}$. Thus, 
	\begin{align*}
	\Delta:&=\nucnorm{\mZ^\natural+\msh{}{\mM}}-\nucnorm{\mZ^\natural}\\
	&\geq \la\mU\mVT + \mS,\msh{}{\mM}\ra\\
	&= \la\mU\mVT,\msh{}{\mM}\ra + \nucnorm{\calP_{T^\perp}(\msh{}{\mM})}\\
	&\geq\nucnorm{\calP_{T^\perp}(\msh{}{\mM})}-\lab\la\mU\mVT-\bLambda,\msh{}{\mM}\ra\rab-\lab\la\bLambda,\msh{}{\mM}\ra\rab.\numberthis\label{eq: three parts}
	\end{align*}
	The condition \eqref{belong} directly implies that there exists a vector $\vp\in\C^{n}$ such that 
	\begin{align*}
	\calGT(\bLambda) = \calAT(\vp).
	\end{align*}
	Therefore, combining \eqref{belong} and \eqref{eq: w cond 2}, we obtain 
	\begin{align*}
	\lab\la \bLambda,\msh{}{\mM} \ra\rab=\lab\la\bLambda,\calG\calGT(\msh{}{\mM})\ra\rab=\lab\la\calGT(\bLambda),\calGT(\msh{}{\mM})\ra\rab=\lab\la\calAT(\vp),\calGT(\msh{}{\mM})\ra\rab=\la\vp,\calA\calGT(\msh{}{\mM})\ra=0.
	\end{align*}
	Moreover, \kw{}{the second term of \eqref{eq: three parts} can be upper bounded as follows:}
	\begin{align*}
	\lab\la\mU\mVT-\bLambda,\msh{}{\mM}\ra\rab &\leq\lab\la\calP_T(\mU\mVT-\bLambda),\msh{}{\mM}\ra\rab+\lab\la\calP_{T^\perp}(\mU\mVT-\bLambda),\msh{}{\mM}\ra\rab\\
	&\leq\fronorm{\calP_T(\mU\mVT-\bLambda)}\cdot\fronorm{\calP_T(\msh{}{\mM})}+\opnorm{\calP_{T^\perp}(\bLambda)}\cdot\nucnorm{\calP_{T^\perp}(\msh{}{\mM})}\\
	&\leq\frac{1}{16s\mu_0}\cdot\fronorm{\calP_T(\msh{}{\mM})}+\frac{1}{2}\cdot\nucnorm{\calP_{T^\perp}(\msh{}{\mM})},
	\end{align*}
	where the last step is due to \eqref{ineq: F norm} and \eqref{ineq: perp op norm}. Consequently,
	\begin{align*}
	\Delta &\geq\nucnorm{\calP_{T^\perp}(\msh{}{\mM})}-\lab\la\mU\mVT-\bLambda,\msh{}{\mM}\ra\rab-\lab\la\bLambda,\msh{}{\mM}\ra\rab\\
	&\geq\frac{1}{2}\cdot\nucnorm{\calP_{T^\perp}(\msh{}{\mM})}-\frac{1}{16s\mu_0}\cdot\fronorm{\calP_T(\msh{}{\mM})}\\
	&\geq\frac{1}{2}\cdot\fronorm{\calP_{T^\perp}(\msh{}{\mM})}-\frac{1}{16s\mu_0}\cdot\fronorm{\calP_T(\msh{}{\mM})}\\
	&\geq\lb\frac{1}{2}-\frac{1}{16s\mu_0}\cdot4s\mu_0\rb\fronorm{\calP_{T^\perp}(\msh{}{\mM})}\\
	&=\frac{1}{4}\fronorm{\calP_{T^\perp}(\msh{}{\mM})},
	\end{align*}
	where the fourth line is due to Lemma \ref{lem:upper bound} in Section \ref{aux}. It follows that $\Delta>0$ unless $\fronorm{\calP_{T^\perp}(\msh{}{\mM})}=0$.

	\kw{}{Note that $\Delta=0$ requires $\calP_{T^\perp}(\msh{}{\mM})=0$, which in turn requires $\msh{}{\mM}=\calP_T(\msh{}{\mM})$.}  In this case, we have 
	\begin{align*}
	\fronorm{\calP_T(\msh{}{\mM})}^2 &=\la\calP_T(\msh{}{\mM}),\msh{}{\mM}\ra\\
	&=\la\calP_T(\msh{}{\mM}),\calG\calGT(\msh{}{\mM})\ra\\
	&=\la\msh{}{\mM},\calP_T\calG\calGT\calP_T(\msh{}{\mM})-\calP_T\calG\calAT\calA\calGT\calP_T(\msh{}{\mM})\ra+\la\msh{}{\mM},\calP_T\calG\calAT\calA\calGT\calP_T(\msh{}{\mM})\ra\\
	&=\la\msh{}{\mM},\calP_T\calG\calGT\calP_T(\msh{}{\mM})-\calP_T\calG\calAT\calA\calGT\calP_T(\msh{}{\mM})\ra+\la\msh{}{\mM},\calP_T\calG\calAT\calA\calGT(\msh{}{\mM})\ra\\
	&=\la\msh{}{\mM},\calP_T\calG\calGT\calP_T(\msh{}{\mM})-\calP_T\calG\calAT\calA\calGT\calP_T(\msh{}{\mM})\ra\\
	&\leq \opnorm{\calP_{T}\calG\calAT \calA\calGT\calP_{T} -\calP_{T}\calG\calGT\calP_{T} }\cdot \fronorm{\calP_T(\msh{}{\mM})}^2\\
	&\leq\frac{1}{2}\fronorm{\calP_T(\msh{}{\mM})}^2,
	\end{align*}
	which implies that $\calP_T(\msh{}{\mM})=\bzero$. Thus $\mZ^\natural$ is the unique minimizer.
\end{proof}
\subsection{Constructing the dual certificate}


\kw{}{It is intuitively clear that} we may construct a dual certificate $\bLambda \in \C^{sn_1\times n_2}$ obeying the conditions \eqref{ineq: F norm}, \eqref{ineq: perp op norm} and \eqref{belong} 
by solving the following constrained least squares problem:
\begin{align*}
\min_{\bLambda}~\fronorm{\calP_{T}(\mU\mVT-\bLambda)}^2~\mbox{s.t.}~\calGT(\bLambda)\in\rm{Range}(\calAT).
\end{align*}
Here only the conditions  \eqref{ineq: F norm} and \eqref{belong} are taken into account because once $\fronorm{\calP_{T}(\mU\mVT-\bLambda)}$ is \kw{}{small}, the projection of $\bLambda$ onto $T^{\perp}$ can be simultaneously small.

\kw{}{Applying} the projected gradient method to solve the above optimization problem, we obtain \kw{}{the following update rule}:
\begin{align*}
\mY^k = \mY^{k-1}+\lb\calG\calAT\calA\calGT+\calI-\calG\calGT\rb\calP_T(\mU\mVT-\mY^{k-1}).
\end{align*}
However, 
\kw{}{due to} the statistical dependence among the iterations, the convergence analysis of the vanilla gradient iteration is difficult. Therefore, \kw{}{the golfing scheme \cite{gross2011recovering} proposes to break the statistical independence by dividing all the linear measurements into a few disjoint partitions and use a fresh partition in each iteration.}

Assume we divide the linear measurements in \eqref{eq: measurements} into $k_0$ partitions, \kw{no by here }{denoted $\{\Omega_k\}_{k=1}^{k_0}$}, and let $m=\frac{n}{k_0}$. Define
\begin{align*}
&\calA_k(\mX)=\lcb\la \vb_{i}\ve_{i}^\tran,\mX\ra\rcb_{i\in\Omega_k} \in \C^{\lab\Omega_k\rab}\numberthis\label{eq:partA}
\end{align*}
\msh{}{and
\begin{align*}
    \calA_k^\ast\calA_k(\mX) = \sum_{i\in\Omega_k}\la\vb_i\ve_i^\tran,\mX\ra\vb_i\ve_i^\tran
    = \sum_{i\in\Omega_k}\vb_i\vb_i^\tran\mX\ve_i\ve_i^\tran \in\C^{s\times n}.\numberthis\label{eq:AktAk}
\end{align*}
}
Then the golfing scheme for the construction of $\bLambda$ satisfying the conditions in Theorem~\ref{thr:optimality} can be formally expressed as
\begin{align*}
&\mY^0=\bzero \in \C^{sn_1\times n_2},\\
&\mY^k=\mY^{k-1}+\lb\frac{n}{m}\calG\calAT_k\calA_k\calGT+\calI-\calG\calGT\rb\calP_T(\mU\mVT-\mY^{k-1}),\quad\text{ for }k=1,\cdots, k_0,\numberthis\label{eq: construct dual}\\
&\bLambda:=\mY^{k_0}.
\end{align*}

\kw{}{Evidently the property of $\bLambda$ relies on the partitions $\{\Omega_k\}_{k=1}^{k_0}$. In order to construct the desirable $\bLambda$, we require   $\{\Omega_k\}_{k=1}^{k_0}$ to satisfy a set of conditions list in the following lemma, in which we have }
\begin{align}
\label{eq: norm definition}
\left\|\mZ\right\|_{\calG,\mathsf{F}}=\sqrt{\sum_{i=0}^{n-1}\frac{\vecnorm{\calGT(\mZ)e_{i}}{2}^2}{w_i}}\quad\mbox{and}\quad
\left\|\mZ\right\|_{\calG,\infty}=\max_{0\leq i\leq n-1}\frac{\vecnorm{\calGT(\mZ)e_{i}}{2}}{\sqrt{w_i}}\quad \mbox{for any }\mZ\in\C^{sn_1\times n_2}.
\end{align}
The proof of this lemma will be presented in Section~\ref{sec: proof of partition}.


\begin{lemma}
	\label{lem: partition}
	Let $k_0\in \{1,\cdots,n\}$ and set $m=\frac{n}{k_0}$. If $n\gtrsim k_0\cdot\max\{\mu_1r\log(sn),\log(k_0)\}$, then there exists a partition $\{\Omega_k\}_{k=1}^{k_0}$ such that the following properties hold :
	\begin{align}
	&\quad\frac{m}{2}\leq\lab\Omega_k\rab\leq\frac{3m}{2},\quad k=1,\cdots,k_0,\numberthis\label{ineq: range of omegak}\\
	&\max_{1\leq k\leq k_0}\opnorm{\calP_T\calG\lb\calI-\frac{n}{m}\mathbb E\lsb\calAT_k\calA_k\rsb\rb\calGT\calP_T}\leq \frac{1}{4},\numberthis\label{ineq: part_prop1}\\
	&\max_{1\leq k\leq k_0}\opnorm{\calG\lb\calI-\frac{n}{m}\mathbb E\lsb\calAT_k\calA_k\rsb\rb\calGT(\mZ)}\lesssim \lb \sqrt{\frac{n\log(sn)}{m}}\left\|\mZ\right\|_{\calG,\mathsf{F}}+\frac{n\log(sn)}{m}\left\|\mZ\right\|_{\calG,\infty}\rb,\numberthis\label{ineq: part_prop2}\\
	&\max_{1\leq k\leq k_0}\left\|\calP_T\calG\lb\calI-\frac{n}{m}\mathbb E\lsb\calAT_k\calA_k\rsb\rb\calGT(\mZ)\right\|_{\calG,\mathsf{F}} \lesssim \sqrt{\frac{\mu_1r\log(sn)}{n}}\lb\sqrt{\frac{n\log(sn)}{m}}\left\|\mZ\right\|_{\calG,\mathsf{F}}+\frac{n\log(sn)}{m}\left\|\mZ\right\|_{\calG,\infty}\rb,\numberthis\label{ineq: part_prop3}\\
	&\max_{1\leq k\leq k_0}\left\|\calP_T\calG\lb\calI-\frac{n}{m}\mathbb E\lsb\calAT_k\calA_k\rsb\rb\calGT(\mZ)\right\|_{\calG,\infty} \lesssim \frac{\mu_1r}{n}\lb\sqrt{\frac{n\log(sn)}{m}}\left\|\mZ\right\|_{\calG,\mathsf{F}}+\frac{n\log(sn)}{m}\left\|\mZ\right\|_{\calG,\infty}\rb\numberthis\label{ineq: part_prop4}.
	\end{align}
	Here \kw{}{$\mZ\in\C^{sn_1\times n_2}$} is fixed. \msh{}{Recalling the definition of the operator $\calA_k^\ast\calA_k$ in \eqref{eq:AktAk}
	}, the expectation is taken with respect to $\{\vb_{i}\}_{i\in\Omega_k}$.
\end{lemma}

\subsection{Validating the dual certificate and completing the proof}
In this section we show that the dual certificate $\bLambda$ constructed from the iteration \eqref{eq: construct dual} satisfies the conditions in Theorem \ref{thr:optimality}. 
The result follows from several lemmas that will be proved in Section \ref{sec: proof of part lemma}. \kw{}{In these lemmas, $\{\Omega_k\}_{k=1}^{k_0}$ is a partition of $\{1,\cdots,n\}$ satisfying the conditions  in Lemma \ref{lem: partition}, and $\{\calA_k\}_{k=1}^{k_0}$ are the associated linear operators defined in \eqref{eq:partA}. Note that we  assume {\eqref{eq: incoherence}} holds in the remainder of this paper, which follows from Assumption~\ref{assumption 2} and Lemma~\ref{incoherence 2}.} 
\begin{lemma}
	\label{lem: key lemma1}
	Assume  $n\gtrsim k_0 s\mu_0\cdot \mu_1r\log(sn)$. \msh{}{Under the condition (III.18) of Lemma \ref{lem: partition},} the event
	\begin{align}
	\max_{1\leq k \leq k_0}\opnorm{\calP_T\calG\lb\calI-\frac{n}{m}\calAT_k\calA_k\rb\calGT\calP_T} \leq \frac{1}{2}
	\end{align}
	occurs with \msh{}{probability at least $1-(sn)^{-c_1}$ for a universal constant $c_1>0$.}
\end{lemma}

The following corollary is the special case of Lemma \ref{lem: key lemma1} when  $k_0=1$ and $n=m$.
\begin{corollary}
	\label{cor: key corollary}
	Assume  $n\gtrsim s\mu_0\cdot \mu_1r\log(sn)$. The event
	\begin{align}
	\opnorm{\calP_T\calG\calAT\calA\calGT\calP_T-\calP_T\calG\calGT\calP_T}\leq \frac{1}{2}
	\end{align}
	occurs with \msh{}{probability at least $1-(sn)^{-c_1}$ for a universal constant $c_1>0$.}
\end{corollary}

\begin{lemma}
	\label{lem: key lemma2} 
	\msh{}{Under the condition (III.19) of Lemma \ref{lem: partition},} for any $1\leq k \leq k_0$ and \kw{}{fixed }$\mZ\in\C^{sn_1\times n_2}$, the event
	\begin{align}
	\label{ineq: part 2}
	\opnorm{\lb\frac{n}{m}\calG\calAT_k\calA_k\calGT-\calG\calGT\rb(\mZ)} \lesssim \sqrt{\frac{4nk_0s\mu_0\log(sn)}{m}}\gfronorm{\mZ}+\frac{2ns\mu_0\log(sn)}{m}\ginfnorm{\mZ}
	\end{align}
	occurs with \msh{}{probability at least $1-(sn)^{-c_1}$ for a universal constant $c_1>0$.}
\end{lemma}

\begin{lemma}
	\label{lem: key lemma3} 
	\msh{}{Under the condition (III.20) of Lemma \ref{lem: partition},} for any $1\leq k \leq k_0$ and \kw{}{fixed} $\mZ\in\C^{sn_1\times n_2}$, 
	the event
	\begin{align}
	\label{ineq: part3}
	\gfronorm{\calP_T\calG\lb\calI-\frac{n}{m}\calAT_k\calA_k\rb\calGT(\mZ)}\lesssim\sqrt{\frac{\mu_1r\log(sn)}{n}}\lb\sqrt{\frac{4nk_0s\mu_0\log(sn)}{m}}\gfronorm{\mZ}+\frac{2ns\mu_0\log(sn)}{m}\ginfnorm{\mZ}\rb
	\end{align}
	occurs with \msh{}{probability at least $1-(sn)^{-c_1}$ for a universal constant $c_1>0$.}    
\end{lemma}

\begin{lemma}
	\label{lem: key lemma4}
	\msh{}{Under the condition (III.21) of Lemma \ref{lem: partition},} for any $1\leq k \leq k_0$ and \kw{}{fixed} $\mZ\in\C^{sn_1\times n_2}$, the event
	\begin{align}
	\ginfnorm{\calP_T\calG\lb\calI-\frac{n}{m}\calAT_k\calA_k\rb\calGT(\mZ)}\lesssim \frac{\mu_1 r}{n}\lb\sqrt{\frac{4nk_0s\mu_0\log(sn)}{m}}\gfronorm{\mZ}+\frac{2ns\mu_0 \log(sn)}{m}\ginfnorm{\mZ}\rb
	\end{align}
	occurs with \msh{}{probability at least $1-ns^{-c_2}$ for a numerical constant $c_2>2$.} 
\end{lemma}

\begin{lemma}
	\label{lem: key lemma5}
	\kw{}{Recalling that $\mU$ and $\mV$ satisfy {\eqref{eq: incoherence}}}, we have
	\begin{align}
	\gfronorm{\mU\mVT}^2\lesssim \frac{\mu_1r\log(sn)}{n}\quad\mbox{and}\quad\ginfnorm{\mU\mVT}\leq \frac{\mu_1r}{n}.
	\end{align}
\end{lemma}

Equipped with these lemmas, \kw{}{we are in position to} validate the conditions in Theorem \ref{thr:optimality}. 
\msh{}{Note that $\opnorm{\calA\calA^\ast}\geq 1$ holds due to \eqref{ineq: lower bound of ATA} in Lemma~\ref{lem: prop of A},} and  \eqref{ineq: RIP}
	is proved in Corollary~\ref{cor: key corollary}.
	As for \eqref{belong}, it follows immediately from the construction of $\bLambda$. Thus, it remains to validate 
	\eqref{ineq: F norm} and \eqref{ineq: perp op norm}.

\paragraph{Validating \eqref{ineq: F norm}}
A simple calculation yields that 
\begin{align*}
\mE^k :&= \calP_T\lb\mU\mVT - \mY^k\rb\\
&= \calP_T\lb\mU\mVT-\mY^{k-1}-\lb\frac{n}{m}\calG\calAT_k\calA_k\calGT+\calI-\calG\calGT\rb\calP_T(\mE^{k-1})\rb\\
&= \calP_T(\mE^{k-1}) - \calP_T\lb\frac{n}{m}\calG\calAT_k\calA_k\calGT+\calI-\calG\calGT\rb\calP_T(\mE^{k-1})\\
&= \calP_T\lb\calG\calGT-\frac{n}{m}\calG\calAT_k\calA_k\calGT\rb\calP_T(\mE^{k-1}),\numberthis\label{eq:iteration of Ek}
\end{align*}
\msh{}{where the second line is due to \eqref{eq: construct dual}}.
By the construction of $\bLambda$, we can obtain 
\begin{align*}
\fronorm{\calP_T\lb\mU\mVT-\bLambda\rb} &= \fronorm{\mE^{k_0}}\\
&= \fronorm{\calP_T\lb\calG\calGT-\frac{n}{m}\calG\calAT_{k_0}\calA_{k_0}\calGT\rb\calP_T(\mE^{k_0-1})}\\
&\leq \opnorm{\calP_T\lb\calG\calGT-\frac{n}{m}\calG\calAT_{k_0}\calA_{k_0}\calGT\rb\calP_T}\cdot\fronorm{\mE^{k_0-1}}\\
&\overset{(a)}{\leq} \frac{1}{2}\fronorm{\mE^{k_0-1}} \leq \frac{1}{2^{k_0}}\fronorm{\mE^0}\\
&= \frac{1}{2^{k_0}}\fronorm{\mU\mVT} \leq \frac{r}{2^{k_0}}\\
&\leq \frac{1}{16s\mu_0},
\end{align*}
where step $(a)$ is due to Lemma \ref{lem: key lemma1} and the last inequality holds when $k_0=\lceil \log_2(16rs\mu_0)\rceil$.

\paragraph{Validating \eqref{ineq: perp op norm}}
\msh{}{First recall that $\mE^k := \calP_T\lb\mU\mVT - \mY^k\rb$. According to \eqref{eq: construct dual}, we have 
\begin{align*}
	\bLambda = \sum_{k=1}^{k_0}\lb\frac{n}{m}\calG\calAT_k\calA_k\calGT + \calI - \calG\calGT\rb(\mE^{k-1}).
\end{align*}
Then it follows that
} 
\begin{align*}
\opnorm{\calP_{T^{\perp}}(\bLambda)} &= \opnorm{\calP_{T^{\perp}}\lb\sum_{k=1}^{k_0}\lb\frac{n}{m}\calG\calAT_k\calA_k\calGT + \calI - \calG\calGT\rb(\mE^{k-1})\rb}\\
&= \opnorm{\calP_{T^{\perp}}\lb\sum_{k=1}^{k_0}\lb\frac{n}{m}\calG\calAT_k\calA_k\calGT - \calG\calGT\rb(\mE^{k-1})\rb}\\
&\leq \sum_{k=1}^{k_0}\opnorm{\lb\frac{n}{m}\calG\calAT_k\calA_k\calGT-\calG\calGT\rb(\mE^{k-1})},\numberthis\label{ineq: PTLambda}
\end{align*}
where the second line follows from the fact that $\mE^{k-1}\in T$.

For any $1\leq k \leq k_0$, Lemma \ref{lem: key lemma2} implies that 
\begin{align}
\label{ineq: vali4.3_1}
\opnorm{\lb\frac{n}{m}\calG\calAT_k\calA_k\calGT-\calG\calGT\rb(\mE^{k-1})} \lesssim \sqrt{\frac{4n k_0s\mu_0\log(s n)}{m}}\gfronorm{\mE^{k-1}} + \frac{2n s\mu_0\log(s n)}{m}\ginfnorm{\mE^{k-1}}.
\end{align}
\msh{}{Recalling from the equality \eqref{eq:iteration of Ek}, we have} 
\begin{align*}
\mE^{k-1} = \calP_T\lb\calG\calGT-\frac{n}{m}\calG\calAT_{k-1}\calA_{k-1}\calGT\rb\calP_T(\mE^{k-2}).
\end{align*}
Applying Lemma \ref{lem: key lemma3} and Lemma \ref{lem: key lemma4} yields that
\begin{align*}
\gfronorm{\mE^{k-1}} &= \msh{}{\gfronorm{\calP_T\calG\lb\calI-\frac{n}{m}\calAT_{k-1}\calA_{k-1}\rb\calGT\calP_T(\mE^{k-2})}}\\ &=\gfronorm{\calP_T\calG\lb\calI-\frac{n}{m}\calAT_{k-1}\calA_{k-1}\rb\calGT(\mE^{k-2})}\\
&\lesssim \sqrt{\frac{\mu_1r\log(s n)}{n}}\lb\sqrt{\frac{4n k_0s\mu_0\log(s n)}{m}}\gfronorm{\mE^{k-2}}+\frac{2n s\mu_0\log(s n)}{m}\ginfnorm{\mE^{k-2}}\rb\numberthis\label{ineq: recursive1}
\end{align*}
and
\begin{align*}
\ginfnorm{\mE^{k-1}} &= \msh{}{\ginfnorm{\calP_T\calG\lb\calI - \frac{n}{m}\calAT_{k-1}\calA_{k-1}\rb\calGT\calP_T(\mE^{k-2})}}\\
&= \ginfnorm{\calP_T\calG\lb\calI - \frac{n}{m}\calAT_{k-1}\calA_{k-1}\rb\calGT(\mE^{k-2})}\\
&\lesssim \frac{\mu_1 r}{n}\lb\sqrt{\frac{4n k_0s\mu_0\log(s n)}{m}}\gfronorm{\mE^{k-2}} + \frac{2n s\mu_0 \log(s n)}{m}\ginfnorm{\mE^{k-2}}\rb.\numberthis\label{ineq: recursive2}
\end{align*}
After substituting \eqref{ineq: recursive1} and \eqref{ineq: recursive2} into \eqref{ineq: vali4.3_1}, we have 
\begin{align*}
\opnorm{\lb\frac{n}{m}\calG\calAT_k\calA_k\calGT-\calG\calGT\rb(\mE^{k-1})} &\lesssim \sqrt{\frac{4n k_0s\mu_0\log(s n)}{m}}\gfronorm{\mE^{k-1}} + \frac{2n s\mu_0\log(s n)}{m}\ginfnorm{\mE^{k-1}}\\
&\lesssim \lb\sqrt{\frac{4n k_0s\mu_0\log(s n)}{m}}\cdot\sqrt{\frac{\mu_1r\log(s n)}{n}} + \frac{2n s\mu_0\log(s n)}{m}\cdot\frac{\mu_1r}{n}\rb\\
&\quad \cdot\lb\sqrt{\frac{4n k_0s\mu_0\log(s n)}{m}}\gfronorm{\mE^{k-2}} + \frac{2n s\mu_0\log(s n)}{m}\ginfnorm{\mE^{k-2}}\rb\\
&= \lb\sqrt{\frac{4k_0s\mu_0\mu_1r\log^2(s n)}{m}} + \frac{2s\mu_0\mu_1r\log(s n)}{m}\rb\\
&\quad \cdot\lb\sqrt{\frac{4n k_0s\mu_0\log(s n)}{m}}\gfronorm{\mE^{k-2}} + \frac{2n s\mu_0\log(s n)}{m}\ginfnorm{\mE^{k-2}}\rb\\
&\overset{(a)}{\leq} \frac{1}{2}\lb\sqrt{\frac{4n k_0s\mu_0\log(s n)}{m}}\gfronorm{\mE^{k-2}} + \frac{2n s\mu_0\log(s n)}{m}\ginfnorm{\mE^{k-2}}\rb\\
&\leq \lb\frac{1}{2}\rb^{k-1}\cdot\lb\sqrt{\frac{4n k_0s\mu_0\log(s n)}{m}}\gfronorm{\mE^{0}} + \frac{2n s\mu_0\log(s n)}{m}\ginfnorm{\mE^{0}}\rb,
\end{align*}
where step $(a)$ holds provided $m \gtrsim k_0s\mu_0\mu_1r\log^2(s n)$.

Finally, noting that $\mE^0 = \mU\mVT$, the application of Lemma \ref{lem: key lemma5} gives
\begin{align*}
\opnorm{\calP_{T^{\perp}}(\bLambda)} &\msh{}{\leq \sum_{k=1}^{k_0}\opnorm{\lb\frac{n}{m}\calG\calAT_k\calA_k\calGT-\calG\calGT\rb(\mE^{k-1})}}\\
&\leq \sum_{k=1}^{k_0} \lb\frac{1}{2}\rb^{k-1}\cdot\lb\sqrt{\frac{4n k_0s\mu_0\log(s n)}{m}}\gfronorm{\mE^{0}} + \frac{2n s\mu_0\log(s n)}{m}\ginfnorm{\mE^{0}}\rb\\
&\lesssim \frac{1}{2}\lb\sqrt{\frac{4 k_0s\mu_0\mu_1r\log^2(s n)}{m}} + \frac{2 s\mu_0\mu_1r\log(s n)}{m}\rb\\
&\leq \frac{1}{2}
\end{align*}
\msh{}{when $m\gtrsim k_0s\mu_0\mu_1r\log^2(s n)$, where the first inequality follows from \eqref{ineq: PTLambda}.
}

\cjc{
Thus we have shown that the dual certificate $\bLambda$ constructed from the iteration \eqref{eq: construct dual} satisfies the conditions in Theorem \ref{thr:optimality} with probability at least $1-c_0(sn)^{-c_1} - ns^{-c_2}$ provided that $n=m k_0\gtrsim \mu_0 \mu_1\cdot sr\log^4(sn)$.   
Corollary~\ref{cor: key corollary} implies \eqref{ineq: RIP} holds with probability at least $1-(sn)^{-c_3}$ if $n\gtrsim \mu_0 \mu_1\cdot sr\log(sn)$. Taking an upper bound on the number of measurements completes the proof of  Theorem \ref{main result}.
}
\section{Proof of Lemma \ref{lem: partition}}
\label{sec: proof of partition}
In this section, we will use probabilistic \kw{}{argument} to show that the events \eqref{ineq: range of omegak} - \eqref{ineq: part_prop4} occur with high probability if we construct $\{\Omega_k\}_{k=1}^{k_0}$ \kw{}{in a random manner} and thus conclude that there at least exists a partition  satisfying \eqref{ineq: range of omegak} - \eqref{ineq: part_prop4}.

Let $\{\epsilon_i\}_{i=0}^{n-1}$ be $n$ independent random variables, each of which takes value in $\{1,\cdots,k_0\}$ \kw{}{uniformly at random}. For any $k\in\{1,\cdots,k_0\}$, we construct $\{\Omega_k\}_{k=1}^{k_0}$ as follows: 
\begin{align*}
\Omega_k=\{i\in[n]:\epsilon_i=k\}.
\end{align*}
Clearly, $\{\Omega_k\}_{k=1}^{k_0}$ form a partition of $[n]$. For any fixed $k\in\{1,\cdots,k_0\}$, we also have
\begin{align*}
\P\lcb i\in\Omega_k\rcb = \P\lcb\epsilon_i=k\rcb = \frac{1}{k_0}\quad\text{for all}~ i=0,\cdots,n-1.
\end{align*}
Therefore $\lab\Omega_k\rab$ can be viewed as the sum of Bernoulli random variables, i.e.,
\begin{align}
\label{eq:delta}
\lab\Omega_k\rab = \sum_{i=0}^{n-1}\boldsymbol{1}\{i \in \Omega_k\} =: \sum_{i=0}^{n-1}\delta_i,
\end{align}
where $\{\delta_i\}_{i=0}^{n-1}$ are i.i.d. Bernoulli random variables with parameter $p=\frac{1}{k_0}=\frac{m}{n}$. The application of the Hoeffding inequality yields that $\frac{m}{2}\leq\lab\Omega_k\rab\leq\frac{3m}{2}$ holds with  probability at least $1-2\exp(-c m)$ for a universal constant $c > 0$. Then we can take the uniform bound to obtain 
\begin{align*}
\P\lcb\frac{m}{2}\leq\lab\Omega_k\rab\leq\frac{3m}{2}~\mbox{for all}~ k\rcb \geq 1-2k_0\exp(-c m)\geq\frac{1}{2},
\end{align*}
where the last inequality is due to $m=\frac{n}{k_0}\gtrsim\log(k_0)$.

\kw{}{Our next goal is to show} that the events \eqref{ineq: part_prop1} - \eqref{ineq: part_prop4} occur with high probability. We will  first apply the matrix Bernstein inequality \eqref{ineq: bernstein} to obtain the desired upper bounds for fixed $k$, and then take the uniform bound analysis to complete the proof. 

\subsection{Proof of \eqref{ineq: part_prop1}}

\kw{}{For any $\mZ\in \C^{s n_1\times n_2}$, by the definition of \msh{}{$\calA_k^\ast\calA_k$ in \eqref{eq:AktAk}}, we have}
\begin{align*}
\E{\calAT_{k}\calA_{k}}\calGT\calP_T(\mZ) &=\E{\sum_{i\in\Omega_k}\la\vb_{i}\ve_{i}^\tran,\calGT\calP_T(\mZ)\ra\vb_{i}\ve_{i}^\tran}\\
&=\sum_{i\in\Omega_k}\E{\vb_{i}\vb_{i}^{\ast}}\calGT\calP_T(\mZ)\ve_{i}\ve_{i}^\tran\\
&=\calGT\calP_T(\mZ)\sum_{i\in\Omega_k}\ve_{i}\ve_{i}^\tran,
\end{align*}
where the third line follows from the isotropy property \eqref{eq: isotropy} of $\{\vb_i\}$. 

As a result, one has the following equality
\begin{align*}
\opnorm{\calP_T\calG\lb\calI-\frac{1}{p}\E{\calAT_{k}\calA_{k}}\rb\calGT\calP_T} &=\sup_{\fronorm{\mW}=1}\fronorm{\calP_T\calG\lb\calI-\frac{1}{p}\E{\calAT_{k}\calA_{k}}\rb\calGT\calP_T(\mW)}\\
&=\sup_{\fronorm{\mW}=1}\fronorm{\frac{1}{p}\calP_T\calG\calGT\calP_T(\mW)\sum_{i\in\Omega_k}\ve_i\ve_i^\tran-\calP_T\calG\calGT\calP_T(\mW)}\\
&= \sup_{\fronorm{\mW}=1}\fronorm{\sum_{i=0}^{n-1}\lb\frac{\delta_i}{p}-1\rb\calP_T\calG\lb\calGT\calP_T(\mW)\ve_i\ve_i^\tran\rb}\\
&=: \opnorm{\sum_{i=0}^{n-1}\lb\frac{\delta_i}{p}-1\rb\calX_i},
\end{align*}
where \msh{}{$\delta_i$ is the Bernoulli random variable defined in \eqref{eq:delta} and}
$\calX_i$ is the operator defined as
\begin{align*}
\calX_i(\mW)=\calP_T\calG\lb\calGT\calP_T(\mW)\ve_i\ve_i^\tran\rb
\end{align*}
\msh{}{for any $\mW\in\C^{sn_1\times n_2}$.}
It is easy to verify that $\calX_i$ is self-adjoint and positive semi-definite. 

In order to apply the matrix Bernstein inequality \eqref{ineq: bernstein} to bound $\opnorm{\sum_{i=0}^{n-1}\lb\frac{\delta_i}{p}-1\rb\calX_i}$, one needs to bound $\opnorm{\lb\frac{\delta_i}{p}-1\rb\calX_i}$ and $\opnorm{\E{\sum_{i=0}^{n-1}\lb\frac{\delta_i}{p}-1\rb^2\calX_i^2}}$.  

For the upper bound of $\opnorm{\lb\frac{\delta_i}{p}-1\rb\calX_i}$, a simple calculation yields that
\begin{align*}
\opnorm{\lb\frac{\delta_i}{p}-1\rb\calX_i} &\leq \frac{1}{p}\opnorm{\calX_i}\\
&=\frac{1}{p}\sup_{\fronorm{\mW}=1}\fronorm{\calP_T\calG\lb\calGT\calP_T(\mW)\ve_i\ve_i^\tran\rb}\\
&\leq \frac{1}{p}\sup_{\fronorm{\mW}=1}\fronorm{\mW}\cdot\frac{2\mu_1r}{n}\\
&=\frac{2\mu_1r}{n p},\numberthis\label{ineq: calX_i}
\end{align*}
where 
the third line follows from Corollary \ref{cor: useful cor1}. 

To bound $\opnorm{\E{\sum_{i=0}^{n-1}\lb\frac{\delta_i}{p}-1\rb^2\calX_i^2}}$, we have
\begin{align*}
\opnorm{\sum_{i=0}^{n-1}\E{\lb\frac{\delta_i}{p}-1\rb^2\calX_i^2}} &\leq\frac{1}{p}\opnorm{\sum_{i=0}^{n-1}\calX_i^2}\\
&\leq \frac{1}{p}\max_{0\leq i\leq n-1}\opnorm{\calX_i}\cdot\opnorm{\sum_{i=0}^{n-1}\calX_i}\\
&\leq \frac{2\mu_1r}{n p}\sup_{\fronorm{\mW}=1}\opnorm{\sum_{i=0}^{n-1}\calX_i(\mW)}\\
&= \frac{2\mu_1r}{n p}\sup_{\fronorm{\mW}=1}\fronorm{\sum_{i=0}^{n-1}\calP_T\calG\lb\calGT\calP_T(\mW)\ve_i\ve_i^\tran\rb}\\
&= \frac{2\mu_1r}{n p}\sup_{\fronorm{\mW}=1}\fronorm{\calP_T\calG\calGT\calP_T(\mW)}\\
&= \frac{2\mu_1r}{n p}\opnorm{\calP_T\calG\calGT\calP_T}\\
&\leq \frac{2\mu_1r}{n p},
\end{align*}
where the second line is due to the positive semi-definite property of $\calX_i$, \msh{}{the third line follows from \eqref{ineq: calX_i},} and the last line follows from the fact that $\opnorm{\calG}=1$, $\opnorm{\calGT}\leq 1$ \msh{}{and $\calP_T$ is the projection operator}.

\kw{}{The application of} the matrix Bernstein inequality implies that
\begin{align*}
\opnorm{\calP_T\calG\lb\calI-\frac{1}{p}\E{\calAT_{k}\calA_{k}}\rb\calGT\calP_T} &\lesssim \sqrt{\frac{\mu_1r\log(s n)}{n p}}+\frac{\mu_1r\log(s n)}{n p}\\
&\lesssim \sqrt{\frac{\mu_1r\log(s n)}{n p}}\\
&\leq \frac{1}{4}
\end{align*}
holds with  probability at least \msh{}{$1-(sn)^{-c}$ for a universal constant $c>0$}, where the second and third lines are due to $p\gtrsim\frac{\mu_1r\log(s n)}{n}$. Finally, we take the uniform bound to obtain that 
\begin{align*}
\P\lcb\max_{1\leq k\leq k_0}\opnorm{\calP_T\calG\lb\calI-\frac{n}{m}\mathbb E\lsb\calAT_k\calA_k\rsb\rb\calGT\calP_T}\leq\frac{1}{4}\rcb\geq \msh{}{1-k_0(sn)^\msh{}{{-c}}\geq 1-(sn)^\msh{}{{-(c-1)}}},
\end{align*}
where the last inequality follows from the fact that $k_0\ll sn$.

\subsection{Proof of \eqref{ineq: part_prop2}}
\msh{}{Following the definition of $\calA_k^\ast\calA_k$ in \eqref{eq:AktAk} and the isotropy property of $\{\vb_i\}$ in \eqref{eq: isotropy}, we have}
\begin{align*}
\opnorm{\calG\lb\calI-\frac{1}{p}\E{\calAT_{k}\calA_{k}}\rb\calGT(\mZ)} &= \opnorm{\frac{1}{p}\calG\calGT(\mZ)\sum_{i\in\Omega_k}\ve_i\ve_i^\tran-\calG\calGT(\mZ)}\\
&=\opnorm{\sum_{i=0}^{n-1}\lb\frac{\delta_i}{p}-1\rb\calG\lb\calGT(\mZ)\ve_i\ve_i^\tran\rb}\\
&=: \opnorm{\sum_{i=0}^{n-1}\mX_i},
\end{align*}
where \msh{}{$\delta_i$ is defined in \eqref{eq:delta} and $\mX_i:= \lb\frac{\delta_i}{p}-1\rb\calG\lb\calGT(\mZ)\ve_i\ve_i^\tran\rb \in\C^{sn_1\times n_2}$} are independent random matrices with zero mean.

Firstly, $\opnorm{\mX_i}$ can be bounded as follows:
\begin{align*}
\opnorm{\mX_i} &\leq\frac{1}{p}\opnorm{\calG\lb\calGT(\mZ)\ve_i\ve_i^\tran\rb}\\
&= \frac{1}{p}\opnorm{\mG_i\otimes\lb\calGT(\mZ)\ve_i\rb}\\
&\leq \frac{1}{p}\opnorm{\mG_i}\cdot\vecnorm{\calGT(\mZ)\ve_i}{2}\\
&\leq \frac{1}{p}\frac{1}{\sqrt{w_i}}\vecnorm{\calGT(\mZ)\ve_i}{2}\\
&\leq \frac{1}{p}\left\|\mZ\right\|_{\calG,\infty},
\end{align*}
where \msh{}{the second line is due to \eqref{eq: calG},} the third line follows from the fact that $\opnorm{\mA\otimes\mB}\leq \opnorm{\mA}\cdot\opnorm{\mB}$, \msh{}{and the last line directly follows from the definition of $\ginfnorm{\cdot}$ in \eqref{eq: norm definition}}.  

Secondly, we have 
\begin{align*}
\opnorm{\E{\sum_{i=0}^{n-1}\mX_i\mX_i^{\ast}}} &\msh{}{= \opnorm{\sum_{i=0}^{n-1}\E{\lb\frac{\delta_i}{p}-1\rb^2}\lb\calG\lb\calGT(\mZ)\ve_i\ve_i^\tran\rb\rb\lb\calG\lb\calGT(\mZ)\ve_i\ve_i^\tran\rb\rb^{\ast}}}\\
&= \opnorm{\sum_{i=0}^{n-1}\E{\lb\frac{\delta_i}{p}-1\rb^2}\lb\mG_i\otimes\lb\calGT(\mZ)\ve_i\rb\rb\lb\mG_i\otimes\lb\calGT(\mZ)\ve_i\rb\rb^{\ast}}\\
&\leq \frac{1}{p}\sum_{i=0}^{n-1}\opnorm{\lb\mG_i\otimes\lb\calGT(\mZ)\ve_i\rb\rb\lb\mG_i\otimes\lb\calGT(\mZ)\ve_i\rb\rb^{\ast}}\\
&\leq \frac{1}{p}\sum_{i=0}^{n-1}\opnorm{(\mG_i\mG_i^{\ast})\otimes\lb\lb\calGT(\mZ)\ve_i\rb\lb\calGT(\mZ)\ve_i\rb^{\ast}\rb}\\
&\leq \frac{1}{p}\sum_{i=0}^{n-1}\opnorm{\mG_i\mG_i^{\ast}}\cdot\vecnorm{\calGT(\mZ)\ve_i}{2}^2\\
&\msh{}{\leq\frac{1}{p}\sum_{i=0}^{n-1} \frac{1}{w_i}\vecnorm{\calGT(\mZ)\ve_i}{2}^2 }\\
&= \frac{1}{p}\left\|\mZ\right\|_{\calG,\mathsf{F}}^2.
\end{align*}
Since \opnorm{\E{\sum_{i=0}^{n-1}\mX_i^*\mX_i}} can be bounded by the same quantity, the application of the matrix Bernstein inequality \eqref{ineq: bernstein} implies that 
\begin{align*}
\opnorm{\calG\lb\calI-\frac{1}{p}\E{\calAT_{k}\calA_{k}}\rb\calGT(\mZ)}=\opnorm{\sum_{i=0}^{n-1}\mX_i} \lesssim \lb \sqrt{\frac{\log(s n)}{p}}\left\|\mZ\right\|_{\calG,\mathsf{F}}+\frac{\log(s n)}{p}\left\|\mZ\right\|_{\calG,\infty}\rb
\end{align*}
holds with  probability at least \msh{}{$1-(sn)^{-c}$ for a numerical constant $c > 0$.} 

\kw{}{By the uniform bound we conclude} that the event \eqref{ineq: part_prop2} occurs with probability at least \msh{}{$1-(sn)^{\msh{}{-(c-1)}}$.} 

\subsection{Proof of \eqref{ineq: part_prop3}}
By the definition of $\left\|\cdot\right\|_{\calG,\mathsf{F}}$ in \eqref{eq: norm definition} \msh{}{and the isotropy property of $\{\vb_i\}$ in \eqref{eq: isotropy},} it follows that  
\begin{align*}
\left\|\calP_T\calG\lb\calI-\frac{1}{p}\mathbb E\lsb\calAT_{k}\calA_{k}\rsb\rb\calGT(\mZ)\right\|_{\calG,\mathsf{F}}^2 &=\left\|\sum_{i=0}^{n-1}\lb\frac{\delta_i}{p}-1\rb\calP_T\calG\lb\calGT(\mZ)\ve_{i}\ve_{i}^\tran\rb\right\|_{\calG,\mathsf{F}}^2\\
&\msh{}{= \sum_{j=0}^{n-1}\frac{1}{w_j}\vecnorm{\calGT\lb\sum_{i=0}^{n-1}\lb\frac{\delta_i}{p}-1\rb\calP_T\calG\lb\calGT(\mZ)\ve_i\ve_i^\tran\rb\rb\ve_j}{2}^2}\\
&=\sum_{j=0}^{n-1}\frac{1}{w_j}\vecnorm{\lb\sum_{i=0}^{n-1}\lb\frac{\delta_i}{p}-1\rb\calGT\calP_T\calG\lb\calGT(\mZ)\ve_i\ve_i^\tran\rb\rb\ve_j}{2}^2.
\end{align*}
If we construct a new vector $\vz_i\in \C^{sn\times 1}$ as
\begin{align*}
\vz_i:=\lb\frac{\delta_i}{p}-1\rb\begin{bmatrix}
\frac{1}{\sqrt{w_0}}\calGT\calP_T\calG\lb\calGT(\mZ)\ve_{i}\ve_{i}^\tran\rb\ve_0\\
\vdots\\
\frac{1}{\sqrt{w_{\ell}}}\calGT\calP_T\calG\lb\calGT(\mZ)\ve_{i}\ve_{i}^\tran\rb\ve_{\ell}\\
\vdots\\
\frac{1}{\sqrt{w_{n-1}}}\calGT\calP_T\calG\lb\calGT(\mZ)\ve_{i}\ve_{i}^\tran\rb\ve_{n-1}
\end{bmatrix},
\end{align*}
then it can be easily seen that 
\begin{align*}
\left\|\calP_T\calG\lb\calI-\frac{1}{p}\mathbb E\lsb\calAT_{k}\calA_{k}\rsb\rb\calGT(\mZ)\right\|_{\calG,\mathsf{F}}^2=:\vecnorm{\sum_{i=0}^{n-1}\vz_i}{2}^2.
\end{align*}
For the upper bound of $\vecnorm{\vz_i}{2}$, a direct calculation yields that 
\begin{align*}
\vecnorm{\vz_i}{2} &\leq \frac{1}{p}\left\|\calP_T\calG\lb\calGT(\mZ)\ve_i\ve_i^\tran\rb\right\|_{\calG,\mathsf{F}}\\
&= \frac{1}{p}\frac{1}{\sqrt{w_i}}\left\|\calP_T\calG\lb\sqrt{w_i}\calGT(\mZ)\ve_i\ve_i^\tran\rb\right\|_{\calG,\mathsf{F}}\\
&\lesssim \frac{1}{p}\sqrt{\frac{\mu_1r\log(s n)}{n}}\cdot\frac{\vecnorm{\calGT(\mZ)\ve_i}{2}}{\sqrt{w_i}}\\
&\lesssim \frac{1}{p}\sqrt{\frac{\mu_1r\log(s n)}{n}}\left\|\mZ\right\|_{\calG,\infty},
\end{align*}
where the third line follows from Lemma~\ref{lem: supp5} \msh{}{and the last line is due to the definition of $\ginfnorm{\cdot}$ in \eqref{eq: norm definition}.}

In addition,
\begin{align*}
\opnorm{\E{\sum_{i=0}^{n-1}\vz_i\vz_i^{\ast}}} &\leq  \sum_{i=0}^{n-1}\E{\vecnorm{\vz_i}{2}^2}\\
&\leq \frac{1}{p}\sum_{i=0}^{n-1}\left\|\calP_T\calG\lb\calGT(\mZ)\ve_{i}\ve_{i}^\tran\rb\right\|_{\calG,\mathsf{F}}^2\\
&\lesssim \frac{1}{p}\frac{\mu_1r\log(s n)}{n}\sum_{i=0}^{n-1}\frac{\vecnorm{\calGT(\mZ)\ve_{i}}{2}^2}{w_i}\\
&=\frac{1}{p}\frac{\mu_1r\log(s n)}{n}\left\|\mZ\right\|_{\calG,\mathsf{F}}^2,
\end{align*}
\msh{}{where the third inequality is due to Lemma~\ref{lem: supp5},}
and the same bound can be obtained for $\opnorm{\E{\sum_{i=0}^{n-1}\vz_i^*\vz_i}}$. 

\kw{}{Therefore, by} the matrix Bernstein inequality \eqref{ineq: bernstein}, we can show that
\begin{align*}
\vecnorm{\sum_{i=0}^{n-1}\vz_i}{2} \lesssim \sqrt{\frac{\mu_1r\log(s n)}{n}}\lb\sqrt{\frac{\log(s n)}{p}}\left\|\mZ\right\|_{\calG,\mathsf{F}}+\frac{\log(s n)}{p}\left\|\mZ\right\|_{\calG,\infty}\rb
\end{align*}
holds with probability at least \msh{}{$1-(sn)^{-c}$} for a universal constant $\msh{}{c>0}$. 
Taking the uniform bound completes the proof.

\subsection{Proof of \eqref{ineq: part_prop4}}
The definition of $\left\|\cdot\right\|_{\calG,\infty}$ in \eqref{eq: norm definition} allows us to express \kw{}{$\left\|\calP_T\calG\lb\calI-\frac{1}{p}\E{\calAT_{k}\calA_{k}}\calGT(\mZ)\rb\right\|_{\calG,\infty}$ as}
\begin{align*}
\left\|\calP_T\calG\lb\calI-\frac{1}{p}\E{\calAT_{k}\calA_{k}}\calGT(\mZ)\rb\right\|_{\calG,\infty} &= \left\|\sum_{i=0}^{n-1}\lb\frac{\delta_i}{p}-1\rb\calP_T\calG\lb\calGT(\mZ)\ve_i\ve_i^\tran\rb\right\|_{\calG,\infty}\\
&\msh{}{= \max_{0\leq j\leq n-1}\frac{1}{\sqrt{w_j}}\vecnorm{\calGT\lb\sum_{i=0}^{n-1}\lb\frac{\delta_i}{p}-1\rb\calP_T\calG\lb\calGT(\mZ)\ve_i\ve_i^\tran\rb\rb\ve_j}{2}}\\
&=\max_{0\leq j\leq n-1}\vecnorm{\lb\sum_{i=0}^{n-1}\lb\frac{\delta_i}{p}-1\rb\frac{1}{\sqrt{w_j}}\calGT\calP_T\calG\lb\calGT(\mZ)\ve_i\ve_i^\tran\rb\rb\ve_j}{2}.
\end{align*}
Define $\vz_i^j$ to be \kw{}{the $s$-dimensional vector}
\begin{align*}
\vz_i^j := \lb\frac{\delta_i}{p}-1\rb\frac{1}{\sqrt{w_j}}\calGT\calP_T\calG\lb\calGT(\mZ)\ve_i\ve_i^\tran\rb\ve_j,\quad (i,j)\in[n]\times[n].
\end{align*}
Then one can easily see that
\begin{align*}
\left\|\calP_T\calG\lb\calI-\frac{1}{p}\E{\calAT_{k}\calA_{k}}\calGT(\mZ)\rb\right\|_{\calG,\infty} =: \max_{0\leq j\leq n-1}\vecnorm{\sum_{i=0}^{n-1}\vz_i^j}{2}.
\end{align*}
For any fixed $j\in[n]$, $\vecnorm{\vz_i^j}{2}$ can be bounded as follows:
\begin{align*}
\vecnorm{\vz_i^j}{2}&\leq \frac{1}{p}\frac{1}{\sqrt{w_i}}\frac{\sqrt{w_i}}{\sqrt{w_j}}\vecnorm{\calGT\calP_T\calG\lb\calGT(\mZ)\ve_i\ve_i^\tran\rb\ve_j}{2}\\
&=\frac{1}{p}\frac{1}{\sqrt{w_i}}\frac{\sqrt{w_i}}{\sqrt{w_j}}\sup_{\msh{}{\vecnorm{\bbeta}{2}=1}}\lab\la\calGT\calP_T\calG\lb\calGT(\mZ)\ve_i\ve_i^\tran\rb\ve_j,\msh{}{\bbeta}\ra\rab\\
&=\frac{1}{p}\frac{1}{\sqrt{w_i}}\sup_{\msh{}{\vecnorm{\bbeta}{2}=1}}\frac{\sqrt{w_i}}{\sqrt{w_j}}\lab\la\calP_T\calG\lb\calGT(\mZ)\ve_i\ve_i^\tran\rb,\calG(\msh{}{\bbeta}\ve_j^\tran)\ra\rab\\
&\leq \frac{1}{p}\frac{1}{\sqrt{w_i}}\frac{3\mu_1r}{n}\vecnorm{\calGT(\mZ)\ve_i}{2}\sup_{\msh{}{\vecnorm{\bbeta}{2}=1}}\vecnorm{\msh{}{\bbeta}}{2}\\
&= \frac{3\mu_1r}{n p}\cdot\frac{\vecnorm{\calGT(\mZ)\ve_i}{2}}{\sqrt{w_i}}\\
&\leq \frac{3\mu_1r}{n p}\left\|\mZ\right\|_{\calG,\infty},\numberthis\label{ineq: zij}
\end{align*}
where the fourth line follows from Lemma \ref{lem: supp3} \msh{}{and the last line is due to the definition of $\ginfnorm{\cdot}$ in \eqref{eq: norm definition}.  } 

Moreover, we have
\begin{align*}
\opnorm{\sum_{i=0}^{n-1}\E{\vz_i^j(\vz_i^j)^{\ast}}} &\leq \sum_{i=0}^{n-1}\E{\vecnorm{\vz_i^j}{2}^2}\\
&\msh{}{\leq \frac{1}{p}\sum_{i=0}^{n-1}\vecnorm{\frac{1}{\sqrt{w_j}}\calGT\calP_T\calG\lb\calGT(\mZ)\ve_i\ve_i^\tran\rb\ve_j}{2}^2}\\
&\leq \frac{1}{p}\lb\frac{3\mu_1r}{n}\rb^2\cdot\sum_{i=0}^{n-1}\lb\frac{\vecnorm{\calGT(\mZ)\ve_{i}}{2}}{\sqrt{w_i}}\rb^2\\
&= \frac{1}{p}\lb\frac{3\mu_1r}{n}\rb^2\cdot\left\|\mZ\right\|_{\calG,\mathsf{F}}^2,
\end{align*}
\msh{}{where the third inequality follows from \eqref{ineq: zij}.}
The same bound can be obtained \kw{}{for $\opnorm{\sum_{i=0}^{n-1}\E{(\vz_i^j)^*\vz_i^j}}$} as well. 

The matrix Bernstein inequality \eqref{ineq: bernstein} taken collectively with the uniform bound yields that
\begin{align*}
\left\|\calP_T\calG\lb\calI-\frac{1}{p}\E{\calAT_{k}\calA_{k}}\calGT(\mZ)\rb\right\|_{\calG,\infty} &= \max_{0\leq j\leq n-1}\vecnorm{\sum_{i=0}^{n-1}\vz_i^j}{2}\\
&\lesssim \frac{\mu_1r}{n}\lb\sqrt{\frac{\log(s n)}{p}}\left\|\mZ\right\|_{\calG,\mathsf{F}}+\frac{\log(s n)}{p}\left\|\mZ\right\|_{\calG,\infty}\rb
\end{align*}
holds with  probability at least $\msh{}{1-ns^{-c_2}}$ \msh{}{for a universal constant $\msh{}{c_2>2}$.} 

Finally, we take the uniform bound  over all $k\in\{1,\cdots,k_0\}$ again to complete the proof.

\section{Proofs of Lemmas \ref{lem: key lemma1} to \ref{lem: key lemma5}}
\label{sec: proof of part lemma}
This section presents the proofs of Lemmas \ref{lem: key lemma1} to \ref{lem: key lemma5}, which have been used to verify  \eqref{ineq: F norm} and \eqref{ineq: perp op norm}. 
\subsection{Proof of Lemma \ref{lem: key lemma1}}
Note that 
\begin{align*}
\opnorm{\calP_T\calG\lb\calI-\frac{n}{m}\calAT_k\calA_k\rb\calGT\calP_T} \leq \opnorm{\calP_T\calG\lb\calI-\frac{n}{m}\E{\calAT_k\calA_k}\rb\calGT\calP_T} + \frac{n}{m}\opnorm{\calP_T\calG\lb\calAT_k\calA_k - \E{\calAT_k\calA_k}\rb\calGT\calP_T}.
\end{align*}
According to \eqref{ineq: part_prop1} in Lemma \ref{lem: partition}, the first term is upper bounded by $\frac{1}{4}$. 
We will bound the second term via the matrix Bernstein inequality \eqref{ineq: bernstein}.

For any $\mZ\in\C^{s n_1\times n_2}$, \msh{}{by the definition of $\calA_k^\ast\calA_k$ in \eqref{eq:AktAk}
,} we have
\begin{align*}
\calP_T\calG\calAT_k\calA_k\calGT\calP_T(\mZ) &\msh{}{= \calP_T\calG\lb\sum_{i\in\Omega_k}\la\vb_i\ve_i^\tran, \calGT\calP_T(\mZ)\ra\vb_i\ve_i^\tran\rb}\\
&= \sum_{i\in\Omega_k}\la\vb_i\ve_i^\tran, \calGT\calP_T(\mZ)\ra\calP_T\calG\lb\vb_i\ve_i^\tran\rb\\
&= \sum_{i\in\Omega_k}\la\calP_T\calG\lb\vb_i\ve_i^\tran\rb,\mZ\ra\calP_T\calG\lb\vb_i\ve_i^\tran\rb.
\end{align*}
If we define $\msh{}{\vz_i} := \vect\lb\calP_T\calG\lb\vb_i\ve_i^\tran\rb\rb \in \C^{sn_1n_2\times 1}$, then it follows that
\begin{align*}
\opnorm{\calP_T\calG\calAT_k\calA_k\calGT\calP_T} &= \sup_{\msh{}{\fronorm{\mW}=1}}\fronorm{\calP_T\calG\calAT_k\calA_k\calGT\calP_T(\msh{}{\mW})}\\
&= \sup_{\msh{}{\fronorm{\mW}=1}}\fronorm{\sum_{i\in\Omega_k}\la\calP_T\calG\lb\vb_i\ve_i^\tran\rb,\msh{}{\mW}\ra\calP_T\calG\lb\vb_i\ve_i^\tran\rb}\\
&= \sup_{\msh{}{\vecnorm{\vect(\msh{}{\mW})}{2}=1}}\vecnorm{\sum_{i\in\Omega_k}\msh{}{\vz_i}^*\vect(\msh{}{\mW})\msh{}{\vz_i}}{2}\\
&= \sup_{\msh{}{\vecnorm{\vect(\msh{}{\mW})}{2}=1}}\vecnorm{\sum_{i\in\Omega_k}\msh{}{\vz_i\vz_i}^*\vect(\msh{}{\mW})}{2}\\
&= \opnorm{\sum_{i\in\Omega_k}\msh{}{\vz_i\vz_i}^*},
\end{align*}
where \kw{}{it is obvious that} $\msh{}{\vz_i\vz_i^*}$  are independent and positive semi-definite random matrices. Hence,
\begin{align*}
\opnorm{\calP_T\calG\lb\calAT_k\calA_k - \E{\calAT_k\calA_k}\rb\calGT\calP_T} = \opnorm{\sum_{i\in\Omega_k}\lb\msh{}{\vz_i\vz_i^*-\E{\vz_i\vz
_i^*}}\rb}.
\end{align*}

Firstly, $\opnorm{\msh{}{\vz_i\vz_i^*-\E{\vz_i\vz_i^*}}}$ can be bounded as follows:
\begin{align*}
\opnorm{\msh{}{\vz_i\vz_i^*}-\E{\msh{}{\vz_i\vz_i^*}}} &\leq \max\lcb\opnorm{\msh{}{\vz_i\vz_i^*}},\opnorm{\E{\msh{}{\vz_i\vz_i^*}}}\rcb\\
&\leq \max\lcb\opnorm{\msh{}{\vz_i\vz_i^*}},\E{\opnorm{\msh{}{\vz_i\vz_i^*}}}\rcb\\
&\leq \max\lcb\vecnorm{\msh{}{\vz_i}}{2}^2,\E{\vecnorm{\msh{}{\vz_i}}{2}^2}\rcb,
\end{align*}
where the second line is due to the Jensen inequality. By the definition of $\msh{}{\vz_i}$, we have $\vecnorm{\msh{}{\vz_i}}{2}^2 = \fronorm{\calP_T\calG\lb\vb_i\ve_i^T\rb}^2$. Then applying \eqref{ineq: supp13} in Corollary \ref{cor: useful cor} implies that 
\begin{align*}
\opnorm{\msh{}{\vz_i\vz_i^*}-\E{\msh{}{\vz_i\vz_i^*}}}\leq\msh{}{\max\lcb\vecnorm{\msh{}{\vz_i}}{2}^2,\E{\vecnorm{\msh{}{\vz_i}}{2}^2}\rcb}\leq \frac{2\mu_1rs\mu_0}{n}.
\end{align*}

Secondly,
\begin{align*}
\opnorm{\sum_{i\in\Omega_k}\E{\lb\msh{}{\vz_i\vz_i^*}-\E{\msh{}{\vz_i\vz_i^*}}\rb^2}} &= \opnorm{\sum_{i\in\Omega_k}\E{(\msh{}{\vz_i\vz_i^*})^2} - \lb\E{\msh{}{\vz_i\vz_i^*}}\rb^2}\\
&\leq \opnorm{\sum_{i\in\Omega_k}\E{(\msh{}{\vz_i\vz_i^*})^2}}\\
&\leq \max_{i\in\Omega_k}\opnorm{\msh{}{\vz_i\vz_i^*}}\cdot\opnorm{\sum_{i\in\Omega_k}\E{(\msh{}{\vz_i\vz_i^*})}}\\
&\leq \frac{2\mu_1rs\mu_0}{n}\cdot\frac{5m}{4n},
\end{align*}
Here  the last line follows from a direct calculation:
\begin{align*}
\opnorm{\sum_{i\in\Omega_k}\E{(\msh{}{\vz_i\vz_i^*})}} &= \sup_{\msh{}{\vecnorm{\vect(\msh{}{\mW})}{2}=1}}\vecnorm{\sum_{i\in\Omega_k}\E{\vect\lb\calP_T\calG(\vb_i\ve_i^\tran)\rb \vect\lb\calP_T\calG(\vb_i\ve_i^\tran)\rb^*\vect(\msh{}{\mW})}}{2}\\
&= \sup_{\msh{}{\fronorm{\mW}=1}}\vecnorm{\sum_{i\in\Omega_k}\E{\la\calP_T\calG(\vb_i\ve_i^\tran),\msh{}{\mW}\ra \vect\lb\calP_T\calG(\vb_i\ve_i^\tran)\rb}}{2}\\
&= \sup_{\msh{}{\fronorm{\mW}=1}}\fronorm{\sum_{i\in\Omega_k}\E{\la\calP_T\calG(\vb_i\ve_i^\tran),\msh{}{\mW}\ra \calP_T\calG(\vb_i\ve_i^\tran)}}\\
&= \sup_{\msh{}{\fronorm{\mW}=1}}\fronorm{\sum_{i\in\Omega_k}\E{\lb\vb_i^*\calGT\calP_T(\msh{}{\mW})\ve_i\rb\calP_T\calG(\vb_i\ve_i^\tran)}}\\
&=  \sup_{\msh{}{\fronorm{\mW}=1}}\fronorm{\sum_{i\in\Omega_k}\E{\calP_T\calG\lb\vb_i\vb_i^*\calGT\calP_T(\msh{}{\mW})\ve_i\ve_i^\tran\rb}}\\
&= \sup_{\msh{}{\fronorm{\mW}=1}}\fronorm{\sum_{i\in\Omega_k}\calP_T\calG\lb\calGT\calP_T(\msh{}{\mW})\ve_i\ve_i^\tran\rb}\\
&\leq \frac{5m}{4n},
\end{align*}
where \kw{}{in the last inequality we have utilized  \eqref{ineq: part_prop1} in the following  way},
\begin{align*}
\frac{1}{4} &\geq \opnorm{\calP_T\calG\lb\calI-\frac{n}{m}\E{\calAT_k\calA_k}\rb\calGT\calP_T}\\
&\geq \frac{n}{m}\opnorm{\calP_T\calG \E{\calAT_k\calA_k}\calGT\calP_T} - \opnorm{\calP_T\calG\calGT\calP_T}\\
&\geq \frac{n}{m}\sup_{\msh{}{\fronorm{\mW}=1}}\fronorm{\sum_{i\in\Omega_k}\calP_T\calG\lb\calGT\calP_T(\msh{}{\mW})\ve_i\ve_i^\tran\rb} - 1.
\end{align*}

Since we can obtain the same bound for $\opnorm{\sum_{i\in\Omega_k}\E{\lb \msh{}{\vz_i^*\vz_i-\E{\vz_i^*\vz_i}}\rb^2}}$,  applying the matrix Bernstein inequality \eqref{ineq: bernstein} implies that  with \msh{}{probability at least $1-(sn)^{-c}$},
\begin{align*}
\frac{n}{m}\opnorm{\calP_T\calG\lb\calAT_k\calA_k - \E{\calAT_k\calA_k}\rb\calGT\calP_T} &= \frac{n}{m}\opnorm{\sum_{i\in\Omega_k}\lb\msh{}{\vz_i\vz_i^* - \E{\vz_i\vz_i^*}}\rb}\\
&\lesssim \frac{n}{m}\cdot\lb\sqrt{\frac{5m}{4n}\cdot\frac{2\mu_1rs\mu_0}{n}\cdot\log(s n)} + \frac{2\mu_1rs\mu_0\log(s n)}{n}\rb\\
&= \frac{1}{m}\cdot\lb\sqrt{\frac{5m\mu_1rs\mu_0\log(s n)}{2}} + 2\mu_1rs\mu_0\log(s n)\rb\\
&\lesssim  \frac{1}{m}\cdot\sqrt{\frac{5m\mu_1rs\mu_0\log(s n)}{2}}\\
&\leq \frac{1}{4},
\end{align*}
where the fourth line \msh{}{and the last line} hold when $m \gtrsim  \mu_1r s\mu_0\log(s n)$.

Finally, combining the two terms together 
completes the proof.

\subsection{Proof of Lemma \ref{lem: key lemma2}}
Notice that
\begin{align*}
\opnorm{\calG\lb\calI-\frac{n}{m}\calAT_k\calA_k\rb\calGT(\mZ)} &\leq \opnorm{\calG\lb\calI - \frac{n}{m}\E{\calAT_k\calA_k}\rb\calGT(\mZ)} + \frac{n}{m}\opnorm{\calG\lb\calAT_k\calA_k - \E{\calAT_k\calA_k}\rb\calGT(\mZ)}\\
&\lesssim \sqrt{\frac{n\log(s n)}{m}}\gfronorm{\mZ} + \frac{n\log(s n)}{m}\ginfnorm{\mZ} \\
&\quad + \frac{n}{m}\opnorm{\calG\lb\calAT_k\calA_k - \E{\calAT_k\calA_k}\rb\calGT(\mZ)},\numberthis\label{eq:kwadd01}
\end{align*}
where the second line follows from \eqref{ineq: part_prop2}. In order to prove \eqref{ineq: part 2}, it suffices to bound the last term.

\msh{}{Recalling the definition of $\calA_k^\ast\calA_k$ in \eqref{eq:AktAk} and using the isotropy property of $\{\vb_i\}$ in \eqref{eq: isotropy}, we can rewrite the last term as} 
\begin{align*}
\frac{n}{m}\opnorm{\calG\lb\calAT_k\calA_k - \E{\calAT_k\calA_k}\rb\calGT(\mZ)} &= \frac{n}{m}\opnorm{\sum_{i\in\Omega_k}\calG\lb(\vb_i\vb_i^*-\mI)\calGT(\mZ)\ve_i\ve_i^\tran\rb}\\
&=: \frac{n}{m}\opnorm{\sum_{i\in\Omega_k}\mX_i},
\end{align*}
where $\mX_i = \calG\lb(\vb_i\vb_i^*-\mI)\calGT(\mZ)\ve_i\ve_i^\tran\rb\in\C^{s n_1 \times n_2}$. It can be easily seen that $\mX_i$ are independent random matrices with zero mean. 

The upper bound of $\opnorm{\mX_i}$ can be \kw{}{established} as follows:
\begin{align*}
\opnorm{\mX_i} &= \opnorm{\calG\lb(\vb_i\vb_i^*-\mI)\calGT(\mZ)\ve_i\ve_i^\tran\rb}\\
&= \opnorm{\mG_i\otimes \lb(\vb_i\vb_i^*-\mI)\calGT(\mZ)\ve_i\rb}\\
&\leq \opnorm{\mG_i}\cdot\opnorm{(\vb_i\vb_i^* - \mI)\calGT(\mZ)\ve_i}\\
&\leq \frac{1}{\sqrt{w_i}}\max\lcb\vecnorm{\vb_i}{2}^2,1\rcb\cdot\vecnorm{\calGT(\mZ)\ve_i}{2}\\
&\leq s\mu_0\ginfnorm{\mZ},
\end{align*}
\msh{}{where the second line follows from \eqref{eq: calG}, the third line is due to $\opnorm{\mA\otimes\mB}\leq \opnorm{\mA}\cdot\opnorm{\mB}$, and the last line follows from the definition of $\ginfnorm{\cdot}$ in \eqref{eq: norm definition}.}

To bound $\opnorm{\E{\sum_{i\in\Omega_k}\mX_i^*\mX_i}}$, we first define $\msh{}{\vz_i} = (\vb_i\vb_i^* - \mI)\calGT(\mZ)\ve_i\in\C^s$. Then a simple calculation yields that 
\begin{align*}
\E{\vecnorm{\msh{}{\vz_i}}{2}^2} &= \E{\ve_i^\tran\lb\calGT(\mZ)\rb^*(\vb_i\vb_i^*-\mI)^2\calGT(\mZ)\ve_i}\\
&= \ve_i^\tran\lb\calGT(\mZ)\rb^*\E{(\vb_i\vb_i^*-\mI)^2}\calGT(\mZ)\ve_i\\
&= \ve_i^\tran\lb\calGT(\mZ)\rb^*\lb\E{\vecnorm{\vb_i}{2}^2\vb_i\vb_i^* - 2\vb_i\vb_i^* + \mI}\rb\calGT(\mZ)\ve_i\\
&= \ve_i^\tran\lb\calGT(\mZ)\rb^*\lb\E{\vecnorm{\vb_i}{2}^2\vb_i\vb_i^* - \mI}\rb\calGT(\mZ)\ve_i\\
&\leq \ve_i^\tran\lb\calGT(\mZ)\rb^*\lb s\mu_0 \E{\vb_i\vb_i^*} - \mI\rb\calGT(\mZ)\ve_i\\
&\leq s\mu_0\cdot\vecnorm{\calGT(\mZ)\ve_i}{2}^2,\numberthis\label{ineq: E(zi^2)}
\end{align*}  
\msh{}{where the last two inequalities follow from the incoherence property \eqref{eq: incoherence b} and the isotropy property \eqref{eq: isotropy} of $\{\vb_i\}$.}
Furthermore, it follows that
\begin{align*}
\opnorm{\E{\sum_{i\in\Omega_k}\mX_i^*\mX_i}} &= \opnorm{\sum_{i\in\Omega_k}\E{\lb\mG_i\otimes\msh{}{\vz_i}\rb^*\lb\mG_i\otimes\msh{}{\vz_i}\rb}}\\
&= \opnorm{\sum_{i\in\Omega_k}\E{(\mG_i^\tran\mG_i)\otimes(\msh{}{\vz_i}^*\msh{}{\vz_i})}}\\
&= \opnorm{\sum_{i\in\Omega_k}(\mG_i^\tran\mG_i)\E{\vecnorm{\msh{}{\vz_i}}{2}^2}}\\
&\leq s\mu_0 \cdot \opnorm{\sum_{i\in\Omega_k}\vecnorm{\calGT(\mZ)\ve_i}{2}^2(\mG_i^\tran\mG_i)}\\
&\leq s\mu_0 \cdot\sum_{i\in\Omega_k}\opnorm{\mG_i^\tran\mG_i}\vecnorm{\calGT(\mZ)\ve_i}{2}^2\\
&\leq s\mu_0 \cdot\sum_{i\in\Omega_k}\frac{\vecnorm{\calGT(\mZ)\ve_i}{2}^2}{w_i}\\
&\leq s\mu_0 \cdot\sum_{i=1}^n\frac{\vecnorm{\calGT(\mZ)\ve_i}{2}^2}{w_i}\\
&= s\mu_0 \cdot\gfronorm{\mZ}^2,
\end{align*}
\msh{}{where the fourth line follows from \eqref{ineq: E(zi^2)}}, and $\opnorm{\E{\sum_{i\in\Omega_k}\mX_i\mX_i^*}}$ can  \kw{}{be similarly bounded}.

Therefore, by the matrix Bernstein inequality \eqref{ineq: bernstein}, 
\begin{align*}
\frac{n}{m}\opnorm{\calG\lb\calAT_k\calA_k - \E{\calAT_k\calA_k}\rb\calGT(\mZ)} &= \frac{n}{m}\opnorm{\sum_{i\in\Omega_k}\mX_i}\\
&\lesssim \frac{n}{m}\lb\sqrt{s\mu_0\log(s n)}\gfronorm{\mZ} + s\mu_0\log(s n)\ginfnorm{\mZ}\rb\\
&= \sqrt{\frac{n k_0s\mu_0\log(s n)}{m}}\gfronorm{\mZ} + \frac{n s\mu_0\log(s n)}{m}\ginfnorm{\mZ}
\end{align*}
holds with \msh{}{probability at least $1-(sn)^{-c}$ for a universal constant $c>0$.}
\kw{}{Inserting  this bound into \eqref{eq:kwadd01} we conclude that}
\begin{align*}
\opnorm{\calG\lb\calI-\frac{n}{m}\calAT_k\calA_k\rb\calGT(\mZ)} &\lesssim \lb\sqrt{\frac{n k_0s\mu_0\log(s n)}{m}} + \sqrt{\frac{n\log(s n)}{m}}\rb\gfronorm{\mZ} + \lb\frac{n s\mu_0\log(s n)}{m} + \frac{n\log(s n)}{m}\rb\ginfnorm{\mZ}\\
&\lesssim \sqrt{\frac{4n k_0s\mu_0\log(s n)}{m}}\gfronorm{\mZ} + \frac{2n s\mu_0\log(s n)}{m}\ginfnorm{\mZ}
\end{align*}
holds with \msh{}{probability exceeding $1-(sn)^{-c}$.}

\subsection{Proof of Lemma \ref{lem: key lemma3}}
Notice that 
\begin{align*}
\gfronorm{\calP_T\calG\lb\calI - \frac{n}{m}\calAT_k\calA_k\rb\calGT(\mZ)} &\leq \gfronorm{\calP_T\calG\lb\calI - \frac{n}{m}\E{\calAT_k\calA_k}\rb\calGT(\mZ)} + \frac{n}{m}\gfronorm{\calP_T\calG\lb\calAT_k\calA_k - \E{\calAT_k\calA_k}\rb\calGT(\mZ)}\\
&\lesssim \sqrt{\frac{\mu_1r\log(s n)}{n}}\lb\sqrt{\frac{n\log(s n)}{m}}\gfronorm{\mZ} + \frac{n\log(s n)}{m}\ginfnorm{\mZ}\rb\\
&\quad + \frac{n}{m}\gfronorm{\calP_T\calG\lb\calAT_k\calA_k - \E{\calAT_k\calA_k}\rb\calGT(\mZ)},\numberthis\label{eq:kwadd02}
\end{align*}
where the second line follows from \eqref{ineq: part_prop3}. We will adopt the matrix Bernstein inequality \eqref{ineq: bernstein} to bound the second term. 

\msh{}{Recalling the definition of $\calA_k^\ast\calA_k$ in \eqref{eq:AktAk} and letting $\msh{}{\vz_i} := \lb\vb_i\vb_i^* - \mI\rb\calGT(\mZ)\ve_i\in\C^s$, we have}
\begin{align*}
\frac{n}{m}\gfronorm{\calP_T\calG\lb\calAT_k\calA_k - \E{\calAT_k\calA_k}\rb\calGT(\mZ)} 
&\msh{}{= \frac{n}{m}\gfronorm{\calP_T\calG\lb\sum_{i\in\Omega_k}\lb(\vb_i\vb_i^\ast - \E{\vb_i\vb_i^\ast})\calGT(\mZ)\ve_i\ve_i^\tran\rb\rb}}\\
&= \frac{n}{m}\gfronorm{\sum_{i\in\Omega_k}\calP_T\calG\lb(\vb_i\vb_i^* - \mI)\calGT(\mZ)\ve_i\ve_i^\tran\rb}\\
&= \frac{n}{m}\gfronorm{\sum_{i\in\Omega_k}\calP_T\calG(\msh{}{\vz_i}\ve_i^\tran)}\\
&= \frac{n}{m}\sqrt{\sum_{j=0}^{n-1}\frac{1}{w_j}\vecnorm{\calGT\lb\sum_{i\in\Omega_k}\calP_T\calG(\msh{}{\vz_i}\ve_i^\tran)\rb\ve_j}{2}^2}\\
&= \frac{n}{m}\sqrt{\sum_{j=0}^{n-1}\frac{1}{w_j}\vecnorm{\sum_{i\in\Omega_k}\calGT\calP_T\calG(\msh{}{\vz_i}\ve_i^\tran)\ve_j}{2}^2},
\end{align*}
\msh{}{where the second equality is due to the isotropy property of $\{\vb_i\}$ in \eqref{eq: isotropy}.}
Furthermore, \kw{}{denoting by $\vy_i\in\C^{s n\times 1}$ the vector}
\begin{align*}
\vy_i := \begin{bmatrix}
\frac{1}{\sqrt{w_0}}\calGT\calP_T\calG(\msh{}{\vz_i}\ve_i^\tran)\ve_0\\
\vdots\\
\frac{1}{\sqrt{w_\ell}}\calGT\calP_T\calG(\msh{}{\vz_i}\ve_i^\tran)\ve_\ell\\
\vdots\\
\frac{1}{\sqrt{w_{n-1}}}\calGT\calP_T\calG(\msh{}{\vz_i}\ve_i^\tran)\ve_{n-1}
\end{bmatrix},
\end{align*}
\kw{}{the  second term can be expressed as}
\begin{align*}
\frac{n}{m}\gfronorm{\calP_T\calG\lb\calAT_k\calA_k - \E{\calAT_k\calA_k}\rb\calGT(\mZ)} =: \frac{n}{m}\vecnorm{\sum_{i\in\Omega_k}\vy_i}{2}.\numberthis\label{eq:kwadd03}
\end{align*}
Clearly, $\vy_i$ are independent random vectors with zero mean. 

A direct calculation yields that  
\begin{align*}
\vecnorm{\vy_i}{2} &= \sqrt{\sum_{j=0}^{n-1}\frac{1}{w_j}\vecnorm{\calGT\calP_T\calG(\msh{}{\vz_i}\ve_i^\tran)\ve_j}{2}^2}\\
&= \gfronorm{\calP_T\calG(\msh{}{\vz_i}\ve_i^\tran)}\\
&= \frac{1}{\sqrt{w_i}}\gfronorm{\calP_T\calG\lb\sqrt{w_i}\msh{}{\vz_i}\ve_i^\tran\rb}\\
&\lesssim \frac{1}{\sqrt{w_i}}\vecnorm{\msh{}{\vz_i}}{2}\sqrt{\frac{\mu_1r\log(s n)}{n}}\\
&= \frac{1}{\sqrt{w_i}}\sqrt{\frac{\mu_1r\log(s n)}{n}}\cdot\vecnorm{\lb\vb_i\vb_i^* - \mI\rb\calGT(\mZ)\ve_i}{2}\\
&\leq \frac{1}{\sqrt{w_i}}\sqrt{\frac{\mu_1r\log(s n)}{n}}\cdot\opnorm{\vb_i\vb_i^* - \mI}\cdot\vecnorm{\calGT(\mZ)\ve_i}{2}\\
&\leq \sqrt{\frac{\mu_1r\log(s n)}{n}}\cdot s\mu_0 \cdot \ginfnorm{\mZ},
\end{align*}
where the fourth line follows from Lemma \ref{lem: supp5} \msh{}{and the last line is due to the definition of $\ginfnorm{\cdot}$ in \eqref{eq: norm definition}.}

\kw{}{Additionally, we have}
\begin{align*}
\opnorm{\E{\sum_{i\in\Omega_k}\vy_i\vy_i^*}} &\leq \sum_{i\in\Omega_k}\E{\vecnorm{\vy_i}{2}^2}\\
&= \sum_{i\in\Omega_k} \E{\gfronorm{\calP_T\calG(\msh{}{\vz_i}\ve_i^\tran)}^2}\\
&\lesssim \sum_{i\in\Omega_k} \frac{1}{w_i}\frac{\mu_1r\log(s n)}{n}\cdot\E{\vecnorm{\msh{}{\vz_i}}{2}^2}\\
&\lesssim s\mu_0\frac{\mu_1r\log(s n)}{n}\cdot\sum_{i\in\Omega_k}\frac{1}{w_i}\vecnorm{\calGT(\mZ)\ve_i}{2}^2\\
&\lesssim \frac{s\mu_0\cdot\mu_1r\log(s n)}{n}\cdot\gfronorm{\mZ}^2, 
\end{align*}
where the third line is due to Lemma \ref{lem: supp5} and the fourth line follows from
\begin{align*}
\E{\vecnorm{\msh{}{\vz_i}}{2}^2} &= \E{\vecnorm{\lb\vb_i\vb_i^* - \mI\rb\calGT(\mZ)\ve_i}{2}^2}\\
&= \E{\ve_i^T\lb\calGT(\mZ)\rb^*(\vb_i\vb_i^* - \mI)^2\calGT(\mZ)\ve_i}\\
&= \ve_i^T\lb\calGT(\mZ)\rb^*\lb\E{\lb\vecnorm{\vb_i}{2}^2\vb_i\vb_i^*\rb} - \mI\rb\calGT(\mZ)\ve_i\\
&\leq s\mu_0\vecnorm{\calGT(\mZ)\ve_i}{2}^2.\numberthis\label{ineq: Ezi^2}
\end{align*}
The same upper bound can be obtained for \opnorm{\E{\sum_{i\in\Omega_k}\vy_i^*\vy_i}}. 

Applying the matrix Bernstein inequality yields that 
\begin{align*}
\frac{n}{m}\vecnorm{\sum_{i\in\Omega_k}\vy_i}{2} &\lesssim \frac{n}{m}\lb\sqrt{\frac{s\mu_0\mu_1r\log^2(s n)}{n}}\gfronorm{\mZ} + \sqrt{\frac{\mu_1r\log(s n)}{n}}\cdot s\mu_0 \log(s n)\cdot \ginfnorm{\mZ}\rb\\
&= \sqrt{\frac{\mu_1r\log(s n)}{n}}\lb\sqrt{\frac{n k_0s\mu_0\log(s n)}{m}}\gfronorm{\mZ} + \frac{n s\mu_0\log(s n)}{m}\ginfnorm{\mZ}\rb
\end{align*}
holds with \msh{}{probability at least $1-(sn)^{-c}$ for a universal constant $c>0$.}
\kw{}{Noting \eqref{eq:kwadd02} and \eqref{eq:kwadd03}, it follows immediately that}
\begin{align*}
\gfronorm{\calP_T\calG\lb\calI - \frac{n}{m}\calAT_k\calA_k\rb\calGT(\mZ)} &\lesssim \sqrt{\frac{\mu_1r\log(s n)}{n}}\lb\sqrt{\frac{n\log(s n)}{m}}\gfronorm{\mZ} + \frac{n\log(s n)}{m}\ginfnorm{\mZ}\rb\\
&\quad + \sqrt{\frac{\mu_1r\log(s n)}{n}}\lb\sqrt{\frac{n k_0s\mu_0\log(s n)}{m}}\gfronorm{\mZ} + \frac{n s\mu_0\log(s n)}{m}\ginfnorm{\mZ}\rb\\
&\lesssim \sqrt{\frac{\mu_1r\log(s n)}{n}}\lb\sqrt{\frac{4n k_0s\mu_0\log(s n)}{m}}\gfronorm{\mZ} + \frac{2n s\mu_0\log(s n)}{m}\ginfnorm{\mZ}\rb
\end{align*}
holds with \msh{}{probability greater than $1-(sn)^{-c}$}.

\subsection{Proof of Lemma \ref{lem: key lemma4}}
By the triangle inequality, we have
\begin{align*}
\ginfnorm{\calP_T\calG\lb\calI - \frac{n}{m}\calAT_k\calA_k\rb\calGT(\mZ)} &\leq \ginfnorm{\calP_T\calG\lb\calI - \frac{n}{m}\E{\calAT_k\calA_k}\rb\calGT(\mZ)} + \frac{n}{m}\ginfnorm{\calP_T\calG\lb\calAT_k\calA_k - \E{\calAT_k\calA_k}\rb\calGT(\mZ)}\\
&\lesssim \frac{\mu_1r}{n}\lb\sqrt{\frac{n\log(s n)}{m}}\gfronorm{\mZ} + \frac{n\log(s n)}{m}\ginfnorm{\mZ}\rb\\
&\quad + \frac{n}{m}\ginfnorm{\calP_T\calG\lb\calAT_k\calA_k - \E{\calAT_k\calA_k}\rb\calGT(\mZ)},\numberthis\label{eq:kwadd04}
\end{align*}
where the second line is due to \eqref{ineq: part_prop4}. In the following proof, we will  \kw{}{upper bound the second term} by the matrix Bernstein inequality \eqref{ineq: bernstein} \kw{}{and the uniform bound argument}.

If we define $\msh{}{\vz_i} = (\vb_i\vb_i^* - \mI)\calGT(\mZ)\ve_i\in\C^s$ and $\vy_i^j = \frac{1}{\sqrt{w_j}}\calGT\calP_T\calG\lb\msh{}{\vz_i}\ve_i^\tran\rb\ve_j\in\C^s$, the second term can be rewritten as 
\begin{align*}
\frac{n}{m}\ginfnorm{\calP_T\calG\lb\calAT_k\calA_k - \E{\calAT_k\calA_k}\rb\calGT(\mZ)} &= \frac{n}{m}\ginfnorm{\sum_{i\in\Omega_k}\calP_T\calG\lb(\vb_i\vb_i^* - \mI)\calGT(\mZ)\ve_i\ve_i^\tran\rb}\\
&= \frac{n}{m}\ginfnorm{\sum_{i\in\Omega_k}\calP_T\calG(\msh{}{\vz_i}\ve_i^\tran)}\\
&= \frac{n}{m}\sup_{0\leq j\leq n-1}\frac{1}{\sqrt{w_j}}\vecnorm{\sum_{i\in\Omega_k}\calGT\lb\calP_T\calG(\msh{}{\vz_i}\ve_i^\tran)\rb\ve_j}{2}\\
&=: \frac{n}{m}\sup_{0\leq j\leq n-1}\vecnorm{\sum_{i\in\Omega_k}\vy_i^j}{2},\numberthis\label{eq:kwadd05}
\end{align*}
\msh{}{where the first equation follows from \eqref{eq:AktAk} and the isotropy property of $\{\vb_i\}$ in \eqref{eq: isotropy}.}


For any fixed $j\in[n]$, $\vecnorm{\vy_i^j}{2}$ can be bounded as follows:
\begin{align*}
\vecnorm{\vy_i^j}{2} &= \frac{1}{\sqrt{w_j}}\vecnorm{\calGT\calP_T\calG\lb\msh{}{\vz_i}\ve_i^\tran\rb\ve_j}{2}\\
&= \frac{1}{\sqrt{w_j}} \sup_{\msh{}{\vecnorm{\bbeta}{2} = 1}}\lab\la\calGT\calP_T\calG(\msh{}{\vz_i}\ve_i^\tran)\ve_j, \bbeta\ra\rab\\
&= \frac{1}{\sqrt{w_i}}\sup_{\msh{}{\vecnorm{\bbeta}{2} = 1}} \frac{\sqrt{w_i}}{\sqrt{w_j}}\lab\la\calP_T\calG(\msh{}{\vz_i}\ve_i^\tran),\calG(\msh{}{\bbeta}\ve_j^\tran)\ra\rab\\
&\leq \frac{1}{\sqrt{w_i}}\frac{3\mu_1r}{n}\vecnorm{\msh{}{\vz_i}}{2}\numberthis\label{ineq: yij}\\
&= \frac{1}{\sqrt{w_i}}\frac{3\mu_1r}{n}\vecnorm{(\vb_i\vb_i^* - \mI)\calGT(\mZ)\ve_i}{2}\\
&\msh{}{\leq \frac{1}{\sqrt{w_i}}\frac{3\mu_1r}{n}\opnorm{\vb_i\vb_i^* - \mI}\cdot\vecnorm{\calGT(\mZ)\ve_i}{2}}\\
&\leq s\mu_0\cdot\frac{3\mu_1r}{n}\ginfnorm{\mZ},
\end{align*}
where the fourth line follows from Lemma \ref{lem: supp3} \msh{}{and the last line is due to the incoherence property of $\{\vb_i\}$ in \eqref{eq: incoherence b} and the definition of $\ginfnorm{\cdot}$ in \eqref{eq: norm definition}.} 

Moreover,  
\begin{align*}
\E{\sum_{i\in\Omega_k}\vy_i^j(\vy_i^j)^*} &\leq \E{\sum_{i\in\Omega_k}\vecnorm{\vy_i^j}{2}^2}\\
&\lesssim \sum_{i\in\Omega_k}\frac{1}{w_i}\lb\frac{\mu_1r}{n}\rb^2\E{\vecnorm{\msh{}{\vz_i}}{2}^2}\\
&\lesssim \sum_{i\in\Omega_k}\frac{1}{w_i}\lb\frac{\mu_1r}{n}\rb^2 s\mu_0\cdot \vecnorm{\calGT(\mZ)\ve_i}{2}^2\\
&\lesssim \lb\frac{\mu_1r}{n}\rb^2 s\mu_0\cdot\gfronorm{\mZ}^2,
\end{align*}
\msh{}{where the second line is due to \eqref{ineq: yij} and the third line follows from \eqref{ineq: Ezi^2}.}
It also holds that $\E{\sum_{i\in\Omega_k}(\vy_i^j)^*\vy_i^j}\leq \lb\frac{\mu_1r}{n}\rb^2 s\mu_0\cdot\gfronorm{\mZ}^2$. 

Applying the matrix Bernstein inequality and taking the uniform bound implies that
\begin{align*}
\frac{n}{m}\sup_{0\leq j\leq n-1}\vecnorm{\sum_{i\in\Omega_k}\vy_i^j}{2} &\lesssim \frac{n}{m}\lb\frac{\mu_1r}{n}\sqrt{s\mu_0\log(s n)}\gfronorm{\mZ} + s\mu_0\log(s n)\frac{\mu_1r}{n}\ginfnorm{\mZ}\rb\\
&= \frac{\mu_1r}{n}\lb\sqrt{\frac{n k_0s\mu_0\log(s n)}{m}}\gfronorm{\mZ} + \frac{n s\mu_0\log(s n)}{m}\ginfnorm{\mZ}\rb
\end{align*}
holds with \msh{}{probability at least $1-ns^{-c_2}$ for a numerical constant $c_2>2$.}
\kw{}{Noting \eqref{eq:kwadd04} and \eqref{eq:kwadd05} we can conclude that}
\begin{align*}
\ginfnorm{\calP_T\calG\lb\calI - \frac{n}{m}\calAT_k\calA_k\rb\calGT(\mZ)} &\lesssim \frac{\mu_1r}{n}\lb\sqrt{\frac{n\log(s n)}{m}}\gfronorm{\mZ} + \frac{n\log(s n)}{m}\ginfnorm{\mZ}\rb\\
&\quad + \frac{\mu_1r}{n}\lb\sqrt{\frac{n k_0s\mu_0\log(s n)}{m}}\gfronorm{\mZ} + \frac{n s\mu_0\log(s n)}{m}\ginfnorm{\mZ}\rb\\
&\lesssim \frac{\mu_1r}{n}\lb\sqrt{\frac{4n k_0s\mu_0\log(s n)}{m}}\gfronorm{\mZ} + \frac{2n s\mu_0\log(s n)}{m}\ginfnorm{\mZ}\rb
\end{align*}
holds with \msh{}{probability exceeding $1-ns^{-c_2}$.}

\subsection{Proof of Lemma \ref{lem: key lemma5}}

According to \eqref{eq: incoherence}, a simple algebra yields that 
\begin{align*}
\max_{0\leq i\leq n_1-1}\fronorm{\mU_i\mVT}^2 \leq \max_{0\leq i\leq n_1-1}\fronorm{\mU_i}^2 \leq \frac{\mu_1r}{n}.
\end{align*}
Then the application of Corollary \ref{cor: useful cor2} implies that 
\begin{align*}
\gfronorm{\mU\mVT}^2 \lesssim \frac{\mu_1r\log(s n)}{n}.
\end{align*}
The upper bound of $\ginfnorm{\mU\mVT}$ can be established as follows. 
Note that
\begin{align*}
\ginfnorm{\mU\mVT} = \max_{0\leq i \leq n-1} \frac{\vecnorm{\calGT(\mU\mVT)\ve_i}{2}}{\sqrt{w_i}}.
\end{align*}
For any fixed $i\in[n]$, we have
\begin{align*}
\frac{\vecnorm{\calGT(\mU\mVT)\ve_i}{2}}{\sqrt{w_i}} &= \frac{1}{\sqrt{w_i}}\sup_{\msh{}{\vecnorm{\msh{}{\bbeta}}{2}=1}}\lab\la\calGT(\mU\mVT)\ve_i,\msh{}{\bbeta}\ra\rab\\
&= \frac{1}{\sqrt{w_i}}\sup_{\msh{}{\vecnorm{\msh{}{\bbeta}}{2}=1}}\lab\la\mU\mVT, \calG(\msh{}{\bbeta}\ve_i^\tran)\ra\rab\\
&= \frac{1}{\sqrt{w_i}}\sup_{\msh{}{\vecnorm{\msh{}{\bbeta}}{2}=1}}\lab\la\mU\mVT,\mG_i\otimes\msh{}{\bbeta}\ra\rab\\
&= \frac{1}{w_i}\sup_{\msh{}{\vecnorm{\msh{}{\bbeta}}{2}=1}}\lab\la\mU\mVT,\lb\sum_{\msh{}{\substack{j+k=i\\0\leq j\leq n_1-1\\0\leq k\leq n_2-1}}}\ve_j\ve_k^\tran\rb\otimes\msh{}{\bbeta}\ra\rab\\
&= \frac{1}{w_i}\sup_{\msh{}{\vecnorm{\msh{}{\bbeta}}{2}=1}}\lab\la\mU\mVT,\sum_{\msh{}{\substack{j+k=i\\0\leq j\leq n_1-1\\0\leq k\leq n_2-1}}}(\ve_j\otimes\msh{}{\bbeta})\ve_k^\tran\ra\rab\\
&= \frac{1}{w_i}\sup_{\msh{}{\vecnorm{\msh{}{\bbeta}}{2}=1}}\lab\sum_{\msh{}{\substack{j+k=i\\0\leq j\leq n_1-1\\0\leq k\leq n_2-1}}}\la(\ve_j\otimes\msh{}{\bbeta})^*\mU,\ve_k^\tran\mV\ra\rab\\
&\leq \frac{1}{w_i}\sup_{\msh{}{\vecnorm{\msh{}{\bbeta}}{2}=1}}\sum_{\msh{}{\substack{j+k=i\\0\leq j\leq n_1-1\\0\leq k\leq n_2-1}}}\vecnorm{(\ve_j\otimes\msh{}{\bbeta})^*\mU}{2}\vecnorm{\ve_k^\tran\mV}{2}\\
&\leq \sup_{\msh{}{\vecnorm{\msh{}{\bbeta}}{2}=1}}\sqrt{\frac{1}{w_i}\sum_{\msh{}{\substack{j+k=i\\0\leq j\leq n_1-1\\0\leq k\leq n_2-1}}}\vecnorm{(\ve_j\otimes\msh{}{\bbeta})^*\mU}{2}^2}\sqrt{\frac{1}{w_i}\sum_{\msh{}{\substack{j+k=i\\0\leq j\leq n_1-1\\0\leq k\leq n_2-1}}}\vecnorm{\ve_k^\tran\mV}{2}^2}\\
&= \sup_{\msh{}{\vecnorm{\msh{}{\bbeta}}{2}=1}}\sqrt{\frac{1}{w_i}\sum_{\msh{}{\substack{j+k=i\\0\leq j\leq n_1-1\\0\leq k\leq n_2-1}}}\vecnorm{\msh{}{\bbeta}^*\mU_j}{2}^2}\sqrt{\frac{1}{w_i}\sum_{\msh{}{\substack{j+k=i\\0\leq j\leq n_1-1\\0\leq k\leq n_2-1}}}\vecnorm{\ve_k^\tran\mV}{2}^2}\\
&\leq \sqrt{\frac{1}{w_i}\sum_{\msh{}{\substack{j+k=i\\0\leq j\leq n_1-1\\0\leq k\leq n_2-1}}}\fronorm{\mU_j}^2}\sqrt{\frac{1}{w_i}\sum_{\msh{}{\substack{j+k=i\\0\leq j\leq n_1-1\\0\leq k\leq n_2-1}}}\vecnorm{\ve_k^\tran\mV}{2}^2}\\
&\leq \frac{\mu_1r}{n},
\end{align*}
\msh{}{where the fourth line is due to the definition of $\mG_i$ in \eqref{eq: Hankel basis} and} the last line follows from \eqref{incoherence condition}. Therefore, $\ginfnorm{\mU\mVT} \leq \frac{\mu_1r}{n}$.


\section{Auxiliary Results}
\label{aux}
In this section, we present some necessary results which have been used in the previous proofs. The following lemma is used in the proof of Theorem \ref{thr:optimality}.
\begin{lemma}\label{lem:upper bound}
	Suppose $\opnorm{\calA\calA^\ast}\geq 1$ and $\opnorm{\calP_{T}\calG\calAT \calA\calGT\calP_{T} -\calP_{T}\calG\calGT\calP_{T} } \leq\frac{1}{2}$. For any $\msh{}{\mM} \in \C^{s n_1\times n_2}$ which obeys
	\begin{align*}
	\calA\calGT(\msh{}{\mM})=0\quad\mbox{and}\quad(\calI-\calG\calGT)(\msh{}{\mM})=\bzero,
	\end{align*}
	we  have
	\begin{align*}
	\fronorm{\calP_T(\msh{}{\mM})}\leq 4s\mu_0\fronorm{\calP_{T^\perp}(\msh{}{\mM})}.
	\end{align*}
\end{lemma}

\begin{proof}
	\kw{}{It follows \eqref{eq: w cond 1} and \eqref{eq: w cond 2} that }
	\begin{align*}
	0 &=\fronorm{(\calG\calAT\calA\calGT+(\calI-\calG\calGT))(\msh{}{\mM})}\\
	&\geq \fronorm{(\calG\calAT\calA\calGT+(\calI-\calG\calGT))\calP_T(\msh{}{\mM})}-\fronorm{(\calG\calAT\calA\calGT+(\calI-\calG\calGT))\calP_{T^\perp}(\msh{}{\mM})}.
	\end{align*}
	For the first term, 
	\begin{align*}
	\fronorm{(\calG\calAT\calA\calGT+(\calI-\calG\calGT))\calP_T(\msh{}{\mM})}^2 &= \fronorm{\calG\calAT\calA\calGT\calP_T(\msh{}{\mM})}^2+\fronorm{(\calI-\calG\calGT)\calP_T(\msh{}{\mM})}^2\\
	&= \la\calG\calAT\calA\calGT\calP_T(\msh{}{\mM}),\calG\calAT\calA\calGT\calP_T(\msh{}{\mM})\ra+\la\calP_T(\msh{}{\mM}),(\calI-\calG\calGT)\calP_T(\msh{}{\mM})\ra\\
	&= \la\calA\calGT\calP_T(\msh{}{\mM}),(\calA\calAT)\calA\calGT\calP_T(\msh{}{\mM})\ra+\la\calP_T(\msh{}{\mM}),(\calI-\calG\calGT)\calP_T(\msh{}{\mM})\ra\\
	&\geq \la\calP_T(\msh{}{\mM}),\calG\calAT\calA\calGT\calP_T(\msh{}{\mM})\ra+\la\calP_T(\msh{}{\mM}),(\calI-\calG\calGT)\calP_T(\msh{}{\mM})\ra\\
	&=\msh{}{\fronorm{\calP_T(\msh{}{\mM})}^2}+\la\calP_T(\msh{}{\mM}),\calP_T(\calG\calAT\calA\calGT-\calG\calGT)\calP_T(\msh{}{\mM})\ra\\
	&\geq \msh{}{\fronorm{\calP_T(\msh{}{\mM})}^2}-\opnorm{\calP_T(\calG\calAT\calA\calGT-\calG\calGT)\calP_T}\cdot\fronorm{\calP_T(\msh{}{\mM})}^2\\
	&\geq \frac{1}{2}\msh{}{\fronorm{\calP_T(\msh{}{\mM})}^2}.
	\end{align*}
	where the fourth step is due to \msh{}{\eqref{ineq: lower bound of ATA} in Lemma~\ref{lem: prop of A}.}
	
	For the second term,
	\begin{align*}
	\fronorm{(\calG\calAT\calA\calGT+(\calI-\calG\calGT))\calP_{T^\perp}(\msh{}{\mM})} &\leq \fronorm{(\calG\calAT\calA\calGT)\calP_{T^\perp}(\msh{}{\mM})}+\fronorm{(\calI-\calG\calGT)\calP_{T^\perp}(\msh{}{\mM})}\\
	&\leq \opnorm{\calG}\cdot\opnorm{\calAT\calA}\cdot\opnorm{\calGT}\cdot\fronorm{\calP_{T^\perp}(\msh{}{\mM})}+\opnorm{\calI-\calG\calGT}\cdot\fronorm{\calP_{T^\perp}(\msh{}{\mM})}\\
	&\leq (1+s\mu_0)\fronorm{\calP_{T^\perp}(\msh{}{\mM})}\\
	&\leq 2s\mu_0\fronorm{\calP_{T^\perp}(\msh{}{\mM})}
	\end{align*}
	where the third line is due to $\opnorm{\calG}=1$, $\opnorm{\calGT}\leq 1$ and \eqref{eq:more of ATA} in Lemma~\ref{lem: prop of A}. 
	
	Combining these two terms together completes the proof.
\end{proof}

The following lemmas play an important role in the proofs of Lemmas \ref{lem: partition} to \ref{lem: key lemma5}.
\begin{lemma}
	\label{lem: supp1}
	\msh{}{Recall that $\mU$ and $\mV$ obey (III.4).} For any fixed \msh{}{$\vz\in\C^s$}, 
	there holds
	\begin{align}
	&\max_{0\leq i\leq n-1}\fronorm{\mU^{\ast}\calG(\msh{}{\vz}\ve_{i}^\tran)}^2\leq\vecnorm{\msh{}{\vz}}{2}^2\cdot\frac{\mu_1r}{n},\numberthis\label{ineq: first}\\
	&\max_{0\leq i\leq n-1}\fronorm{\calG(\msh{}{\vz}\ve_{i}^\tran)\mV}^2\leq\vecnorm{\msh{}{\vz}}{2}^2\cdot\frac{\mu_1r}{n},\numberthis\label{ineq: second}\\
	&\max_{0\leq i\leq n-1}\fronorm{\calP_T\calG(\msh{}{\vz}\ve_{i}^\tran)}^2\leq 2\vecnorm{\msh{}{\vz}}{2}^2\cdot\frac{\mu_1r}{n}.\numberthis\label{ineq: third}
	\end{align}
\end{lemma}
\begin{proof}
	To show \eqref{ineq: first}, note that for any $0\leq i \leq n-1$,  
	\begin{align*}
	\calG(\msh{}{\vz}\ve_i^\tran) &= \mG_i\otimes\msh{}{\vz}\\
	&= \lb\sum_{\msh{}{\substack{j+k=i\\0\leq j\leq n_1-1\\0\leq k\leq n_2-1}}}\frac{1}{\sqrt{w_i}}\ve_j\ve_k^\tran\rb\otimes\msh{}{\vz}\\
	&= \sum_{\msh{}{\substack{j+k=i\\0\leq j\leq n_1-1\\0\leq k\leq n_2-1}}}\frac{1}{\sqrt{w_i}}\lb\ve_j\otimes\msh{}{\vz}\rb\ve_k^\tran,
	\end{align*}
	\msh{}{where the second equality is due to the definition of $\mG_i$ in \eqref{eq: Hankel basis}.} 
	It follows that
	\begin{align*}
	\fronorm{\mU^{\ast}\calG(\msh{}{\vz}\ve_i^\tran)}^2 &=\la\mU^{\ast}\calG(\msh{}{\vz}\ve_i^\tran),\mU^{\ast}\calG(\msh{}{\vz}\ve_i^\tran)\ra\\
	&=\frac{1}{w_i}\la\sum_{\msh{}{\substack{j+k=i\\0\leq j\leq n_1-1\\0\leq k\leq n_2-1}}}\mU^{\ast}\lb\ve_j\otimes\msh{}{\vz}\rb\ve_k^\tran,\sum_{\msh{}{\substack{p+q=i\\0\leq p\leq n_1-1\\0\leq q\leq n_2-1}}}\mU^{\ast}\lb\ve_p\otimes\msh{}{\vz}\rb\ve_q^\tran\ra\\
	&=\frac{1}{w_i}\sum_{\msh{}{\substack{j+k=i\\0\leq j\leq n_1-1\\0\leq k\leq n_2-1}}}\la\mU^{\ast}\lb\ve_j\otimes\msh{}{\vz}\rb,\mU^{\ast}\lb\ve_j\otimes\msh{}{\vz}\rb\ra\\
	&=\frac{1}{w_i}\sum_{\msh{}{\substack{j+k=i\\0\leq j\leq n_1-1\\0\leq k\leq n_2-1}}}\vecnorm{\mU^{\ast}\lb\ve_j\otimes\msh{}{\vz}\rb}{2}^2\\
	&=\frac{1}{w_i}\sum_{\msh{}{\substack{j+k=i\\0\leq j\leq n_1-1\\0\leq k\leq n_2-1}}}\vecnorm{\mU_j^{\ast}\msh{}{\vz}}{2}^2\\
	&\leq \frac{1}{w_i}\sum_{\msh{}{\substack{j+k=i\\0\leq j\leq n_1-1\\0\leq k\leq n_2-1}}}\vecnorm{\msh{}{\vz}}{2}^2\cdot\fronorm{\mU_j}^2\\
	&\leq\vecnorm{\msh{}{\vz}}{2}^2\cdot\frac{\mu_1r}{n},
	\end{align*}
	where the last step follows from \eqref{incoherence condition}.
	
	As for \eqref{ineq: second}, note that
	\begin{align*}
	\fronorm{\calG(\msh{}{\vz}\ve_i^\tran)\mV}^2 &= \la\calG(\msh{}{\vz}\ve_i^\tran)\mV,\calG(\msh{}{\vz}\ve_i^\tran)\mV\ra\\
	&= \frac{1}{w_i}\la\sum_{\msh{}{\substack{j+k=i\\0\leq j\leq n_1-1\\0\leq k\leq n_2-1}}}(\ve_j\otimes\msh{}{\vz})\ve_k^\tran\mV,\sum_{\msh{}{\substack{p+q=i\\0\leq p\leq n_1-1\\0\leq q\leq n_2-1}}}(\ve_p\otimes\msh{}{\vz})\ve_q^\tran\mV\ra\\
	&= \frac{1}{w_i}\sum_{\msh{}{\substack{j+k=i\\0\leq j\leq n_1-1\\0\leq k\leq n_2-1}}}\sum_{\msh{}{\substack{p+q=i\\0\leq p\leq n_1-1\\0\leq q\leq n_2-1}}}\la(\ve_j\otimes\msh{}{\vz})\ve_k^\tran\mV,(\ve_p\otimes\msh{}{\vz})\ve_q^\tran\mV\ra\\
	&= \frac{1}{w_i}\sum_{\msh{}{\substack{j+k=i\\0\leq j\leq n_1-1\\0\leq k\leq n_2-1}}}\sum_{\msh{}{\substack{p+q=i\\0\leq p\leq n_1-1\\0\leq q\leq n_2-1}}}\la\lb\ve_p^\tran\otimes\msh{}{\vz}^*\rb\lb\ve_j\otimes\msh{}{\vz}\rb\ve_k^\tran\mV,\ve_q^\tran\mV\ra\\
	&= \frac{1}{w_i}\sum_{\msh{}{\substack{j+k=i\\0\leq j\leq n_1-1\\0\leq k\leq n_2-1}}}\sum_{\msh{}{\substack{p+q=i\\0\leq p\leq n_1-1\\0\leq q\leq n_2-1}}}\la(\ve_p^\tran\ve_j)\otimes(\msh{}{\vz}^*\msh{}{\vz})\ve_k^\tran\mV,\ve_q^\tran\mV\ra\\
	&=\frac{1}{w_i}\sum_{\msh{}{\substack{j+k=i\\0\leq j\leq n_1-1\\0\leq k\leq n_2-1}}}\la\msh{}{\vz}^*\msh{}{\vz}\ve_k^\tran\mV,\ve_k^\tran\mV\ra\\
	&= \frac{\vecnorm{\msh{}{\vz}}{2}^2}{w_i}\sum_{\msh{}{\substack{j+k=i\\0\leq j\leq n_1-1\\0\leq k\leq n_2-1}}}\la\ve_k^\tran\mV,\ve_k^\tran\mV\ra\\
	&\leq \vecnorm{\msh{}{\vz}}{2}^2\frac{\mu_1r}{n},
	\end{align*}
	where the last step is also due to \eqref{incoherence condition}. 
	
	For the inequality \eqref{ineq: third}, \msh{}{by the definition of $\calP_T$ in \eqref{eq: PT},} we have
	\begin{align*}
	\fronorm{\calP_T\calG\lb\msh{}{\vz}\ve_i^\tran\rb}^2 &= \la\calP_T\calG\lb\msh{}{\vz}\ve_i^\tran\rb,\calP_T\calG\lb\msh{}{\vz}\ve_i^\tran\rb\ra\\
	&= \la\calP_T\calG\lb\msh{}{\vz}\ve_i^\tran\rb,\calG(\msh{}{\vz}\ve_i^\tran)\ra\\
	&= \la\mU\mU^*\calG\lb\msh{}{\vz}\ve_i^\tran\rb+\calG\lb\msh{}{\vz}\ve_i^\tran\rb\mV\mV^*-\mU\mU^*\calG\lb\msh{}{\vz}\ve_i^\tran\rb\mV\mV^*,\calG\lb\msh{}{\vz}\ve_i^\tran\rb\ra\\
	&= \fronorm{\mU^*\calG(\msh{}{\vz}\ve_i^\tran)}^2+\fronorm{\calG\lb\msh{}{\vz}\ve_i^\tran\rb\mV}^2 -\fronorm{\mU^*\calG\lb\msh{}{\vz}\ve_i^\tran\rb\mV}^2\\
	&\leq \fronorm{\mU^*\calG(\msh{}{\vz}\ve_i^\tran)}^2+\fronorm{\calG\lb\msh{}{\vz}\ve_i^\tran\rb\mV}^2\\
	&\leq 2\vecnorm{\msh{}{\vz}}{2}^2\frac{\mu_1 r}{n},
	\end{align*}
	which completes the proof.
\end{proof}

After \kw{}{replacing} $\msh{}{\vz}$ with $\vb_i$ in Lemma~\ref{lem: supp1}, we obtain the following corollary based on the incoherence property \eqref{eq: incoherence b} of $\vb_i$, where $\vb_i$ is the $i$th column of $\mB^*$.
\begin{corollary}
	\label{cor: useful cor}
	Under the condition \eqref{eq: incoherence}, there holds
	\begin{align}
	&\max_{0\leq i\leq n-1}\fronorm{\mU^{\ast}\calG(\vb_i\ve_{i}^\tran)}^2\leq\frac{\mu_0\mu_1sr}{n},\\
	&\max_{0\leq i\leq n-1}\fronorm{\calG(\vb_i\ve_{i}^\tran)\mV}^2\leq\frac{\mu_0\mu_1sr}{n},\\
	\label{ineq: supp13}
	&\max_{0\leq i\leq n-1}\fronorm{\calP_T\calG(\vb_i\ve_{i}^\tran)}^2\leq \frac{2\mu_0\mu_1sr}{n}.
	\end{align}
\end{corollary}

\begin{lemma}
	\label{lem: supp2}
	Under the condition \eqref{eq: incoherence}, for any fixed matrix $\mW\in\C^{sn_1\times n_2}$,
	\begin{align}
	\vecnorm{\calGT\calP_T(\mW)\ve_i}{2}\leq\fronorm{\mW}\cdot\sqrt{\frac{2\mu_1r}{n}}.
	\end{align}
\end{lemma}
\begin{proof}
	\kw{}{The result follows from a direct calculation: }
	\begin{align*}
	\vecnorm{\calGT\calP_T(\mW)\ve_i}{2} &= \sup_{\msh{}{\vecnorm{\bbeta}{2}=1}}\lab\la\calGT\calP_T(\mW)\ve_i,\msh{}{\bbeta}\ra\rab\\
	&=\sup_{\msh{}{\vecnorm{\bbeta}{2}=1}}\lab\la\calGT\calP_T(\mW),\msh{}{\bbeta}\ve_i^\tran\ra\rab\\
	&= \sup_{\msh{}{\vecnorm{\bbeta}{2}=1}}\lab\la\mW,\calP_T\calG(\msh{}{\bbeta}\ve_i^\tran)\ra\rab\\
	&\leq \fronorm{\mW}\cdot\sup_{\msh{}{\vecnorm{\bbeta}{2}=1}}\fronorm{\calP_T\calG(\msh{}{\bbeta}\ve_i^\tran)}\\
	&\leq \fronorm{\mW}\cdot\sqrt{\frac{2\mu_1r}{n}},
	\end{align*}
	where the last line follows from \eqref{ineq: third} in Lemma \ref{lem: supp1}.
\end{proof}

By \kw{}{combining Lemmas \ref{lem: supp1} and \ref{lem: supp2}}, the following corollary can be established, which is used in the proof of \eqref{ineq: part_prop1}.
\begin{corollary}
	\label{cor: useful cor1}
	For any fixed matrix $\mW\in\C^{s n_1\times n_2}$, under the condition \eqref{eq: incoherence}, there holds 
	\begin{align*}
	\max_{0\leq i\leq n-1}\fronorm{\calP_T\calG\lb\calGT\calP_T(\mW)\ve_i\ve_i^\tran\rb}^2 \leq \fronorm{\mW}^2\cdot\lb\frac{2\mu_1r}{n}\rb^2,
	\end{align*}
\end{corollary}
\begin{proof}
	Applying Lemma \ref{lem: supp1} yields that
	\begin{align*}
	\max_{0\leq i\leq n-1}\fronorm{\calP_T\calG\lb\calGT\calP_T(\mW)\ve_i\ve_i^\tran\rb}^2 
	&\leq \vecnorm{\calGT\calP_T(\mW)\ve_i}{2}^2\cdot\frac{2\mu_1r}{n}\\
	&\leq \fronorm{\mW}^2\cdot\lb\frac{2\mu_1r}{n}\rb^2,
	\end{align*}
	where the last line is due to Lemma \ref{lem: supp2}.
\end{proof}

\begin{lemma}
	\label{lem: supp3}
	For any two fixed vectors $\msh{}{\bbeta,\bgamma}\in\C^s$, 
	\begin{align*}
	\sqrt{\frac{w_i}{w_j}}\lab\la\calP_T\calG(\msh{}{\bbeta}\ve_{i}^\tran), \calG(\msh{}{\bgamma}\ve_j^\tran)\ra\rab \leq \frac{3\mu_1r}{n}\cdot\vecnorm{\msh{}{\bbeta}}{2}\vecnorm{\msh{}{\bgamma}}{2}
	\end{align*}
	\kw{}{holds} for any $(i,j)\in [n]\times[n]$.
\end{lemma}
\begin{proof}
	Recall that 
	\begin{align*}
	\calG(\msh{}{\bbeta}\ve_i^\tran)=\mG_i\otimes\msh{}{\bbeta}=\lb\frac{1}{\sqrt{w_i}}\sum_{\substack{k+t=i\\0\leq k\leq n_1-1\\0\leq t\leq n_2-1}}\ve_k\ve_t^\tran\rb\otimes\msh{}{\bbeta}\quad\mbox{and}\quad
	\calG(\msh{}{\bgamma}\ve_j^\tran)=\mG_j\otimes\msh{}{\bgamma}=\lb\frac{1}{\sqrt{w_j}}\sum_{\substack{p+q=j\\0\leq p\leq n_1-1\\0\leq q\leq n_2-1}}\ve_p\ve_q^\tran\rb\otimes\msh{}{\bgamma}.
	\end{align*}
	By the definition of $\calP_T$ \msh{}{in \eqref{eq: PT}}, we have
	\begin{align*}
	\sqrt{\frac{w_i}{w_j}}\lab\la\calP_T\calG(\msh{}{\bbeta}\ve_i^\tran),\calG(\msh{}{\bgamma}\ve_j^\tran)\ra\rab \leq &\sqrt{\frac{w_i}{w_j}}\lab\la\mU\mU^*\lb\mG_i\otimes\msh{}{\bbeta}\rb,\mG_j\otimes\msh{}{\bgamma}\ra\rab + \sqrt{\frac{w_i}{w_j}}\lab\la\lb\mG_i\otimes\msh{}{\bbeta}\rb\mV\mVT,\mG_j\otimes\msh{}{\bgamma}\ra\rab\\ 
	&+ \sqrt{\frac{w_i}{w_j}}\lab\la\mU\mU^*\lb\mG_i\otimes\msh{}{\bbeta}\rb\mV\mVT,\mG_j\otimes\msh{}{\bgamma}\ra\rab.
	\end{align*}
	
	It suffices to bound each of the three terms separately. For the first term, we have
	\begin{align*}
	\sqrt{\frac{w_i}{w_j}}\lab\la\mU\mU^*\lb\mG_i\otimes\msh{}{\bbeta}\rb,\mG_j\otimes\msh{}{\bgamma}\ra\rab &=\sqrt{\frac{w_i}{w_j}}\lab\la\mU\mU^*\lb\lb\sum_{\msh{}{\substack{k+t=i\\0\leq k\leq n_1-1\\0\leq t\leq n_2-1}}}\frac{1}{\sqrt{w_i}}\ve_k\ve_t^\tran\rb\otimes\msh{}{\bbeta}\rb,\lb\sum_{\msh{}{\substack{p+q=j\\0\leq p\leq n_1-1\\0\leq q\leq n_2-1}}}\frac{1}{\sqrt{w_j}}\ve_p\ve_q^\tran\rb\otimes\msh{}{\bgamma}\ra\rab\\
	&= \sqrt{\frac{w_i}{w_j}}\lab\frac{1}{\sqrt{w_i w_j}}\la\sum_{\msh{}{\substack{k+t=i\\0\leq k\leq n_1-1\\0\leq t\leq n_2-1}}}\mU\mU^*(\ve_k\otimes\msh{}{\bbeta})\ve_t^\tran,\sum_{\msh{}{\substack{p+q=j\\0\leq p\leq n_1-1\\0\leq q\leq n_2-1}}}(\ve_p\otimes\msh{}{\bgamma})\ve_q^\tran\ra\rab\\
	&= \frac{1}{w_j}\lab\sum_{\msh{}{\substack{k+t=i\\0\leq k\leq n_1-1\\0\leq t\leq n_2-1}}}\sum_{\msh{}{\substack{p+q=j\\0\leq p\leq n_1-1\\0\leq q\leq n_2-1}}}\la\mU\mU^*(\ve_k\otimes\msh{}{\bbeta})\ve_t^\tran,(\ve_p\otimes\msh{}{\bgamma})\ve_q^\tran\ra\rab\\
	&= \frac{1}{w_j}\lab\sum_{\msh{}{\substack{p+q=j,q\leq i\\0\leq p\leq n_1-1\\0\leq q\leq n_2-1}}}\la\mU^*(\ve_{i-q}\otimes\msh{}{\bbeta}),\mU^*(\ve_p\otimes\msh{}{\bgamma})\ra\rab\\
	&\leq \frac{1}{w_j}\sum_{\msh{}{\substack{p+q=j,q\leq i\\0\leq p\leq n_1-1\\0\leq q\leq n_2-1}}}\vecnorm{\mU^*(\ve_{i-q}\otimes\msh{}{\bbeta})}{2}\cdot\vecnorm{\mU^*(\ve_p\otimes\msh{}{\bgamma})}{2}\\
	&= \frac{1}{w_j}\sum_{\msh{}{\substack{p+q=j,q\leq i\\0\leq p\leq n_1-1\\0\leq q\leq n_2-1}}}\vecnorm{\mU_{i-q}^*\msh{}{\bbeta}}{2}\cdot\vecnorm{\mU_p^*\msh{}{\bgamma}}{2}\\
	&\leq \frac{1}{w_j} \sum_{\msh{}{\substack{p+q=j,q\leq i\\0\leq p\leq n_1-1\\0\leq q\leq n_2-1}}}  \fronorm{\mU_{i-q}}\cdot\fronorm{\mU_p}\cdot\vecnorm{\msh{}{\bbeta}}{2}\cdot \vecnorm{\msh{}{\bgamma}}{2}\\
	&\leq \sqrt{\frac{1}{w_j} \sum_{\msh{}{\substack{p+q=j,q\leq i\\0\leq p\leq n_1-1\\0\leq q\leq n_2-1}}}\fronorm{\mU_{i-q}}^2}\cdot \sqrt{ \frac{1}{w_j} \sum_{\msh{}{\substack{p+q=j,q\leq i\\0\leq p\leq n_1-1\\0\leq q\leq n_2-1}}}\fronorm{\mU_p}^2}\cdot\vecnorm{\msh{}{\bbeta}}{2}\cdot\vecnorm{\msh{}{\bgamma}}{2}\\
	&\leq \frac{\mu_1r}{n}\cdot\vecnorm{\msh{}{\bbeta}}{2}\cdot\vecnorm{\msh{}{\bgamma}}{2}.
	\end{align*}
	The second term can be bounded in a similar way. For the last term, we have
	\begin{align*}
	\sqrt{\frac{w_i}{w_j}}\lab\la\mU\mU^*\lb\mG_i\otimes\msh{}{\bbeta}\rb\mV\mVT,\mG_j\otimes\msh{}{\bgamma}\ra\rab &= \sqrt{\frac{w_i}{w_j}}\lab\la\mU\mU^*\lb\mG_i\otimes\msh{}{\bbeta}\rb,(\mG_j\otimes\msh{}{\bgamma})\mV\mVT\ra\rab\\
	&= \frac{1}{w_j} \lab\sum_{\msh{}{\substack{k+t=i\\0\leq k\leq n_1-1\\0\leq t\leq n_2-1}}}\sum_{\msh{}{\substack{p+q=j\\0\leq p\leq n_1-1\\0\leq q\leq n_2-1}}}\la\mU\mU^*\lb\ve_k\otimes\msh{}{\bbeta}\rb\ve_t^\tran,\lb\ve_p\otimes\msh{}{\gamma}\rb\ve_q^\tran\mV\mVT\ra\rab \\
	&= \frac{1}{w_j}\lab\sum_{\msh{}{\substack{k+t=i\\0\leq k\leq n_1-1\\0\leq t\leq n_2-1}}}\sum_{\msh{}{\substack{p+q=j\\0\leq p\leq n_1-1\\0\leq q\leq n_2-1}}}\la\lb\ve_p^\tran\otimes\msh{}{\bgamma}^*\rb\mU\mU^*\lb\ve_k\otimes\msh{}{\bbeta}\rb,\ve_q^\tran\mV\mVT\ve_t\ra\rab\\
	&= \frac{1}{w_j}\lab\sum_{\msh{}{\substack{k+t=i\\0\leq k\leq n_1-1\\0\leq t\leq n_2-1}}}\sum_{\msh{}{\substack{p+q=j\\0\leq p\leq n_1-1\\0\leq q\leq n_2-1}}}\la\lb\mU_p^*\msh{}{\bgamma}\rb^*\lb\mU_k^*\msh{}{\bbeta}\rb, \ve_q^\tran\mV\mVT\ve_t\ra\rab\\
	&\leq \frac{1}{w_j} \sum_{\msh{}{\substack{k+t=i\\0\leq k\leq n_1-1\\0\leq t\leq n_2-1}}}\sum_{\msh{}{\substack{p+q=j\\0\leq p\leq n_1-1\\0\leq q\leq n_2-1}}}\lab\msh{}{\bgamma}^*\mU_p\mU_k^*\msh{}{\bbeta}\rab\cdot \lab\ve_q^\tran\mV\mVT\ve_t\rab\\
	&\leq \sqrt{\frac{1}{w_j}\sum_{\msh{}{\substack{k+t=i\\0\leq k\leq n_1-1\\0\leq t\leq n_2-1}}} \sum_{\msh{}{\substack{p+q=j\\0\leq p\leq n_1-1\\0\leq q\leq n_2-1}}}  \lab\msh{}{\bgamma}^*\mU_p\mU_k^*\msh{}{\bbeta}\rab^2} \sqrt{\frac{1}{w_j}\sum_{\msh{}{\substack{k+t=i\\0\leq k\leq n_1-1\\0\leq t\leq n_2-1}}} \sum_{\msh{}{\substack{p+q=j\\0\leq p\leq n_1-1\\0\leq q\leq n_2-1}}}\lab\ve_q^\tran\mV\mV^*\ve_t\rab^2}\\
	&\leq \vecnorm{\msh{}{\bbeta}}{2}\sqrt{\frac{1}{w_j}\sum_{\msh{}{\substack{k+t=i\\0\leq k\leq n_1-1\\0\leq t\leq n_2-1}}} \sum_{\msh{}{\substack{p+q=j\\0\leq p\leq n_1-1\\0\leq q\leq n_2-1}}}\vecnorm{\msh{}{\bgamma}^*\mU_p\mU_k^*}{2}^2}\sqrt{\frac{1}{w_j}\sum_{\msh{}{\substack{k+t=i\\0\leq k\leq n_1-1\\0\leq t\leq n_2-1}}} \sum_{\msh{}{\substack{p+q=j\\0\leq p\leq n_1-1\\0\leq q\leq n_2-1}}}\lab\ve_q^\tran\mV\mV^*\ve_t\rab^2}\\
	&\leq \vecnorm{\msh{}{\bbeta}}{2}\sqrt{\frac{1}{w_j} \sum_{\msh{}{\substack{p+q=j\\0\leq p\leq n_1-1\\0\leq q\leq n_2-1}}}\sum_{k=0}^{n_1-1}\vecnorm{\msh{}{\bgamma}^*\mU_p\mU_k^*}{2}^2}\sqrt{\frac{1}{w_j}\sum_{\msh{}{\substack{p+q=j\\0\leq p\leq n_1-1\\0\leq q\leq n_2-1}}}\sum_{t=0}^{n_2-1}\lab\ve_q^\tran\mV\mV^*\ve_t\rab^2}\\
	&= \vecnorm{\msh{}{\bbeta}}{2}\sqrt{\frac{1}{w_j}\sum_{\msh{}{\substack{p+q=j\\0\leq p\leq n_1-1\\0\leq q\leq n_2-1}}}\vecnorm{\msh{}{\bgamma}^*\mU_p\mU^*}{2}^2}\sqrt{\frac{1}{w_j}\sum_{\msh{}{\substack{p+q=j\\0\leq p\leq n_1-1\\0\leq q\leq n_2-1}}}\vecnorm{\ve_q^\tran\mV\mVT}{2}^2}\\
	&\leq \vecnorm{\msh{}{\bbeta}}{2}\vecnorm{\msh{}{\bgamma}}{2}\opnorm{\mU}\opnorm{\mV}\sqrt{\frac{1}{w_j}\sum_{\msh{}{\substack{p+q=j\\0\leq p\leq n_1-1\\0\leq q\leq n_2-1}}}\fronorm{\mU_p}^2}\cdot\sqrt{\frac{1}{w_j}\sum_{\msh{}{\substack{p+q=j\\0\leq p\leq n_1-1\\0\leq q\leq n_2-1}}}\vecnorm{\ve_q^\tran\mV}{2}^2}\\
	&\leq \frac{\mu_1r}{n}\vecnorm{\msh{}{\bbeta}}{2}\cdot\vecnorm{\msh{}{\bgamma}}{2},
	\end{align*}
	\msh{}{where the last step is due to \eqref{incoherence condition}.}
	
	Combining the three bounds together completes the proof.
\end{proof}

The following lemma is established in \cite{chen2014robust} and the proof will be omitted here. 
\begin{lemma}
	\label{lem: supp4}
	Suppose a matrix $\msh{}{\mF}\in\C^{n_1\times n_2}$ satisfies
	\begin{align}
	\max_{0\leq i \leq n_1-1}\vecnorm{\ve_i^\tran\msh{}{\mF}}{2}^2\leq B.
	\end{align}
	We have 
	\begin{align}
	\sum_{i=0}^{n-1}\frac{1}{w_i}\lab\la\msh{}{\mF},\mG_i\ra\rab^2\lesssim B\log(n).
	\end{align}
\end{lemma}

We will apply \kw{}{this lemma to upper bound} $\gfronorm{\mZ}$ for $\mZ\in\C^{s n_1\times n_2}$. Note that $\mZ$ can be written as
\begin{align*}
\mZ = \begin{bmatrix}
\vz_{0,0}&\cdots&\vz_{0,n_2-1}\\
\vdots&\ddots&\vdots\\
\vz_{n_1-1,0}&\cdots&\vz_{n_1-1,n_2-1}
\end{bmatrix},
\end{align*}
where $\vz_{i,j}\in\C^s$ is the $(i,j)$th block of $\mZ$. 

\begin{corollary}
	\label{cor: useful cor2}
	For any matrix $\mZ\in\C^{s n_1\times n_2}$ satisfying
	\begin{align}
	\label{ineq: block bound}
	\max_{0\leq i \leq n_1-1}\sum_{j=0}^{n_2-1}\vecnorm{\vz_{i,j}}{2}^2\leq B,
	\end{align}
	we have
	\begin{align}
	\gfronorm{\mZ}^2 \lesssim B\log(n).
	\end{align}
\end{corollary}
\begin{proof}
	Define the matrix
	\begin{align*}
	\msh{}{\widetilde\mZ}=\begin{bmatrix}
	\vecnorm{\vz_{0,0}}{2}&\cdots&\vecnorm{\vz_{0,n_2-1}}{2}\\
	\vdots&\ddots&\vdots\\
	\vecnorm{\vz_{n_1-1,0}}{2}&\cdots&\vecnorm{\vz_{n_1-1,n_2-1}}{2}
	\end{bmatrix}\in\R^{n_1\times n_2}.
	\end{align*}
	The definition of $\calGT$ implies that the $i$th column of $\calGT(\mZ)$ is given by
	\begin{align*}
	\calGT(\mZ)\ve_i = \frac{1}{\sqrt{w_i}}\sum_{\msh{}{\substack{j+k=i\\0\leq j\leq n_1-1\\0\leq k\leq n_2-1}}} \vz_{j,k},
	\end{align*}
	It follows that
	\begin{align*}
	\gfronorm{\mZ}^2 &= \sum_{i=0}^{n-1}\frac{1}{w_i}\vecnorm{\calGT(\mZ)\ve_i}{2}^2\\
	&= \sum_{i=0}^{n-1}\frac{1}{w_i}\vecnorm{\frac{1}{\sqrt{w_i}}\sum_{\msh{}{\substack{j+k=i\\0\leq j\leq n_1-1\\0\leq k\leq n_2-1}}} \vz_{j,k}}{2}^2\\
	&\leq \sum_{i=0}^{n-1}\frac{1}{w_i}\lb\frac{1}{\sqrt{w_i}}\sum_{\msh{}{\substack{j+k=i\\0\leq j\leq n_1-1\\0\leq k\leq n_2-1}}}\vecnorm{\vz_{j,k}}{2}\rb^2\\
	&= \sum_{i=0}^{n-1}\frac{1}{w_i}\lb\frac{1}{\sqrt{w_i}}\sum_{\msh{}{\substack{j+k=i\\0\leq j\leq n_1-1\\0\leq k\leq n_2-1}}}\la\msh{}{\widetilde\mZ}\ve_k,\ve_j\ra\rb^2\\
	&= \sum_{i=0}^{n-1}\frac{1}{w_i}\lb\frac{1}{\sqrt{w_i}}\sum_{\msh{}{\substack{j+k=i\\0\leq j\leq n_1-1\\0\leq k\leq n_2-1}}}\la\msh{}{\widetilde\mZ},\ve_j\ve_k^\tran\ra\rb^2\\
	&= \sum_{i=0}^{n-1}\frac{1}{w_i}\lb\la\msh{}{\widetilde\mZ},\mG_i\ra\rb^2,
	\end{align*}
	\msh{}{where the last line follows from the definition of $\mG_i$ in \eqref{eq: Hankel basis}.}
	
	\kw{}{Since the condition \eqref{ineq: block bound}  implies that 
		$
		\max\limits_{0\leq i \leq n_1-1}\vecnorm{\ve_i^\tran\msh{}{\widetilde\mZ}}{2}^2\leq B,
		$}
	applying Lemma \ref{lem: supp4} completes the proof.
\end{proof}

The following lemma can be established based on Corollary \ref{cor: useful cor2}. It has been used in the proofs of \eqref{ineq: part_prop3} and \eqref{ineq: part3}.
\begin{lemma}
	\label{lem: supp5}
	For any fixed $\msh{}{\vz}\in \C^s$,
	\begin{align*}
	\gfronorm{\calP_T\calG(\sqrt{w_i}\msh{}{\vz}\ve_i^\tran)}^2 \lesssim \vecnorm{\msh{}{\vz}}{2}^2\cdot\frac{\mu_1r\log(s n)}{n}.
	\end{align*}
\end{lemma}
\begin{proof}
	\msh{}{Recalling the definition of $\calP_T$ in \eqref{eq: PT}}, we have
	\begin{align*}
	\calP_T\calG(\sqrt{w_i}\msh{}{\vz}\ve_i^\tran) = \mU\mU^*\calG\lb\sqrt{w_i}\msh{}{\vz}\ve_i^\tran\rb + \calG\lb\sqrt{w_i}\msh{}{\vz}\ve_i^\tran\rb\mV\mVT - \mU\mU^*\calG\lb\sqrt{w_i}\msh{}{\vz}\ve_i^\tran\rb\mV\mVT.
	\end{align*}
	It suffices to bound the three terms separately. 
	For the first term, recall that $\mU\in\C^{s n_1\times r}$ can be rewritten as 
	\begin{align*}
	\mU=\begin{bmatrix}
	\mU_0\\
	\vdots\\
	\mU_{n_1-1}
	\end{bmatrix},
	\end{align*}
	where $\mU_\ell\in \C^{s\times r}$ is the $\ell$-th block. 
	Since
	\begin{align*}
	\fronorm{\mU_\ell\mU^*\calG\lb\sqrt{w_i}\msh{}{\vz}\ve_i^\tran\rb}^2 &= w_i\fronorm{\mU_\ell\mU^*\lb\mG_i\otimes\msh{}{\vz}\rb}^2\\
	&\leq w_i \fronorm{\mU_\ell}^2\cdot\opnorm{\mU}^2\cdot\opnorm{\mG_i\otimes\msh{}{\vz}}^2\\
	&\leq w_i\frac{\mu_1r}{n}\cdot\opnorm{\mG_i}^2\cdot\vecnorm{\msh{}{\vz}}{2}^2\\
	&\leq \frac{\mu_1r}{n}\cdot\vecnorm{\msh{}{\vz}}{2}^2,
	\end{align*}
	\msh{}{where the third line follows from \eqref{eq: incoherence}, then}
	the application of Corollary \ref{cor: useful cor2} yields that 
	\begin{align*}
	\gfronorm{\mU\mU^*\calG\lb\sqrt{w_i}\msh{}{\vz}\ve_i^\tran\rb}^2\lesssim \frac{\mu_1r\log(s n)}{n}\cdot\vecnorm{\msh{}{\vz}}{2}^2.
	\end{align*}
	The same bound can be obtained for $\calG\lb\sqrt{w_i}\msh{}{\vz}\ve_i^\tran\rb\mV\mVT$.
	
	For the last term, we have
	\begin{align*}
	\fronorm{\mU_\ell\mU^*\calG\lb\sqrt{w_i}\msh{}{\vz}\ve_i^\tran\rb\mV\mVT}^2&\leq w_i \fronorm{\mU_\ell}^2\cdot\opnorm{\mU}^2\cdot\opnorm{\calG(\msh{}{\vz}\ve_i^\tran)}^2\cdot\opnorm{\mV\mVT}^2\\
	&\leq w_i \frac{\mu_1r}{n}\opnorm{\mG_i}^2\cdot\vecnorm{\msh{}{\vz}}{2}^2\\
	&\leq \frac{\mu_1r}{n}\cdot\vecnorm{\msh{}{\vz}}{2}^2,
	\end{align*}
	\msh{}{where the second line is due to \eqref{eq: incoherence}}.
	Applying Corollary \ref{cor: useful cor2} again yields that
	\begin{align*}
	\gfronorm{\mU\mU^*\calG\lb\sqrt{w_i}\msh{}{\vz}\ve_i^\tran\rb\mV\mVT}^2\ \lesssim \frac{\mu_1r\log(s n)}{n}\cdot\vecnorm{\msh{}{\vz}}{2}^2.
	\end{align*}
	The proof is completed after combining the three bounds together.
\end{proof}

\section{Conclusion}
A convex approach called Vectorized Hankel Lift is proposed for blind super-resolution. It is based on the observation that the corresponding vectorized Hankel matrix is low rank if the Fourier  samples of the unknown PSFs lie in a low dimensional subspace. Theoretical guarantee has been established for Vectorized Hankel Lift, showing that exact  resolution can be achieved provided the number of samples is nearly optimal. We leave the robust analysis of the method to the future work. {\color{black}In particular, we would like  to see whether the technique that bridges convex and nonconvex programs in \cite{chen2020noisy} may yield  an optimal error bound for the blind supoer-resolution problem.}

For low rank matrix recovery and spectrally sparse signal recovery, many simple yet efficient nonconvex iterative algorithms have been developed and analysed based on inherent low rank structures of the problems \cite{wei2016guarantees,wei2016guarantees1,cai2015fast,cai2019fast,cai2018spectral}. Thus, it is also interesting  to develop nonconvex optimization methods for blind super-resolution based on the low rank structure  of the vectorized Hankel matrix. In fact, preliminary numerical results suggest that a variant of the gradient method in \cite{cai2018spectral} is also able to reconstruct the target matrix arsing in the blind super-resolution problem from a few number of the spectrum samples. A detailed discussion towards this line of research will be reported separately. 

{\color{black} For the single snapshot MUSIC  and the MMV MUSIC, the super-resolution effect has been studied in \cite{liao2016music, li2019conditioning, li2021stability}.  Since the spatial smoothing MUSIC is designed to improve the performance of the MMV MUSIC, it is also interesting to  investigate the super-resolution effect of this variant. The equivalence between it and MUSIC through  Vectorized Hankle Lift (i.e., Lemma~\ref{lem:MUSIC via VHM}) may provide a new perspective to approach  this problem. }
\label{conclusion}

\section*{Acknowledgments}
KW would like to thank Wenjing Liao for fruitful discussions on the subspace methods for line spectrum estimation, 
and would like to thank Zai Yang for pointing out that the MUSIC variant arising naturally from the vectorized
Hankel lift framework is indeed equivalent to the spatial smoothing technique proposed to improve the performance
of the MMV MUSIC.
\bibliographystyle{plain}
\bibliography{ref}

\end{document}